\documentclass[journal]{IEEEtran}
\normalsize
\usepackage{amsmath,amssymb}
\usepackage{xifthen}
\usepackage{mathtools}
\usepackage{enumerate}
\usepackage{microtype}
\usepackage{xspace}
\usepackage{bm}
\usepackage[T1]{fontenc}
\usepackage{fancyhdr}
\usepackage{lastpage}
\usepackage{bbm}

\newcommand{\mbs}[1]{\bm{#1}}
\newcommand{\vect}[1]{{\lowercase{\mbs{#1}}}}
\newcommand{\mat}[1]{{\uppercase{\mbs{#1}}}}

\newcommand{\T}{{\scriptscriptstyle\mathsf{T}}}

\newcommand{\cond}{\,\vert\,}
\renewcommand{\Re}[1][]{\ifthenelse{\isempty{#1}}{\operatorname{Re}}{\operatorname{Re}\left(#1\right)}}
\renewcommand{\Im}[1][]{\ifthenelse{\isempty{#1}}{\operatorname{Im}}{\operatorname{Im}\left(#1\right)}}

\newcommand{\cv}{\vect{c}}

\newcommand{\rv}{\vect{r}}

\newcommand{\yv}{\vect{y}}

\newcommand{\muv}{\vect{\mu}}

\newcommand{\Sigmam}{\pmb{\Sigma}}

\newcommand{\Ac}{{\mathcal A}}

\newcommand{\Cc}{{\mathcal C}}
\newcommand{\Dc}{{\mathcal D}}

\newcommand{\Hc}{{\mathcal H}}

\newcommand{\Nc}{{\mathcal N}}

\newcommand{\Tc}{{\mathcal T}}
\newcommand{\Uc}{{\mathcal U}}
\newcommand{\Wc}{{\mathcal W}}

\newcommand{\CC}{\mathbb{C}}

\newcommand{\Id}{\mat{\mathrm{I}}}

\newcommand{\CN}[1][]{\ifthenelse{\isempty{#1}}{\mathcal{N}_{\mathbb{C}}}{\mathcal{N}_{\mathbb{C}}\left(#1\right)}}

\renewcommand{\P}[1][]{\ifthenelse{\isempty{#1}}{\mathbb{P}}{\mathbb{P}\left[{#1}\right]}}
\newcommand{\E}[1][]{\ifthenelse{\isempty{#1}}{\mathbb{E}}{\mathbb{E}\left[#1\right]}}
\newcommand{\I}[1][]{\ifthenelse{\isempty{#1}}{\mathbb{I}}{\mathbb{I}\left\{#1\right\}}}
\renewcommand{\det}[1][]{\ifthenelse{\isempty{#1}}{\mathrm{det}}{\mathrm{det}\left(#1\right)}}
\newcommand{\trace}[1][]{\ifthenelse{\isempty{#1}}{{\rm tr}}{\mathrm{tr}\left(#1\right)}}
\newcommand{\rank}[1][]{\ifthenelse{\isempty{#1}}{\mathrm{rank}}{\mathrm{rank}\left(#1\right)}}
\newcommand{\diag}[1][]{\ifthenelse{\isempty{#1}}{\mathrm{diag}}{\mathrm{diag}\left(#1\right)}}
\newcommand{\Cov}[1][]{\ifthenelse{\isempty{#1}}{\mathsf{Cov}}{\mathsf{Cov}\left(#1\right)}}

\newcommand{\defeq}{\triangleq}

\newtheorem{remark}{Remark}
\newtheorem{definition}{Definition}
\newtheorem{theorem}{Theorem}
\newtheorem{corollary}{Corollary}
\newtheorem{lemma}{Lemma}

\newcounter{enumi_saved}
\usepackage{answers}
\Newassociation{solution}{Solution}{solutionfile}

\AtBeginDocument{\Opensolutionfile{solutionfile}[\jobname]}
\AtEndDocument{\Closesolutionfile{solutionfile}\clearpage
}

\IfFileExists{MinionPro.sty}{
}{
}

\usepackage[labelformat=simple]{subcaption}

\usepackage{array}
\usepackage{amsmath}
\usepackage[table,dvipsnames]{xcolor}
\usepackage{enumitem}
\usepackage{cite}
\usepackage[toc,acronym,nonumberlist,nopostdot]{glossaries}
\PassOptionsToPackage{hyphens}{url}\usepackage[hidelinks]{hyperref}

\makeindex
\makenoidxglossaries
\usepackage{multicol,multirow,pgf}
\usepackage{tikz,pgfplots}
\usetikzlibrary{shapes}
\usetikzlibrary{plotmarks}
\usetikzlibrary{spy}
\usepackage{pgfplots,pgf}
\pgfplotsset{minor grid style={dotted,gray!40}}
\pgfplotsset{major grid style={dashed,gray!40}}

\allowdisplaybreaks

\renewcommand{\rv}[1]{{\mathsf{#1}}}
\newcommand{\rvVec}[1]{{\bm{\mathsf{#1}}}}

\newcommand{\ind}[1]{{\mathbbm{1}{\left\{#1\right\}}}}

\renewcommand{\defeq}{=}

\renewcommand{\Id}{\mat{I}}

\newcommand{\EbNo}{{E_{\rm b}/N_0}}

\newcommand{\md}{\frac{|\revise{\Wc_{\rm MD}}|}{|\widetilde{\Wc}|}}
\newcommand{\fa}{\frac{|\revise{\Wc_{\rm FA}}|}{|\widehat{\Wc}|} }

\newacronym{MAC}{MAC}{multiple access channel}
\newacronym{RAC}{MAC}{random-access channel}
\newacronym{UMRA}{UMRA}{unsourced massive random-access}
\newacronym{SIMO}{SIMO}{single-input multiple-output}
\newacronym{SISO}{SISO}{single-input single-output}
\newacronym{iid}{i.i.d.}{independent and identically distributed}
\newacronym{ML}{ML}{maximum likelihood}
\newacronym{PEP}{PEP}{pair-wise error probability}
\newacronym{LLR}{LLR}{log-likelihood ratio}
\newacronym{SNR}{SNR}{signal-to-noise ratio}
\newacronym{RCB}{RCB}{random-coding bound}
\newacronym{MD}{MD}{misdetection}
\newacronym{FA}{FA}{false alarm}
\newacronym{PMF}{PMF}{probability mass function}
\newacronym{wlog}{w.l.o.g.}{without loss of generality}
\newacronym{wrt}{w.r.t.}{with respect to}
\newacronym{SA}{SA}{slotted {ALOHA}}
\newacronym{DoF}{DoF}{degrees of freedom}
\newacronym{rDoF}{rDoF}{real degrees of freedom}
\newacronym{ROC}{ROC}{receiver operating characteristic}
\newacronym{IoT}{IoT}{Internet of Things}
\newacronym{TIN}{TIN}{treating-interference-as-noise}
\newacronym{CDF}{CDF}{cummulative distribution function}
\newacronym{PDF}{PDF}{probability density function}
\newacronym{MMSE}{MMSE}{minimum mean squared error}
\newacronym{some name}{some name}{{\color{red} some name}}

\title{Unsourced Multiple Access \\ With Random User Activity} 

\author{\IEEEauthorblockN{Khac-Hoang Ngo, \emph{Member, IEEE}, Alejandro Lancho, \emph{Member, IEEE}, Giuseppe Durisi, \emph{Senior Member, IEEE}, \\ and Alexandre Graell i Amat, \emph{Senior Member, IEEE}} 
	\thanks{Khac-Hoang Ngo, Giuseppe Durisi, and Alexandre Graell i Amat are with the Department of Electrical Engineering, Chalmers University of Technology, 41296 Gothenburg, Sweden~(e-mails: {\tt \{ngok, durisi, alexandre.graell\}@chalmers.se}). Alejandro Lancho is with the Department of Electrical Engineering and Computer Science, Massachusetts Institute of Technology, Cambridge, MA 02139, USA 
		(e-mail: {\tt lancho@mit.edu}).}
	\thanks{This paper has been presented in part at the IEEE  International Symposium on Information Theory (ISIT), Melbourne, Victoria, Australia, July 2021~\cite{Ngo2021ISITmassive}.}
}

\newcommand{\revise}[1]{#1} 
\newcommand{\revisee}[1]{#1} 
\newcommand{\reviseee}[1]{#1} 

\begin{document}
	
	\maketitle
	\date{\today}
	\begin{abstract}
		To account for the massive uncoordinated random access scenario, which is relevant for the Internet of Things, Polyanskiy (2017) proposed a novel formulation of the multiple-access problem, commonly referred to as unsourced multiple access, where all users employ a common codebook and the receiver decodes up to a permutation of the messages. In this paper, we extend this seminal work to the case where the number of active users is random and unknown {\em a priori}. We define a random-access code accounting for both misdetection~(MD) and false alarm~(FA), and derive a random-coding achievability bound for the Gaussian multiple access channel. Our bound captures the fundamental trade-off between MD and FA probabilities. 
		It suggests that 
		the lack of knowledge of the number of active users entails a small penalty in energy efficiency when the target MD and FA probabilities are high. However, as the target MD and FA probabilities decrease, the energy efficiency penalty becomes more significant. For example, in a typical IoT scenario with framelength $19200$ complex channel uses and $25$--$300$ active users in average, the required  energy per bit to achieve both MD and FA probabilities below $10^{-1}$, predicted by our bound, is only $0.5$--$0.7$ dB higher than that predicted by the bound in Polyanskiy (2017) for a known number of active users. This gap increases to $3$--$4$~dB when the target MD probability and/or FA probability is below $10^{-3}$. Taking both MD and FA into account, we use our bound to benchmark the energy efficiency of slotted-ALOHA with multi-packet reception, of a decoder that simply treats interference as noise, and of some recently proposed unsourced multiple access schemes. Numerical results suggest that, when the target MD and FA probabilities are high, it is effective to estimate the number of active users, then treat this estimate as the true value, and use a coding scheme that performs well for the case of known number of active users. However, this approach becomes energy inefficient when the requirements on MD and FA probabilities are stringent. 
	\end{abstract}
	

	\section{Introduction} \label{sec:intro}
	The \gls{IoT} enables a variety of applications, such as autonomous driving, smart homes, and smart cities, by providing wireless access to a massive number of devices. A significant fraction of IoT devices is battery limited and transmits short packets in a sporadic and uncoordinated manner~\cite{Chen2020_massiveAccess,Wu2020_massiveAccess}. This calls for a new theoretical framework that helps to understand the fundamental limits on the energy efficiency achievable in massive uncoordinated access and provide guidelines for system design. 
	To this end, Polyanskiy~\cite{PolyanskiyISIT2017massive_random_access} proposed a novel formulation for the multiple-access problem, commonly referred to as unsourced multiple access that relies on three key assumptions: i) all users employ a common codebook and the decoder only aims to return a list of messages; ii) the error event is defined per user 
	and the error probability is averaged over the users; iii) each user sends a fixed amount of information bits within a finite-length frame. This formulation provides a unified framework within which traditional as well as modern random access protocols \cite{Berioli2016NOW} give achievability results. In \cite{PolyanskiyISIT2017massive_random_access}, an upper bound on the minimum energy per bit achievable on the Gaussian \gls{MAC} was derived. Modern random access schemes exhibit a large gap to this bound, and devising coding schemes approaching the bound is an active area of research~\cite{Ordentlich2017low_complexity_random_access,Vem2019,Fengler2019sparcs,Amalladinne2020,Amalladinne2021,Pradhan2020,Han2021}. The framework proposed in~\cite{PolyanskiyISIT2017massive_random_access} has recently been extended to the quasi-static fading channel~\cite{Kowshik2020}, multiple-antenna channel~\cite{Fengler2019nonBayesian,Shyianov2021}, and a setup with common alarm messages~\cite{Stern2019}.
	
	\revise{The achievability bound in~\cite{PolyanskiyISIT2017massive_random_access} is established for the scenario in which the number of active users is fixed and known to the receiver. This assumption has also been considered in most extensions of~\cite{PolyanskiyISIT2017massive_random_access}.}
	In practice, however, since \gls{IoT} devices access the channel at random times and in a grant-free manner, the number of active users varies over time, and hence, it is typically unknown to the receiver. Therefore, the bound in \cite{PolyanskiyISIT2017massive_random_access} may provide an overoptimistic benchmark on the performance of random-access schemes operating in the practically relevant scenario in which the number of active users is not known to the receiver. In this scenario, the decoder needs to determine the size of the list of transmitted messages. Choosing a list size smaller than the number of active users will result in \glspl{MD}, i.e., transmitted messages that are not included in the decoded list, whereas choosing it larger than the number of active users will result in \glspl{FA}, i.e., decoded messages that have not been transmitted. Furthermore, additional \glspl{MD} and \glspl{FA} may occur in the decoding process. 
	There is a trade-off between \gls{MD} and \gls{FA} probabilities: a decoder that always outputs the whole codebook will never misdetect, but has \gls{FA} probability close to one; similarly, a decoder that always outputs an empty set will never raise an \gls{FA} but will always misdetect. 
	
	Characterizing the \gls{MD}--\gls{FA} trade-off is a fundamental engineering challenge. \revisee{In the \gls{IoT}, an \gls{MD} can cause the system to be unaware of an important event reported by the devices,  whereas an \gls{FA} can trigger an unnecessary reaction that interrupts the system operation.  Depending on the application, the MD probability can be more critical than the FA probability or vice versa. For example, in security inspection, where the cost of an MD can be extremely high (such as not detecting a threat) while the cost of an FA is relatively low (simply a further inspection), the MD probability should be favored. In healthcare, alarm fatigue, i.e., desensitization of clinicians due to high exposure to alarms, has been recognized as a serious issue. 
	If MDs committed by health monitors do not cause vital consequences, the FA probability should be kept low to avoid alarm fatigue~\cite{sendelbach2013alarm,ruskin2015alarm}.}
	
	\revise{The \gls{MD}--\gls{FA} trade-off} was not addressed in \cite{PolyanskiyISIT2017massive_random_access}. \revise{Achievability bounds based on variable-length codes and feedback have been proposed for the general random-access channel~\cite{Effros2018ISIT}, and in particular for the Gaussian \gls{MAC}~\cite{Yavas2021Gaussian}, with unknown number of active users}. However, the authors considered the joint-user error event instead of the per-user error event, and thus, \gls{MD} and \gls{FA} were not explicitly considered. 
	\revise{The per-user error probability achieved with this random-coding scheme was considered in~\cite[Sec.~V-E]{Yavas2021}, but the \gls{MD}--\gls{FA} trade-off was not addressed.} To summarize, a random-coding bound accounting for both \gls{MD} and \gls{FA}, which can serve as a benchmark for unsourced multiple access with random user activity, is missing.

	
	Most of the practical coding schemes that have been proposed so far for common-codebook massive random access require knowledge of the number of active users. 
	Modern variations of the ALOHA protocol, such as irregular repetition slotted ALOHA~(IRSA)~\cite{Liva2011IRSA}, 
	can also operate when the number of active users is unknown. However, most of the results available for these schemes pertain to the packet loss rate, which accounts only for \gls{MD}. The successive interference cancellation coding scheme proposed in \cite{Vem2019} is also analyzed in terms of \gls{MD} only. Note that minimizing the \gls{MD} probability alone can entail a high \gls{FA} probability. In~\cite{Decurninge2020}, a tensor-based communication scheme was proposed, and both \gls{MD} and \gls{FA} probabilities were reported in the performance evaluation. Another scheme for which both \gls{MD} and \gls{FA} probabilities have been reported was recently proposed in~\cite{fengler2020pilot} for the quasi-static fading \gls{MAC} and for the case in which the receiver has a large number of antennas. However, in~\cite{Decurninge2020} and~\cite{fengler2020pilot}, the \gls{MD} and \gls{FA} probabilities are not reported separately but rather through their sum.
	
	
	\subsubsection*{Contributions}
	In this work, we extend Polyanskiy's bound to the case where the number of active users is {\em random} and {\em unknown}. Our contributions are summarized as follows. We first extend the definition of a random-access code provided in~\cite{PolyanskiyISIT2017massive_random_access} to account for both \gls{MD} and \gls{FA} probabilities. We then derive a random-coding achievability bound for the Gaussian \gls{MAC}. Unlike~\cite{PolyanskiyISIT2017massive_random_access}, we do not assume that the receiver knows the number of active users. To circumvent this issue, we let the decoder seek the best list size within a predetermined interval around an estimated value of the number of active users. Our decoding metric \revise{to determine the list of transmitted messages}, which is based on nearest neighbor decoding, is similar to the one used in \cite{Stern2019}. However, different from \cite{Stern2019}, we limit the decoded list size to belong to a finite-size set to avoid noise overfitting, especially in the low energy-per-bit regime. We use our random-coding bound to characterize both \gls{MD} and \gls{FA} in slotted ALOHA with multi-packet reception (SA-MPR). Finally, we derive a random-coding bound for a scheme that simply treats interference as noise, referred to as \gls{TIN}, in which the number of active users is unknown. 
	
	To gain engineering insights on the role of the knowledge of the number of active users, we compare our bound with the bound in~\cite[Th.~1]{PolyanskiyISIT2017massive_random_access}. We also use our bound to benchmark the energy efficiency of SA-MPR and \gls{TIN}, which do not require \emph{a priori} knowledge of the number of active users. Finally, we consider the scheme based on sparse regression codes (SPARCs) proposed in~\cite{Fengler2019sparcs} and its enhancement in~\cite{Amalladinne2021}, which were both derived for the case of known number of active users. We adapt these schemes to the case of unknown number of active users by performing an energy-based estimation of the number of active users and by letting the decoder treat the estimate as the true value.
	Numerical results pertaining to a scenario with $\revise{300}$ active users in average \revise{and framelength $19200$ complex channel uses} show that to achieve both \gls{MD} and \gls{FA} probabilities below $10^{-1}$, the required energy per bit predicted by our achievability bound is only $0.65$~dB higher than that predicted by the bound for a known number of active users~\cite[Th.~1]{PolyanskiyISIT2017massive_random_access}. For the same setting, the required energy per bit predicted by our bound is $9$~dB, $4.1$~dB, and $3.6$~dB lower than that of the SA-MPR bound, SPARC~\cite{Fengler2019sparcs}, and enhanced SPARC~\cite{Amalladinne2021}, respectively. 
	The gap between the performance of enhanced SPARC with known number of active users and with unknown number of active users is small. This suggests that, for mild requirements on $P_{\rm MD}$ and $P_{\rm FA}$, it is sufficient to adapt existing coding schemes that perform well in the case of known number of active users by simply adding an active-user estimation step and then treating the estimated number of active users as the true value.
	On the contrary, when we consider \gls{MD} probability and/or \gls{FA} probability below $10^{-3}$, the gap between our bound and the bound for a known number of active users in~\cite{PolyanskiyISIT2017massive_random_access} is much larger: around $3$--$4$~dB. For these stringent requirements, it turns out that it is energy inefficient to simply treat the estimated number of active users as the true value due to the errors that occur in the estimation step. 
	
	To summarize, our results suggest that for the Gaussian MAC, the lack of knowledge of the number of active users entails a small loss in terms of energy efficiency if the target \gls{MD} and \gls{FA} probabilities are high, as typically considered in the literature. In this case, it is effective to first estimate the number of active users, then treat this estimate as the true value, and use a coding scheme that performs well for the case where the number of active users is known. 
	However, for more stringent targets, the loss due to the lack of knowledge of the number of active users \revise{might} be large. \revise{It remains unclear whether this large gap is fundamental or pertains to the considered random-coding scheme only. For stringent requirements,} the approach to simply using the estimate of the number of active users to set the decoded list size results in poor energy efficiency. This calls for more sophisticated methods to handle effectively the uncertainty about the number of active users.


	\subsubsection*{Paper Outline}
	The remainder of the paper is organized as follows. In Section~\ref{sec:channel}, we present the channel model and define a random-access code. In Section~\ref{sec:RCU}, we propose a random-coding bound for the Gaussian \gls{MAC}, apply it to SA-MPR, and derive a bound for \gls{TIN} as well. 
	In Section~\ref{sec:numerical}, we present numerical results and discussions. \revisee{We provide some discussions in Section~\ref{sec:discussion} and conclude the paper in Section~\ref{sec:conclusions}.} The proofs are given in the appendices.
	
	\subsubsection*{Notation}
	Lowercase boldface letters denote vectors. Random quantities are denoted with non-italic letters with sans-serif font, e.g., a scalar $\rv{x}$ and a vector $\rvVec{v}$. 
	Deterministic quantities are denoted 
	with italic letters, e.g., a scalar $x$ and a vector $\bm{v}$. The $n\times n$ identity matrix is denoted by $\Id_n$. 
	The Euclidean norm is denoted by $\|\cdot\|$. 
	Calligraphic uppercase letters, e.g., $\Ac$, denote sets. We use $\mathfrak{P}(\Ac)$ to denote the set of all subsets of $\Ac$; 
	$[m:n] \defeq \{m,m+1,\dots,n\}$ if $m \le n$ and $[m:n] \defeq \emptyset$ if $m>n$; $[n] \defeq [1:n]$; $x^+ \defeq \max\{x,0\}$; $\ind{\cdot}$ is the indicator function. The set of natural and complex numbers are denoted by $\mathbb{N}$ and $\CC$, respectively. We denote the Gamma function by $\Gamma(x) \defeq \int_{0}^{\infty}z^{x-1}e^{-z}dz$, and the lower and upper incomplete Gamma functions by $\gamma(x,y) \defeq \int_{0}^{y}z^{x-1}e^{-z}dz$ and $\Gamma(x,y) \defeq \int_{y}^{\infty}z^{x-1}e^{-z}dz$, respectively. Finally, 
	$\Cc\Nc(\muv,\Sigmam)$ denotes the distribution of a complex proper
	Gaussian random vector with mean $\muv$ and covariance matrix~$\Sigmam$, and ${\rm Pois}(\lambda)$ the Poisson distribution with mean $\lambda.$
	
	\subsubsection*{Reproducible Research} 
	The Matlab code used to evaluate our random-coding bound is available at: \url{https://github.com/khachoang1412/UMA_random_user_activity}.
	
	\section{Random-Access Channel} \label{sec:channel}
	We consider a \gls{MAC} in which a random number $\rv{K}_{\rm a}$ of users transmit their messages to a receiver over $n$ uses of a stationary memoryless additive white Gaussian noise channel. Here, $\rv{K}_{\rm a}$ follows a distribution with \gls{PMF} $P_{\rv{K}_{\rm a}}$. 
	Let $\rvVec{x}_i \in \CC^n$ be the signal transmitted by user~$i$ over $n$ channel uses. The corresponding channel output is given by
	\begin{align} \label{eq:gaussian_MAC}
		\rvVec{y} = \sum_{i=1}^{\rv{K}_{\rm a}}\rvVec{x}_i + \rvVec{z}, 
	\end{align}
	where $\rvVec{z} \sim \Cc\Nc(\mathbf{0},\Id_n)$ is the Gaussian noise, which we assume being
	independent of $\{\rvVec{x}_i\}_{i=1}^{\rv{K}_{\rm a}}$.
	We also assume that the transmitted signals satisfy the power constraint 
	$
	\|\rvVec{x}_i\|^2 \le nP, ~ \forall i \in [\rv{K}_{\rm a}],
	$
	\revise{almost surely.}
	We further assume that the receiver does not know $\rv{K}_{\rm a}$ \emph{a priori}, but can choose to estimate it. 
	As in~\cite{PolyanskiyISIT2017massive_random_access}, our model differs from the classical \gls{MAC} in that the total number of users is not limited, all users employ the same codebook, and the receiver decodes up to a permutation of the messages. 
	However, as opposed to~\cite{PolyanskiyISIT2017massive_random_access}, where the number of active users is assumed to be fixed and known, we assume that $\rv{K}_{\rm a}$ is random and unknown. 
	We therefore need to account for both \gls{MD} and \gls{FA} probabilities.
	We next rigorously define these two quantities, as well as the notion of a random-access code.\footnote{Our definition of a random-access code can be extended straightforwardly to more general \glspl{MAC} with permutation-invariant channel law, as considered in~\cite[Def.~1]{PolyanskiyISIT2017massive_random_access}.}
	
	\begin{definition}[Random-access code for the Gaussian \gls{MAC}] \label{def:code}
		Consider the $\rv{K}_{\rm a}$-user Gaussian \gls{MAC} with $\rv{K}_{\rm a} \sim P_{\rv{K}_{\rm a}}$. 
		An $(M,n,\epsilon_{\rm MD},\epsilon_{\rm FA})$ random-access code for this channel, where $M$ is the size of the codebook, $n$ is the codeword length, and $\epsilon_{\rm MD},\epsilon_{\rm FA} \in (0,1)$, consists of:
		\begin{itemize}
			\item A random variable $\rv{U}$ defined on a set $\Uc$ 
			that is revealed to both the transmitters and the receiver before the start of the transmission. 
			
			\item An encoding function $f\colon \Uc \times [M] \to \CC^n$ that produces the transmitted codeword $\rvVec{x}_i = f(\rv{U},\rv{w}_i)$, satisfying the power constraint, of user $i$ for a given message $\rv{w}_i$ uniformly distributed over~$[M]$.
			
			\item A decoding function $g\colon \Uc \times \CC^n \to \mathfrak{P}([M])$ that provides an estimate $\widehat{\Wc} = \{\widehat{\rv{w}}_1,\dots,\widehat{\rv{w}}_{|\widehat{\Wc}|}\} = g(\rv{U},\rvVec{y})$ of the list of transmitted messages. 
		\end{itemize}
		Let $\widetilde{\Wc} = \{\widetilde{\rv{w}}_1,\dots,\widetilde{\rv{w}}_{|\widetilde{\Wc}|}\}$ denote the set of distinct elements of $\Wc = \{{\rv{w}}_1,\dots,{\rv{w}}_{\rv{K}_{\rm a}}\}$. We assume that the decoding function satisfies the following constraints on the \gls{MD} and \gls{FA} probabilities:
		\begin{align}
			P_{\rm MD} &\defeq \revise{\E\Bigg[{\frac{1}{|\widetilde{\Wc}|}\sum_{i=1}^{|\widetilde{\Wc}|} \P[\widetilde{\rv{w}}_i \notin \widehat{\Wc}]}\Bigg]} \le \epsilon_{\rm MD} \label{eq:def_pMD}\\
			P_{\rm FA} &\defeq \revise{\E\Bigg[{\frac{1}{|\widehat{\Wc}|} \sum_{i=1}^{|\widehat{\Wc}|} \P[\widehat{\rv{w}}_i \notin \widetilde{\Wc}]}\Bigg]} \le \epsilon_{\rm FA}. \label{eq:def_pFA}
		\end{align}
		The expectations in~\eqref{eq:def_pMD} and~\eqref{eq:def_pFA} are with respect to the \revise{random user activity, and we use the convention $0/0 = 0$ to circumvent the case $|\widetilde{\Wc}| = 0$ or $|\widehat{\Wc}| = 0$}. 
	\end{definition} 	
		
		\begin{remark}
			According to Definition~\ref{def:code}, the receiver aims to produce the list of \emph{distinct} transmitted messages\textemdash there are neither \gls{MD} nor \gls{FA} if $\widehat{\Wc} = \widetilde{\Wc}$. Therefore, 
			the event that two users transmit the same message does not result in an error if the message is included in the list of decoded messages.
		\end{remark}
		
		In the definition of random-access code proposed in~\cite[Def.~1]{PolyanskiyISIT2017massive_random_access}, the decoder outputs a list of messages of size equal to the number of active users, which is assumed to be known. In such a setup, an \gls{MD} implies an \gls{FA}, and vice versa. Hence, \revise{\gls{MD} and \gls{FA} events either occur simultaneously or do not occur simultaneously}. In our setup, the number of decoded messages $|\widehat{\Wc}|$ can be different from the number of distinct transmitted messages $|\widetilde{\Wc}|$. This motivates the definition of \gls{MD} and \gls{FA} probabilities provided in~\eqref{eq:def_pMD} and~\eqref{eq:def_pFA}, respectively. 

	\begin{remark} \label{remark:U}
		In the next section, we shall use a random-coding argument to obtain achievability bounds. We will construct a codebook ensemble for which~\eqref{eq:def_pMD} and~\eqref{eq:def_pFA} holds \revise{on} average. \revise{Specifically, the \gls{MD} and \gls{FA} probabilities, averaged over the codebook ensemble, are upper-bounded by $\epsilon_{\rm MD}$ and $\epsilon_{\rm FA}$, respectively.} Unfortunately, this does not imply that there exists a single code in this ensemble that achieves both~\eqref{eq:def_pMD} and~\eqref{eq:def_pFA}. \revise{In other words, the fact that a random code satisfies both~\eqref{eq:def_pMD} and~\eqref{eq:def_pFA} on average does not imply the existence of a deterministic code that satisfies these constraints.} The introduction of the random variable $\rv{U}$ in Definition~\ref{def:code} allows us to circumvent this issue by enabling randomized coding strategies. 
		Specifically, proceeding as in~\cite[Th.~19]{Polyanskiy2011feedback}, one can show that there exists a randomized coding strategy that achieves both \eqref{eq:def_pMD} and~\eqref{eq:def_pFA} and involves time-sharing among at most three deterministic codes (i.e., $|\Uc| \le 3$) in this ensemble. \revise{In fact, following the improvement in the size of the common randomness reported in~\cite[Th.~8]{Yavas2021}, one can show that $|\Uc| \le 2$ suffices.}
	\end{remark}
	
	\section{Random-Coding Bound} \label{sec:RCU}
	
	\subsection{Proposed Random-Coding Bound for $\rv{K}_{\rm a}$ Random and Unknown} \label{sec:random_coding_scheme}
	We first review the random-coding bound in~\cite[Th.~1]{PolyanskiyISIT2017massive_random_access}. 
	Let $\Wc = \{\rv{w}_1, \dots, \rv{w}_{K_{\rm a}}\}$ be the set of transmitted messages. Each of the active users picks a codeword $\cv_{\rv{w}_i}$, $i\in [\rv{K}_{\rm a}]$, from a common codebook containing $M$ codewords $\cv_1,\dots,\cv_M$ drawn independently from the distribution $\Cc\Nc(\mathbf{0},P'\Id_n)$ for a fixed $P' < P$.  To convey message $\rv{w}_i$, the corresponding active user transmits $\cv_{\rv{w}_i}$ provided that $\|\cv_{\rv{w}_i}\|^2 \le nP$. Otherwise, it transmits the all-zero codeword. That is, $\rvVec{x}_i = \cv_{\rv{w}_i} \ind{\|\cv_{\rv{w}_i}\|^2 \le P}$.
	The receiver employs a minimum distance decoder, in which the list of decoded messages is obtained as $\widehat{\Wc} = \arg\min_{{\Wc'} \subset [M], |{\Wc'}| = K_{\rm a}} \|c({\Wc'}) - \rvVec{y}\|^2$ where $c(\Wc') \defeq \sum_{i\in \Wc'} \cv_{i}$. The error analysis involves manipulations of unions of the pairwise error events via a change of measure and the application of the Chernoff bound combined with Gallager's $\rho$-trick~\cite[p.~136]{Gallager1968information}. 
	An alternative bound is also obtained by writing the pairwise error event as an inequality involving information densities, and by applying the tail bound on the information density given in~\cite[Cor.~17.1]{Polyanskiy2019lecture}.
	
	\begin{remark} \label{remark:iid_vs_spherical}
		\revisee{Using independent and identically distributed~(i.i.d.) Gaussian codewords allows for a tractable analysis that leads to explicit bounds for fixed $n$ and $P$. An alternative approach consists in using codebooks with codewords uniformly distributed on the power sphere, i.e., the so-called spherical/shell codebooks. For the conventional Gaussian MAC~\cite{Yavas2021Gaussian}, spherical codebooks were shown to achieve a better second-order (dispersion) term in the asymptotic expansion of the achievable rate region as $n \to \infty$ than i.i.d.~Gaussian codebooks. For the UMA setting, however, it is unclear if spherical codebooks achieve a better second-order term. Furthermore, with spherical codebooks, it appears challenging to obtain explicit closed-form bounds for fixed $n$ and $P$.}
	\end{remark}
	
	In the following, we derive a random-coding bound for the case in which $\rv{K}_{\rm a}$ is random and unknown to the receiver. Specifically, we consider a random-coding scheme with the same encoder as in \cite{PolyanskiyISIT2017massive_random_access}.\footnote{\revise{Strictly speaking, our encoding function takes not only the message as input (as in~\cite{PolyanskiyISIT2017massive_random_access}) but also the common randomness variable $\rv{U}$. However, for brevity, we omit $\rv{U}$ in the encoding and decoding functions. For details on how this common randomness is incorporated in the encoding and decoding functions, we refer the readers to~\cite{Polyanskiy2011feedback} and~\cite{Yavas2021}.}} \revise{The new challenge in our setting with respect to~\cite{PolyanskiyISIT2017massive_random_access} is that the receiver does not know $\rv{K}_{\rm a}$, and thus cannot use this number to set the decoded list size. To overcome this challenge, we let the receiver} estimate $\rv{K}_{\rm a}$ from $\rvVec{y}$, and then decide the best list size within an interval around the initial estimate of $\rv{K}_{\rm a}$. Specifically, given the channel output $\yv$, the receiver estimates $\rv{K}_{\rm a}$ as
	\begin{align} \label{eq:est_Ka}
		K_{\rm a}' = \arg\max_{K \in [K_\ell:K_u]} m(\yv,K),
	\end{align}
	where $m(\yv,K)$ is a suitably chosen metric, and $K_\ell$ and $K_u$ are suitably chosen lower and upper limits on $K_{\rm a}'$, respectively. The \revise{metric $m(\yv,K)$ and the} parameters $K_\ell$ and $K_u$ can, for example, be chosen based on prior knowledge (if available) on the distribution of $\rv{K}_{\rm a}$. 
	Then, given $K_{\rm a}'$, the receiver produces a list of decoded messages as
	\begin{equation} \label{eq:decoder_Ka'}
		\widehat{\Wc} = \arg\min_{{\Wc}' \subset [M], \underline{K_{\rm a}'} \le |{\Wc}'| \le \overline{K_{\rm a}'}} \|c({\Wc}') - \rvVec{y}\|^2,
	\end{equation}   
	where $\underline{K_{\rm a}'} =\max\{K_\ell,K_{\rm a}'-r\}$ and $\overline{K_{\rm a}'}\defeq \min\{K_u,K_{\rm a}'+r\}$, with $r$ being a nonnegative integer. 
	We refer to $r$ as the {\em decoding radius}. \revise{Note that if $r = 0$, the receiver outputs $K_{\rm a}'$ codewords, i.e., it treats the estimate $K_{\rm a}'$ as the true value.}
	An error analysis of this random-coding scheme conducted along similar lines as in \cite{PolyanskiyISIT2017massive_random_access} leads to the following result.
	
	\begin{theorem}[Random-coding bound, $\rv{K}_{\rm a}$ random and unknown]  \label{thm:RCU_unknownKa}
		Fix $P' < P$, $r$, $K_\ell$, and $K_{u}$ ($K_\ell \le K_{u}$). For the $\rv{K}_{\rm a}$-user Gaussian \gls{MAC} with $\rv{K}_{\rm a} \sim P_{\rv{K}_{\rm a}}$, there exists an $(M,n,\epsilon_{\rm MD},\epsilon_{\rm FA})$ random-access code satisfying the power constraint $P$ for which
		\begin{align}
			\epsilon_{\rm MD} &= \sum_{K_{\rm a} =\max\{K_\ell,1\}}^{K_{u}} \Bigg(P_{\rv{K}_{\rm a}}(K_{\rm a}) \sum_{K_{\rm a}' = K_\ell}^{K_{u}} \sum_{t\in \Tc}\frac{t+(K_{\rm a}-\overline{K_{\rm a}'})^+}{K_{\rm a}} \notag \\ 		&\qquad \cdot
			\min\{p_t,q_t, \xi(K_{\rm a},K_{\rm a}')\} \Bigg) + \tilde{p}, \label{eq:eps_MD}\\
			\epsilon_{\rm FA} &= \sum_{K_{\rm a} =K_\ell}^{K_{u}} \Bigg(P_{\rv{K}_{\rm a}}(K_{\rm a}) \sum_{K_{\rm a}' = K_\ell}^{K_{u}} \sum_{t\in \Tc} \sum_{t' \in \Tc_t}   \notag \\ 		&\qquad 
			\frac{t'+(\underline{K_{\rm a}'}-K_{\rm a})^+}{K_{\rm a} - t - {(K_{\rm a} - \overline{K_{\rm a}'})}^+ + t' + {(\underline{K_{\rm a}'}-K_{\rm a})}^+} \notag \\ 
			&\qquad \cdot \min\{p_{t,t'}, q_{t,t'}, \xi(K_{\rm a},K_{\rm a}')\} \Bigg) + \tilde{p}, \label{eq:eps_FA}
		\end{align}	
		where 
		\begin{align}
			\tilde{p} &= 2 - \sum_{K_{\rm a} = K_\ell}^{K_{u}}P_{\rv{K}_{\rm a}}(K_{\rm a}) - \E_{\rv{K}_{\rm a}}\left[\frac{M!}{M^{\rv{K}_{\rm a}}(M\!-\!\rv{K}_{\rm a})!} \right] \notag \\ 		&\quad 
			+ \E[\rv{K}_{\rm a}]  \frac{\Gamma(n,nP/P')}{\Gamma(n)}, \label{eq:p0}\\
			p_t &= \sum_{t'\in \overline{\Tc}_t} p_{t,t'}, \label{eq:pt}\\
			p_{t,t'} &= e^{-n E(t,t')}, \label{eq:ptt} \\
			E(t,t') &= \max_{\rho,\rho_1 \in [0,1]} -\rho\rho_1 t' R_1 - \rho_1 R_2 + E_0(\rho,\rho_1), \label{eq:Ett} \\
			 E_0(\rho,\rho_1) &= \max_{\lambda} \rho_1 a + \ln(1-\rho_1 P_2 b), \label{eq:E0}\\
			a &= \rho \ln(1+ P' t' \lambda) + \ln(1+ P't \mu), \label{eq:a}\\ 
			b &= \rho\lambda -\frac{\mu}{1+ P't\mu}, \label{eq:b} \\ 
			\mu &= \frac{\rho \lambda}{1+P't'\lambda}, \\
			P_2 &= 1+ \big((K_{\rm a} - \overline{K_{\rm a}'})^+ + (\underline{K_{\rm a}'} - K_{\rm a})^+\big)P', \label{eq:P2}\\
			R_1 &=  \frac{1}{nt'} \ln\binom{M - \max\{K_{\rm a},\underline{K_{\rm a}'}\}}{t'}, \label{eq:R1}
			\\
			R_2 &=  \frac{1}{n} \ln \binom{\min\{K_{\rm a}, \overline{K_{\rm a}'}\}}{t}, \\
			q_t &= \inf_{\gamma} \bigg(\P[\rv{I}_{t} \le \gamma] \notag \\ 
                &\quad + \sum_{t'\in \overline{\Tc}_t}
			\exp(n(t'R_1 + R_2) - \gamma)\bigg), \label{eq:qt}\\
			q_{t,t'} &= \inf_{\gamma} \Big(\P[\rv{I}_{t} \le \gamma] + \exp(n(t'R_1 + R_2) - \gamma)\Big), \label{eq:qtt} \\
			\Tc &= [0:\min\{\overline{K_{\rm a}'},K_{\rm a},M-\underline{K_{\rm a}'} - (K_{\rm a} - \overline{K_{\rm a}'})^+\}], \label{eq:T} \\
			\Tc_t &= \big[\big({(K_{\rm a} - \overline{K_{\rm a}'})}^+ \!- {(\underline{K_{\rm a}'} - K_{\rm a})}^+ \!+ \max\{\underline{K_{\rm a}'},1\} \notag \\ 		&\quad \quad 
			- K_{\rm a} + t\big)^+ : u_t\big],	\label{eq:Tt}	\\
			\overline{\Tc}_t &= \big[\big({(K_{\rm a} - \overline{K_{\rm a}'})}^+ - {(K_{\rm a}-\underline{K_{\rm a}'})}^+ + t\big)^+ : u_t \big], \label{eq:Tbart} \\
			u_t &= \min\big\{{(\overline{K_{\rm a}'} - K_{\rm a})}^+ - {(\underline{K_{\rm a}'}-K_{\rm a})}^+ + t,  \notag \\
                &\quad \quad \overline{K_{\rm a}'} - {(\underline{K_{\rm a}'}-K_{\rm a})}^+, M-\max\{\underline{K_{\rm a}'},K_{\rm a}\}\big\}, \\
			\xi(K_{\rm a},K_{\rm a}') &= 
			\min_{K \in [K_\ell:K_u]\colon K \ne K_{\rm a}'} \P[m\left(\rvVec{y}_0,K_{\rm a}' \right) \!>\! m\left(\rvVec{y}_0,K\right)]. \label{eq:xi}
		\end{align}
		In~\eqref{eq:xi}, $\rvVec{y}_0 \sim \Cc\Nc(\mathbf{0},(1+K_{\rm a}P')\Id_n)$. The random variable $\rv{I}_t$ in~\eqref{eq:qt} and \eqref{eq:qtt} is defined as
		\begin{align} \label{eq:def_It}
			\rv{I}_t &\defeq \min_{\revise{\Wc_{\rm aMD}} \subset [(K_{\rm a} - \overline{K_{\rm a}'})^+ + 1:K_{\rm a}] \atop |\revise{\Wc_{\rm aMD}}| = t} \notag \\
            &\qquad\imath_t(c(\revise{\Wc_{\rm iFA}}) + c(\revise{\Wc_{\rm aMD}});\rvVec{y} \cond c([K_{\rm a}] \setminus \revise{\Wc_{\rm MD}})),
		\end{align} 
		where $\revise{\Wc_{\rm iFA}} = [K_{\rm a} + 1: \underline{K_{\rm a}'}]$, $\revise{\Wc_{\rm MD}} = [(K_{\rm a} - \overline{K_{\rm a}'})^+] \cup \revise{\Wc_{\rm aMD}}$, 
		and the information density is defined as 
		\begin{align}
			&\imath_t(c(\revise{\Wc_{\rm MD}});\rvVec{y} \cond c(\Wc \setminus \revise{\Wc_{\rm MD}})) \notag \\ 
			&\defeq n \ln(1+(t+(K_{\rm a}-\overline{K_{\rm a}'})^+)P') \notag \\	&\quad + \frac{\|\rvVec{y} - c(\Wc \setminus \revise{\Wc_{\rm MD}})\|^2}{1+(t+(K_{\rm a}-\overline{K_{\rm a}'})^+)P'} \notag \\	&\quad 
			- \|\rvVec{y} - c(\revise{\Wc_{\rm MD}}) - c(\Wc \setminus \revise{\Wc_{\rm MD}})\|^2.  \label{eq:infor_den}
		\end{align}
	\end{theorem}
	
	\begin{proof}
		\revise{The proof follows the footsteps of~\cite{PolyanskiyISIT2017massive_random_access} with some new ingredients to overcome the challenges posed by the random and unknown number of active users. First, for the proposed two-step decoder, the quality of both the estimation and decoding steps needs to be analyzed. Second, due to the possible mismatch between the numbers of transmitted and decoded messages, the sets of MDs and FAs need to be carefully handled. See Appendix~\ref{app:proof} for details.}
	\end{proof}
	
	Some remarks are in order.
	\begin{enumerate} [label={\roman*)}]
		\item \revise{Let the sets of misdetected messages and falsely alarmed messages be denoted by
			\begin{align}
				\Wc_{\rm MD} \defeq \widetilde{\Wc} \setminus \widehat{\Wc}, \\ 
				\Wc_{\rm FA} \defeq \widehat{\Wc} \setminus \widetilde{\Wc},
			\end{align}
			respectively. Let $K_{\rm a} \to K'_{\rm a}$ denote the event that the estimation step outputs $K_{\rm a}'$ when $K_{\rm a}$ users are active. Given $K_{\rm a} \to K_{\rm a}'$, note that the set of possible decoded list sizes $[\underline{K_{\rm a}'}:\overline{K_{\rm a}'}]$ might not contain the number $K_{\rm a}$ of transmitted messages. If $\overline{K_{\rm a}'} < K_{\rm a}$, the decoder commits at least $K_{\rm a} - \overline{K_{\rm a}'}$ \glspl{MD}; if $\underline{K_{\rm a}'} > K_{\rm a}$, the decoder commits at least $\underline{K_{\rm a}'} - K_{\rm a}$ \glspl{FA}. Furthermore, there can be additional \glspl{MD} and \glspl{FA} occurring during the decoding process. Accordingly, we further break down the sets $\Wc_{\rm MD}$ and $\Wc_{\rm FA}$ as follows. We set $\Wc_{\rm MD} = \Wc_{\rm iMD} \cup \Wc_{\rm aMD}$, where $\Wc_{\rm iMD}$ denotes the list of $({K}_{\rm a} - \overline{K_{\rm a}'})^+$ \textit{initial} \glspl{MD} due to insufficient decoded list size, and $\Wc_{\rm aMD}$ the \textit{additional} \glspl{MD} that occurs during decoding. Similarly, we set $\Wc_{\rm FA} = \Wc_{\rm iFA} \cup \Wc_{\rm aFA}$, where $\Wc_{\rm iFA}$ denotes the list of $(\underline{K_{\rm a}'}-{K}_{\rm a})^+$ \textit{initial} \glspl{FA} due to excessive decoded list size, and $\Wc_{\rm aFA}$ the \textit{additional} \glspl{FA}. This is explained in details in Appendix~\ref{app:proof}.}
		
		\item \revise{The bounds in~\eqref{eq:eps_MD} and~\eqref{eq:eps_FA} are obtained by writing the \gls{MD} and \gls{FA} probabilities as $P_{\rm MD}=\E[\md]$ and $P_{\rm FA} = \E[\fa]$. Similar to~\cite{PolyanskiyISIT2017massive_random_access}, to facilitate the bounding, we make a change of measure over which these expectations are taken  at a cost of adding a constant bounded by $\tilde{p}$ given in~\eqref{eq:p0}. We then expand these expectations over possible values of the size of the sets $\widetilde{\Wc}$, $\widehat{\Wc}$, $\Wc_{\rm MD}$, and $\Wc_{\rm FA}$. The terms $\frac{t+(K_{\rm a}-\overline{K_{\rm a}'})^+}{K_{\rm a}}$ and $\frac{t'+(\underline{K_{\rm a}'}-K_{\rm a})^+}{K_{\rm a} - t - {(K_{\rm a} - \overline{K_{\rm a}'})}^+ + t' + {(\underline{K_{\rm a}'}-K_{\rm a})}^+}$ are realizations of $\md$ and $\fa$, respectively, given that $K_{\rm a} \to K_{\rm a}'$, $|\Wc_{\rm aMD}| = t$, and $|\Wc_{\rm aFA}| = t'$. Furthermore, the probabilities $\P[K_{\rm a} \to K_{\rm a}', |\Wc_{\rm aMD}| = t]$ and $\P[K_{\rm a} \to K_{\rm a}', |\Wc_{\rm aMD}| = t, |\Wc_{\rm aFA}| = t']$ associated with these realizations are upper-bounded by $P_{\rv{K}_{\rm a}}(K_{\rm a})\min\{p_t,q_t, \xi(K_{\rm a},K_{\rm a}')\}$ and $P_{\rv{K}_{\rm a}}(K_{\rm a}) \min\{p_{t,t'}, q_{t,t'}, \xi(K_{\rm a},K_{\rm a}')\}$, respectively. Here, $\xi(K_{\rm a},K_{\rm a}')$ is an upper bound on $\P[K_{\rm a} \to K_{\rm a}']$, whereas $\min\{p_t,q_t\}$ and $\min\{p_{t,t'}, q_{t,t'}\}$ are upper bounds on $\P[|\Wc_{\rm aMD}| = t]$ and $\P[|\Wc_{\rm aMD}| = t, |\Wc_{\rm aFA}| = t']$, respectively.}
		
		\item  \revise{The bounds $p_{t,t'}$ and $q_{t,t'}$ on $\P[|\Wc_{\rm aMD}| = t, |\Wc_{\rm aFA}| = t']$ are obtained following the error-exponent-based and dependence-testing-based approaches, respectively. Similar approaches are used to obtain the bounds $p_{t}$ and $q_{t}$, respectively, on $\P[|\Wc_{\rm aMD}| = t]$. Both approaches have been used in~\cite{PolyanskiyISIT2017massive_random_access} and can improve the bounds in different regimes. Numerical experiments suggest that $p_t$ and $p_{t,t'}$ dominates for medium and high values of $t, t'$ and $K_{\rm a}$, while $q_t$ and $q_{t,t'}$ dominates when $t,t'$, and $K_{\rm a}$ are small. Computing $q_t$ and $q_{t,t'}$ is more cumbersome than $p_t$ and $p_{t,t'}$. Thus, one can limit the evaluation of these terms to small $t,t'$, and $K_{\rm a}$  to reduce the complexity. This only loosens the bound slightly.}
		
		\item 
		\revise{The parameters $K_\ell$ and $K_u$ are introduced to facilitate the numerical evaluation of the bounds $\epsilon_{\rm MD}$ and $\epsilon_{\rm FA}$. In particular, we use them to avoid infinite sums over $K_{\rm a}$ and $K_{\rm a}'$, since the domain of $\rv{K}_{\rm a}$ can be unbounded.} 
		It is often convenient to set $K_\ell$ to be the largest value and $K_{u}$ the smallest value for which $\sum_{K_{\rm a} = K_\ell}^{K_{u}}P_{\rv{K}_{\rm a}}(K_{\rm a})$ exceeds a predetermined threshold.  
		
		\item The term $1 - \E_{\rv{K}_{\rm a}}\left[\frac{M!}{M^{\rv{K}_{\rm a}}(M-\rv{K}_{\rm a})!} \right]$ in $\tilde{p}$ (see~\eqref{eq:p0}) can be upper-bounded by $\E_{\rv{K}_{\rm a}}\big[\binom{\rv{K}_{\rm a}}{2}/M\big]$ as in \cite{PolyanskiyISIT2017massive_random_access}. 
		
		\item The term $R_1$ in~\eqref{eq:R1} can be upper-bounded by $ \frac{1}{n} \ln (M - \max\{K_{\rm a},\underline{K_{\rm a}'}\}) - \frac{1}{nt'} \ln t'!$, which allows for a numerically stable computation when $M - \max\{K_{\rm a},\underline{K_{\rm a}'}\}$ is large.
		
		\item The optimal $\lambda$ in~\eqref{eq:E0} is given by the largest real root of the cubic function $c_1x^3 + c_2x^2 + c_3x + c_4$ with
        \begin{align}
            c_1 &=  -\rho \rho_1(\rho\rho_1 + 1)t'P'P_2P_3^2, \\ 
            c_2 &= \rho\rho_1 t'P'P_3^2 - \rho\rho_1(3-\rho_1)t'P'P_2P_3 	\notag \\  
                &\quad -\rho\rho_1(\rho_1+1)P_2P_3^2, \\
            c_3 &= (2\rho-1)\rho_1 t'P'P_3 + \rho_1 P_3^2 - 2\rho\rho_1 P_2P_3, \\ 
            c_4 &= (\rho-1)\rho_1 t'P' + \rho_1 P_3, 
        \end{align}
		where~$P_2$ is given by~\eqref{eq:P2} and $P_3 \defeq (t' + \rho t)P'$.
		
	\end{enumerate}
	
	
	\revise{Although we follow the approach in~\cite{PolyanskiyISIT2017massive_random_access}, the novel aspects of our proof are that we provide a combined analysis of both the estimation and decoding steps. Furthermore, we single out the initial MDs and FAs that can not be avoided and carefully count the additional MDs and FAs that can occur during the decoding process. The separate treatment of the initial and additional MDs and FAs is crucial for our bound.}
	
	In the following theorem, we provide closed-form expressions for $\xi(K_{\rm a},K_{\rm a}')$ for two different estimators of $\rv{K}_{\rm a}$.
	
	\begin{theorem}[Closed-form expressions for $\xi(K_{\rm a},K_{\rm a}')$] \label{thm:xi}
		For the \gls{ML} estimation of $\rv{K}_{\rm a}$, i.e., $m(\yv,K) = \ln p_{\rvVec{y} | \rv{K}_{\rm a}}(\yv | K)$, $\xi(K_{\rm a},K_{\rm a}')$ is given by
		\begin{align} \label{eq:xi_ML}
			&\xi(K_{\rm a},K_{\rm a}') \notag \\
            &\defeq \min_{K \in [K_\ell:K_u] \colon K \ne K_{\rm a}'} \bigg(\ind{K < K_{\rm a}'}\frac{\Gamma(n,\zeta(K,K_{\rm a},K_{\rm a}')\revise{)}}{\Gamma(n)} \notag \\ &	\qquad \qquad \qquad 
			+ \ind{K > K_{\rm a}'}\frac{\gamma(n,\zeta(K,K_{\rm a},K_{\rm a}'))}{\Gamma(n)}\bigg),
		\end{align}
		with 
		\begin{align}
			\zeta(K,K_{\rm a},K_{\rm a}') &\defeq n \ln\left(\frac{1+KP'}{1+K_{\rm a}'P'}\right)(1+K_{\rm a}P')^{-1} \notag \\		&\quad \cdot
			\left(\frac{1}{1+K_{\rm a}'P'}-\frac{1}{1+KP'}\right)^{-1}. \label{eq:zeta_ML}
		\end{align}
		
		For an energy-based estimation\footnote{\revise{The energy-based estimator relies on the fact that the normalized squared norm of the channel output $\yv$ concentrates around its mean for large $n$. This effect was also exploited to determine decoding times for the variable-length code in~\cite{Yavas2021Gaussian}.}} of $\rv{K}_{\rm a}$, i.e., $m(\yv,K) = - |\|\yv\|^2 - n(1 + KP')|$, $\xi(K_{\rm a},K_{\rm a}')$ is given by~\eqref{eq:xi_ML} with \begin{align}
			\zeta(K,K_{\rm a},K_{\rm a}') \defeq \frac{n}{1+K_{\rm a}P'}\left(1+\frac{K + K_{\rm a}'}{2}P'\right). \label{eq:zeta_energy}
		\end{align}
		
	\end{theorem}
	\begin{proof}
		See Appendix~\ref{proof:xi}.
	\end{proof}

	As mentioned in Remark i, \revisee{since our decoder outputs a list of size within the interval $[\underline{K_{\rm a}'}:\overline{K_{\rm a}'}]$, it commits initial \glspl{MD} or \glspl{FA} when the true $\rv{K}_{\rm a}$ falls outside of this interval. In Section~\ref{sec:discussion_decoder}, we shall discuss why we choose to restrict the decoded list size to this interval. Due to the initial \glspl{MD} and \glspl{FA},  the \gls{MD} and \gls{FA} probabilities do not vanish even when $P\to\infty$ and all other system parameters are kept fixed.} In other words, the bounds $\epsilon_{\rm MD}$ and $\epsilon_{\rm FA}$ in Theorem~\ref{thm:RCU_unknownKa} exhibits error floors when $P$ is large. To characterize this effect, we put forth the following asymptotic lower bounds on $\epsilon_{\rm MD}$ and $\epsilon_{\rm FA}$, which are obtained by assuming that no additional \gls{MD} or \gls{FA} occurs on top of the initial \glspl{MD} or \glspl{FA}.
	
	\begin{corollary}[Asymptotic lower bounds on $\epsilon_{\rm MD}$ and $\epsilon_{\rm FA}$] \label{cor:error_floor}
		With ML or energy-based estimation of $\rv{K}_{\rm a}$, 
		\revisee{$\epsilon_{\rm MD}$ and $\epsilon_{\rm FA}$ given in~\eqref{eq:eps_MD} and~\eqref{eq:eps_FA}, respectively, satisfy}
		\begin{align}
			&\lim_{P\to\infty} \epsilon_{\rm MD} \notag \\&\ge \bar{\epsilon}_{\rm MD} \\
            &= \sum_{K_{\rm a} =\max\{K_\ell,1\}}^{K_{u}} \bigg(P_{\rv{K}_{\rm a}}(K_{\rm a}) \sum_{K_{\rm a}' = K_\ell}^{K_{u}} \frac{(K_{\rm a}-\overline{K_{\rm a}'})^+}{K_{\rm a}} 
			{\xi}(K_{\rm a},K_{\rm a}') \bigg) \notag \\ 		&\quad+ \bar{p}, \label{eq:eps_MD_floor}\\
			&\lim_{P\to\infty} \epsilon_{\rm FA} \notag \\&\ge \bar{\epsilon}_{\rm FA} \\
            &= \sum_{K_{\rm a} =K_\ell}^{K_{u}} \bigg(P_{\rv{K}_{\rm a}}(K_{\rm a}) 
			\!\sum_{K_{\rm a}' = K_\ell}^{K_{u}}  \frac{(\underline{K_{\rm a}'}-K_{\rm a})^+}{K_{\rm a} - {(K_{\rm a} \!-\! \overline{K_{\rm a}'})}^+ + {(\underline{K_{\rm a}'}\!-\!K_{\rm a})}^+} \notag \\ &\qquad \cdot 
			{\xi}(K_{\rm a},K_{\rm a}') \bigg) + \bar{p}, \label{eq:eps_FA_floor}
		\end{align}	
		where $$\bar{p} = 2 - \sum_{K_{\rm a} = K_\ell}^{K_{u}}P_{\rv{K}_{\rm a}}(K_{\rm a}) - \E_{\rv{K}_{\rm a}}\left[\frac{M!}{M^{\rv{K}_{\rm a}}(M-\rv{K}_{\rm a})!} \right],$$ and $\xi(K_{\rm a},K_{\rm a}')$ is given by~\eqref{eq:xi_ML} with
		$\zeta(K,K_{\rm a},K_{\rm a}') = n \ln\big(\frac{K}{K_{\rm a}'}\big) K_{\rm a}^{-1}\big(\frac{1}{K_{\rm a}'} - \frac{1}{K}\big)^{-1}$ for ML estimation of $\rv{K}_{\rm a}$ and $\zeta(K,K_{\rm a},K_{\rm a}') = n\frac{K+K_{\rm a}'}{2 K_{\rm a}}$ for energy-based estimation of $\rv{K}_{\rm a}$. 
	\end{corollary}
	\begin{proof}
		See Appendix~\ref{proof:error_floor}.
	\end{proof}
	\begin{remark} \label{remark:error_floor}
		The lower bounds in~\eqref{eq:eps_MD_floor} and~\eqref{eq:eps_FA_floor} are tight for typical IoT settings. Indeed, equalities in~\eqref{eq:eps_MD_floor} and~\eqref{eq:eps_FA_floor} hold if the probability of having additional \glspl{MD} and \glspl{FA} vanishes, i.e., $\min\{p_t,q_t\} \to 0$ for $t \ne 0$ and $\min\{p_{t,t'},p_{t,t'}\} \to 0$ for $(t,t') \ne (0,0)$ as $P\to\infty$. With $\rho = \rho_1 = 1$, 
		the optimal $\lambda$ in~\eqref{eq:E0} is given by $\lambda = 1/(2P_2)$. Thus, by replacing the maximization over $\rho$ and $\rho_1$ in~\eqref{eq:Ett} with $\rho = \rho_1 = 1$, we obtain that $E(t,t') \ge -t' R_1 - R_2 + \ln\big(1+\frac{(t+t')P'}{4P_2}\big)$. It follows that 
		\begin{align}
			p_{t,t'} &\le \binom{M-\max\{K_{\rm a}, \underline{K_{\rm a}'}\}}{t'} \binom{\min\{K_{\rm a}, \overline{K_{\rm a}'}\}}{t} \notag \\ 
            &\quad \cdot \left(1+\frac{(t+t')P'}{4P_2}\right)^{-n}.  \label{eq:bound_ptt} 
	\end{align}
	If $K_{\rm a} \in  [\underline{K_{\rm a}'}:\overline{K_{\rm a}'}]$, i.e., $P_2 = 1$, the right-hand side of~\eqref{eq:bound_ptt} vanishes as $P' \to\infty$. Otherwise, the right-hand side of~\eqref{eq:bound_ptt} converges to 
	\begin{align}
		\bar{p}_{t,t'} &= \binom{M-\max\{K_{\rm a}, \underline{K_{\rm a}'}\}}{t'} \binom{\min\{K_{\rm a}, \overline{K_{\rm a}'}\}}{t} \notag \\ 
            &\quad \cdot \bigg(1+\frac{t+t'}{4((K_{\rm a} - \overline{K_{\rm a}'})^+ + (\underline{K_{\rm a}'} - K_{\rm a})^+)}\bigg)^{-n} \\
		&\le M^{t'} K_{\rm a}^t \bigg(1+\frac{t+t'}{4((K_{\rm a} - \overline{K_{\rm a}'})^+ + (\underline{K_{\rm a}'} - K_{\rm a})^+)}\bigg)^{-n}.
	\end{align}
	Observe that $\bar{p}_{t,t'}$ is small if $n$ is relately large compared to $\ln M$ and $\ln K_{\rm a}$, which is true for relevant values of $n,M$ and $K_{\rm a}$ in the IoT. Specifically, in typical IoT scenarios, $\log_2 M$ and $K_{\rm a}$ are in the order of $10^2$ \revise{to $10^3$}, while $K_{\rm a}/n$ is from $10^{-4}$ to $10^{-3}$\textemdash see~\cite{PolyanskiyISIT2017massive_random_access} and \cite[Rem.~3]{Zadik2019}.\footnote{\revise{A typical IoT setting was discussed in \cite[Rem.~3]{Zadik2019}, wherein a metropolitan area was considered with $10^6$--$10^7$ devices communicating using the sub-GHz industrial, scientific and medical (ISM) band with $20$ MHz bandwidth. If each device is active a few times per hour, several hundreds to several thousands of devices are active per second. The number of degrees of freedom per active device is thus about $10^3$--$10^4$. Therefore, the average number of active devices per degree of freedom is from $10^{-4}$ to $10^{-3}$.}} For example, with $(M,n) = (2^{100}, 15000)$ and $K_{\rm a} \le 300$ as considered in~\cite{PolyanskiyISIT2017massive_random_access} and many follow-up works, assume that $(K_{\rm a} - \overline{K_{\rm a}'})^+ + (\underline{K_{\rm a}'} - K_{\rm a})^+ \le 20$, then $\bar{p}_{t,t'} < 10^{-128}$ for every $t\le 300$ and $t' \le 300$. 
	As a consequence, $p_{t,t'}$ and $p_t$ are very small. We conclude that $\lim\limits_{P\to\infty} \epsilon_{\rm MD}$ and $\lim\limits_{P\to\infty} \epsilon_{\rm FA}$ approach closely $\bar{\epsilon}_{\rm MD}$ and $\bar{\epsilon}_{\rm FA}$, respectively. In other words, $\bar{\epsilon}_{\rm MD}$ and $\bar{\epsilon}_{\rm FA}$ essentially characterize the error floors of ${\epsilon}_{\rm MD}$ and ${\epsilon}_{\rm FA}$, respectively, as $P\to \infty$. We shall further validate this argument through numerical experiments in Section~\ref{sec:numerical}. 
\end{remark}

\revise{The choice of the decoding radius $r$ turns out to be crucial. It} can be optimized according to the target \gls{MD} and \gls{FA} probabilities.\footnote{\revise{The choice of $r$ can also depend on the knowledge of $P_{\rv{K}_{\rm a}}$ (if available). For example, if the domain of $\rv{K}_{\rm a}$ contains well-separated values, it is not necessary to set $r > 0$. However, for $P_{\rv{K}_{\rm a}}$ that has non-negligible probability mass for consecutive values of $\rv{K}_{\rm a}$, setting $r > 0$ may improve the performance.}} On the one hand, a large decoding radius results in a reduction of the initial \glspl{MD} and \glspl{FA}, and thus a reduction of the error floors. Indeed, it is easy to verify that both $\bar{\epsilon}_{\rm MD}$ in~\eqref{eq:eps_MD_floor} and $\bar{\epsilon}_{\rm FA}$ in~\eqref{eq:eps_FA_floor} decrease with $r$. As $r \to \infty$, both $\bar{\epsilon}_{\rm MD}$ and $\bar{\epsilon}_{\rm FA}$ vanish. On the other hand, a large decoding radius leads to \emph{overfitting}, especially when $P$ is small. Specifically, \reviseee{as we shall clarify in Section~\ref{sec:discussion_decoder},} when the noise dominates the signal and the interference, increasing $r$ \reviseee{seems to increase} the chance that the decoder~\eqref{eq:decoder_Ka'} returns a list containing codewords whose sum is closer in Euclidean distance to the noise than to the sum of the transmitted codewords. 
When $P$ is sufficiently small, numerical experiments indicate that it is optimal to set $r = 0$, \revise{which introduces a bias that helps overcome overfitting.} 
In this case, we obtain the following achievability bound. 
\begin{corollary}[Zero decoding radius]  \label{coro:zero_decRad}
	Fix $P' < P$, $K_\ell$, and $K_{u}$ ($K_\ell \le K_{u}$). There exists an $(M,n,\epsilon_{\rm MD},\epsilon_{\rm FA})$ random-access code satisfying the power constraint $P$ for which 
	\begin{align}
		\epsilon_{\rm MD} &= \sum_{K_{\rm a} =\max\{K_\ell,1\}}^{K_{u}} \bigg( P_{\rv{K}_{\rm a}}(K_{\rm a}) \sum_{K_{\rm a}' = K_\ell}^{K_u} \sum_{t=0}^{\psi}\frac{t+(K_{\rm a}-K_{\rm a}')^+}{K_{\rm a}} \notag \\ &\qquad \cdot 
		\min\{p_{t,t},q_{t,t},\xi(K_{\rm a},K'_{\rm a})\} \bigg) + \tilde{p}, \label{eq:zeroRad_MD}\\
		\epsilon_{\rm FA} &= \sum_{K_{\rm a} =K_\ell}^{K_{u}} \bigg(P_{\rv{K}_{\rm a}}(K_{\rm a}) \sum_{K_{\rm a}' = \max\{K_\ell,1\}}^{K_u} \sum_{t=0}^{\psi} \frac{t+(K_{\rm a}' - K_{\rm a})^+}{K_{\rm a}'} \notag \\ &\qquad \cdot 
		\min\{p_{t,t}, q_{t,t},\xi(K_{\rm a},K'_{\rm a})\}\bigg) + \tilde{p},  \label{eq:zeroRad_FA}
	\end{align}
	where $\psi \defeq \min\{K_{\rm a}',K_{\rm a},M-\max\{K_{\rm a}',K_{\rm a}\}\}$, and $\tilde{p}$, $p_{t,t'}$, and  $q_{t,t'}$ are given by~\eqref{eq:p0}, \eqref{eq:ptt}, and \eqref{eq:qtt}, respectively, with $r = 0$.
\end{corollary}

If the number of active users is fixed to $K_{\rm a}$, by letting $K_{\rm a}' = K_{\rm a}$, one obtains from Corollary~\ref{coro:zero_decRad} a trivial generalization of \cite[Th.~1]{PolyanskiyISIT2017massive_random_access} to the complex case. 

Our random-coding scheme belongs to a family of random-coding schemes that first estimate the number of active users and then exploit this estimate to choose the number of returned messages. The next theorem gives an ensemble converse bound on the \gls{MD} and \gls{FA} probabilities, i.e., a converse bound on the \gls{MD} and \gls{FA} probabilities averaged over a random codebook ensemble, for schemes of this type.   

\begin{theorem} [An ensemble converse bound for a family of coding schemes] \label{th:converse}
	Fix $K_\ell$ and $K_u$ ($K_\ell \le K_u$).  Consider a decoding function that first estimates $\rv{K}_{\rm a}$ by $K_{\rm a}'$ as in~\eqref{eq:est_Ka} and then returns a list of at least $\underline{K_{\rm a}'}$ and at most $\overline{K_{\rm a}'}$ messages, where $\underline{K_{\rm a}'}$ and $\overline{K_{\rm a}'}$ are functions of $K_{\rm a}'$. Consider a random-access code consisting of $M$ length-$n$ codewords drawn independently from a distribution $P_{\rv{x}}$ and the aforementioned decoding function. It holds that
	\begin{align}
		\E[P_{\rm MD}] &\ge \underline{\epsilon}_{\rm MD} \notag \\
		&= \E_{\rv{K}_{\rm a}}\left[\frac{M!}{M^{\rv{K}_{\rm a}}(M-\rv{K}_{\rm a})!} \right] \sum_{K_{\rm a} = \max\{K_\ell,1\}}^{K_u} \bigg(P_{\rv{K}_{\rm a}}(K_{\rm a}) \notag \\ 		&\qquad \cdot  \sum_{K_{\rm a}' = K_\ell}^{K_u} \frac{(K_{\rm a}-\overline{K_{\rm a}'})^+}{K_{\rm a}} 
		\P[K_{\rm a} \to K_{\rm a}'] \bigg), \label{eq:converse_MD}\\ 
		\E[P_{\rm FA}] &\ge \underline{\epsilon}_{\rm FA} \notag \\
		&= \E_{\rv{K}_{\rm a}}\left[\frac{M!}{M^{\rv{K}_{\rm a}}(M-\rv{K}_{\rm a})!} \right]  \sum_{K_{\rm a} = K_\ell}^{K_u} \bigg(P_{\rv{K}_{\rm a}}(K_{\rm a}) \notag \\ 		&\qquad \cdot 
		\sum_{K_{\rm a}' =K_\ell}^{K_u} \frac{(\underline{K_{\rm a}'}-K_{\rm a})^+}{\overline{K_{\rm a}'}} 
		\P[K_{\rm a} \to K_{\rm a}'] \bigg). \label{eq:converse_FA}
	\end{align}	
	Here, the expectations on the left-hand sides of~\eqref{eq:converse_MD} and~\eqref{eq:converse_FA} are with respect to the codeword distribution $P_{\rv{x}}$, and $\P[K_{\rm a} \to K_{\rm a}']$ is the probability that the estimation step outputs $K'_{\rm a}$ when $K_{\rm a}$ users are active. That is, $\P[K_{\rm a} \to K_{\rm a}'] = \P[K_{\rm a}' = \displaystyle \arg\max_{K \in [K_\ell:K_u]} m(\rvVec{y},K) \cond \rv{K}_{\rm a} = K_{\rm a}]$, where the distribution of the received signal $\rvVec{y}$ is the one induced by the codeword distribution $P_{\rv{x}}$ and the Gaussian noise via~\eqref{eq:gaussian_MAC}.
\end{theorem} 
\begin{proof}
	See Appendix~\ref{proof:converse}.
\end{proof}


\begin{remark} \label{remark:est_then_dec}
	A practical approach to adapt the coding schemes proposed for the case of known $\rv{K}_{\rm a}$ to the setting where $\rv{K}_{\rm a}$ is unknown is to first estimate $\rv{K}_{\rm a}$ and then treat this estimate as the true $\rv{K}_{\rm a}$ in the decoding process. In this way, the number of decoded messages is equal to the estimate of $\rv{K}_{\rm a}$. 
	\revise{For this approach, our random-coding bound with zero decoding radius, provided in Corollary~\ref{coro:zero_decRad}, gives an achievability bound, whereas Theorem~\ref{th:converse} with $\overline{K_{\rm a}'} = \underline{K_{\rm a}'} = K_{\rm a}'$ and with $P_{\rv{x}}$ being the i.i.d. Gaussian codeword distribution with average power $P$ gives an ensemble converse bound.}
\end{remark}

\revisee{A general converse bound on the \gls{MD} and \gls{FA} probabilities seems difficult to obtain, as we discuss in Section~\ref{sec:discussion_converse}.} 

\subsection{Application to Slotted ALOHA With Multi-Packet Reception}
Our random-coding bound can also be applied to SA-MPR to investigate the resulting \gls{MD}--\gls{FA} trade-off. Consider an SA-MPR scheme where a length-$n$ frame is divided into $L$ slots, and where each user chooses randomly a slot in which to transmit. For $K_{\rm a}$ active users, the number of users transmitting in a slot follows a binomial distribution with parameters $(K_{\rm a},1/L)$. 
The \gls{PMF} of the number of active users in a slot, $\rv{K}_{\rm {SA}}$,  is given by
\begin{equation} \label{eq:PMF_Ka_perSlot}
	P_{\rv{K}_{\rm {SA}}}(K_{\rm SA}) = \sum_{K_{\rm a}} P_{\rv{K}_{\rm a}}(K_{\rm a}) \binom{K_{\rm a}}{K_{\rm SA}} L^{-K_{\rm SA}}\Big(1-\frac{1}{L}\Big)^{K_{\rm a} - K_{\rm SA}}.
\end{equation}
Existing analyses of slotted ALOHA usually assume that the decoder can detect perfectly if no user, one user, or more than one user have transmitted in a slot. Furthermore, it is usually assumed that a collision-free slot leads to successful message decoding. However, \revise{in the presence of noise, single-user decoding in a collision-free
	slot may fail. Furthermore, the more the slots, the shorter the slot length
	over which a user transmits its signal. In this case, the single-user code
	becomes shorter and thus less resilient to noise.} 
To account for both detection and decoding errors, in Corollary~\ref{coro:slotted_ALOHA} below we apply the decoder introduced in~\eqref{eq:decoder_Ka'} in a slot-by-slot manner, and obtain a random-coding bound similar to Theorem~\ref{thm:RCU_unknownKa}. 

\begin{corollary}[Random-coding bound for SA-MPR] \label{coro:slotted_ALOHA}
	For the Gaussian \gls{MAC} with the number of active users distributed according to $P_{\rv{K}_{\rm a}}$ and frame length $n$, an SA-MPR scheme with $L$ slots can achieve the \gls{MD} and \gls{FA} probabilities given in~\eqref{eq:eps_MD} and \eqref{eq:eps_FA}, respectively, with codebook size $M$, codeword length $n/L$, power constraint $PL$, and per-slot number of active users distributed according to $P_{\rv{K}_{\rm {SA}}}$ defined in~\eqref{eq:PMF_Ka_perSlot} (i.e., $P_{\rv{K}_{\rm a}}$ is replaced by $P_{\rv{K}_{\rm {SA}}}$ in \eqref{eq:eps_MD} and \eqref{eq:eps_FA}).
\end{corollary}

\subsection{Random-Coding Bound for a Scheme That Treats Interference as Noise} \label{sec:TIN}
We next present a random-coding union bound with parameter $s$ (RCUs) for \gls{TIN} to account for both \gls{MD} and \gls{FA} when the number of active users is unknown. This is an extension of the dependence testing (DT) bound for TIN with a known number of active users proposed in~\cite[Sec.~III]{PolyanskiyISIT2017massive_random_access}. We consider a random-coding scheme with the same encoder as in Theorem~\ref{thm:RCU_unknownKa} and a \gls{TIN} decoder. Specifically, the receiver first estimates the number of active users $\rv{K}_{\rm a}$ using the metric $m(\yv,K)$, as in the decoder analyzed in Theorem~\ref{thm:RCU_unknownKa}.
Then, it outputs a list composed of the $K_{\rm a}^{\prime}$ codewords that are closest to the received signal $\rvVec{y}$ in Euclidean distance, i.e.,
\begin{equation}
	\widehat{\mathcal{W}} = \arg\min_{\mathcal{W}' = \{{w}'_1,\dots,{w}'_{K'_{\rm a}}\}\subset [M]} \sum_{i=1}^{K_{\rm a}^{\prime}} \| \rvVec{y} - {\cv}_{{w}'_i} \|^2. \label{eq:NN_decoder}
\end{equation}
Operationally, the receiver decodes the message of each user by treating the signals of all other users as noise. 
An error analysis of this coding scheme gives the following result. 

\begin{theorem}[Random-coding bound for TIN] \label{thm:TIN}
	Fix $P' < P$ and nonnegative integers $K_\ell \le K_{u}$. For the $\rv{K}_{\rm a}$-user Gaussian \gls{MAC} with $\rv{K}_{\rm a} \sim P_{\rv{K}_{\rm a}}$, there exists an $(M,n,\epsilon_{\rm MD},\epsilon_{\rm FA})$ random-access code satisfying the power constraint $P$ for which, for every $s>0$, 
	\begin{align}
		\epsilon_{\rm MD} &= \sum_{K_{\rm a} =\max\{K_\ell,1\}}^{K_{u}} \sum_{K_{\rm a}^{\prime} = K_\ell}^{K_{u}} \! \frac{P_{\rv{K}_{\rm a}}(K_{\rm a})}{K_{\rm a}}  
		\Big( (K_{\rm a}\!-\!K_{\rm a}^{\prime})^+ \xi(K_{\rm a},K_{\rm a}') \nonumber\\ 		& \qquad 
		+ \min\{K_{\rm a},K_{\rm a}^{\prime}\}\min\left\{\eta(s),\xi(K_{\rm a},K_{\rm a}')\right\}\Big) 
        + \tilde{p},  \label{eq:md_TIN}\\ 
		\epsilon_{\rm FA} &= \sum_{K_{\rm a} =K_\ell}^{K_{u}} \sum_{K_{\rm a}^{\prime} = K_\ell}^{K_{u}} \frac{P_{\rv{K}_{\rm a}}(K_{\rm a})}{K_{\rm a}^{\prime}} 
		\Big((K_{\rm a}^{\prime}-K_{\rm a})^+\xi(K_{\rm a},K_{\rm a}')  \nonumber\\ & \qquad
		+ \min\{K_{\rm a},K_{\rm a}^{\prime}\} \min\left\{\eta(s),\xi(K_{\rm a},K_{\rm a}')\right\}\Big) 
        + \tilde{p}, \label{eq:fa_TIN}
	\end{align}
	where $\tilde{p}$ is given in~\eqref{eq:p0}, $\xi(K_{\rm a},K_{\rm a}')$ is given in~\eqref{eq:xi}, and 
	\begin{align}
		\eta(s) = \P[\sum\limits_{i=1}^n  \imath_s(\rv{x}_{i};\rv{y}_i) \le \ln\left(\frac{M-K_{\rm a}}{\rv{u}}\right)].
		\label{eq:eta_s}
	\end{align}
	Here, $\rv{u}$ is uniformly distributed on $[0,1]$ and $[\rv{x}_1 \dots \rv{x}_n]^\T \sim \Cc\Nc(\mathbf{0},\bar{P} \Id_n)$ with $\bar{P} = P'/(1+(K_{\rm a}-1)P')$. Furthermore, given $\rv{x}_i = x_i$, we have that $\rv{y}_i \sim \Cc\Nc(x_i,1)$. Finally, $\imath_s(\cdot;\cdot)$ is the generalized information density, given by
	\begin{align} \label{eq:gen_infor_den}
		\imath_s(x;y) \defeq -s|y-x|^2 + \frac{s|y|^2}{1+s\bar{P}} + \ln(1+s\bar{P}).
	\end{align}
	
\end{theorem}
\begin{proof}
	See Appendix~\ref{app:TIN}.
\end{proof}

When $\rv{K}_{\rm a}$ is fixed to $K_{\rm a}$ and $K'_{\rm a} = K_{\rm a}$, by setting $K_\ell = K_u = K_{\rm a}$ and $s = 1$, we obtain from Theorem~\ref{thm:TIN} that
\begin{equation}
\epsilon_{\rm MD} = \epsilon_{\rm FA} = \eta(1) + 1 - \frac{M!}{M^{{K}_{\rm a}}(M-{K}_{\rm a})!} + {K}_{\rm a}  \frac{\Gamma(n,nP/P')}{\Gamma(n)}.
\end{equation}
This bound is similar to the one presented in \cite[Sec.~III]{PolyanskiyISIT2017massive_random_access}. 
To simplify the computation of $\eta(s)$, we can use the normal approximation given in \cite[Eq. (7)]{Ostman20}. Specifically, the term $\eta(s)$ can be expanded as $\eta(s) = Q\big(\frac{nI_s-\ln(M-K_{\rm a})}{\sqrt{nV_s}}\big) + O\big(\frac{1}{\sqrt{n}}\big)$
where $	I_s \defeq \mathbb{E}\left[\imath_s(\rv{x}_{i};\rv{y}_i)\right]$, $V_s \defeq \mathbb{E}\left[\vert \imath_s(\rv{x}_{i};\rv{y}_i) - I_s \rvert^2\right]$, and $Q(\cdot)$ is the Gaussian $Q$ function.
Note that, by choosing $s=1$, we obtain the normal approximation for the AWGN channel with Gaussian input, which is capacity achieving but not optimal from a dispersion perspective~\cite{Scarlett2017}. 
It then follows that 
$\eta(1)= Q\big(\frac{nC-\ln(M-K_{\rm a})}{\sqrt{nV}}\big) + O\big(\frac{1}{\sqrt{n}}\big)$, 
where $C \defeq I_1 = \ln\big(1+\frac{P'}{1+(K_{\rm a}^{\prime}-1)P'}\big)$ and $V \defeq V_1 = \frac{2 P'}{1+K_{\rm a}^{\prime}P'}$.

		\section{Numerical Experiments} \label{sec:numerical}
		We numerically evaluate the proposed random-coding bound and compare it with the random-coding bound in~\cite{PolyanskiyISIT2017massive_random_access} and with the performance of some unsourced multiple access schemes, namely, SA-MPR, TIN, and the schemes in~\cite{Fengler2019sparcs,Amalladinne2021}. We assume that $\rv{K}_{\rm a} \sim \mathrm{Pois}(\E[\rv{K}_{\rm a}])$ \revise{and the distribution of $\rv{K}_{\rm a}$ is known by the receiver}. Poisson processes are commonly used to model message arrivals in packet-data networks and result in accurate models in many realistic IoT scenarios~\cite{Metzger2019}. 
		We assume an information payload per user of $k \defeq \log_2 M = 128$ bits and consider a transmission duration of $n = 19200$ complex channel uses (i.e., $38400$ real degrees of freedom). As performance metric, we consider the average energy per bit $\EbNo \defeq nP/k$ needed to operate at given target \gls{MD} and \gls{FA} probabilities.

		In Fig.~\ref{fig:EbN0_vs_EKa}, we compare the random-coding bound provided in Theorem~\ref{thm:RCU_unknownKa} with that in~\cite[Th.~1]{PolyanskiyISIT2017massive_random_access} in terms of the required $\EbNo$ so that neither $P_{\rm MD}$ nor $P_{\rm FA}$ exceed $10^{-1}$. For our bound, we consider the \gls{ML} estimation of $\rv{K}_{\rm a}$ and set the decoding radius $r = 0$, i.e., $\overline{K_{\rm a}'} = \underline{K_{\rm a}'} = K_{\rm a}'$. Numerical experiments indicate that this choice of decoding radius is optimal for the targeted MD and FA probabilities.\footnote{We discuss the trade-offs resulting from the choice of the decoding radius later on in this section.} We choose $K_\ell$ to be the largest value and $K_u$ the smallest value for which $\P[{\rv{K}_{\rm a} \notin [K_\ell,K_u]}] < 10^{-9}$. Due to computational complexity, the terms $q_t$ and $q_{t,t'}$ are evaluated for $t = 1$ and $K_{\rm a} \le 50$ only. For the bound in \cite[Th.~1]{PolyanskiyISIT2017massive_random_access}, we average over the Poisson distribution of $\rv{K}_{\rm a}$. This corresponds to the scenario in which $\rv{K}_{\rm a}$ is random but known to the receiver. 
		As can be seen, similar to the bound in~\cite[Th.~1]{PolyanskiyISIT2017massive_random_access}, our bound exhibits two behaviors as a function of the average number of active users: one where the required $\EbNo$ is almost constant and finite-blocklength effect dominates, and one where it grows with the average number of active users and multi-user interference dominates. 
		To achieve $\max\{P_{\rm MD}, P_{\rm FA}\} \le 10^{-1}$, the extra required $\EbNo$ due to the lack of knowledge of $\rv{K}_{\rm a}$ is about $0.5$--$0.7$~dB. 
		This small gap suggests that for the mild target $\max\{P_{\rm MD}, P_{\rm FA}\} \le 10^{-1}$, lack of knowledge of the number of active users entails a small penalty in terms of energy efficiency. 
		
		In Fig.~\ref{fig:EbN0_vs_EKa}, we also show the performance of the SA-MPR bound given in Corollary~\ref{coro:slotted_ALOHA} where we 
		optimize $L$ and the decoding radius for each $\E[\rv{K}_{\rm a}]$. We further consider the possibility to encode~$\lfloor \log_2 L\rfloor$ extra bits for each user in the slot index, and assume perfect decoding of these bits. We refer to this scheme as SA-MPR with slot-index coding. Furthermore, we plot the RCUs bound on the performance of \gls{TIN} given in Theorem~\ref{thm:TIN} optimized over $s$, and its normal approximation evaluated with $s = 1$. 
		We also evaluate the performance of two state-of-the-art practical schemes, namely:
		\begin{itemize}
			
			\item the SPARC scheme proposed in~\cite{Fengler2019sparcs}, which employs a concatenated coding structure, consisting of an inner approximate message passing~(AMP) decoder followed by an outer tree decoder;
			
			\item an enhancement, proposed in~\cite{Amalladinne2021}, of the SPARC scheme in~\cite{Fengler2019sparcs}, where belief propagation between the inner AMP decoder and the outer tree decoder is introduced. 
		\end{itemize}
		Note that 
		the SPARC and enhanced SPARC schemes were proposed for the Gaussian \gls{MAC} with {known} number of active users. To adapt these schemes to the case of unknown $\rv{K}_{\rm a}$, we follow the approach discussed in Remark~\ref{remark:est_then_dec}: we first employ an energy-based estimation of $\rv{K}_{\rm a}$, and then treat this estimate as the true $\rv{K}_{\rm a}$ in the decoding process. As shown in Fig.~\ref{fig:EbN0_vs_EKa}, \gls{TIN} (light blue) and SA-MPR (pink), even with slot-index coding, become energy inefficient as $\E[\rv{K}_{\rm a}]$ increases. 
		The enhanced SPARC scheme (orange) achieves the closest performance to our bound for $\E[\rv{K}_{\rm a}] \ge 100$ and outperforms the original SPARC scheme by about $0.5$~dB for large $\E[\rv{K}_{\rm a}]$. The performance of enhanced SPARC with known $\rv{K}_{\rm a}$  is only slightly better than that with unknown $\rv{K}_{\rm a}$. This suggests that, for the considered mild requirements on $P_{\rm MD}$ and $P_{\rm FA}$, it is effective to adapt existing coding schemes that perform well for known $\rv{K}_{\rm a}$ by simply adding a $\rv{K}_{\rm a}$-estimation step, and then treating the estimate of $\rv{K}_{\rm a}$ as the true value.
		
		\begin{figure*}[t!]
			\centering
			\input{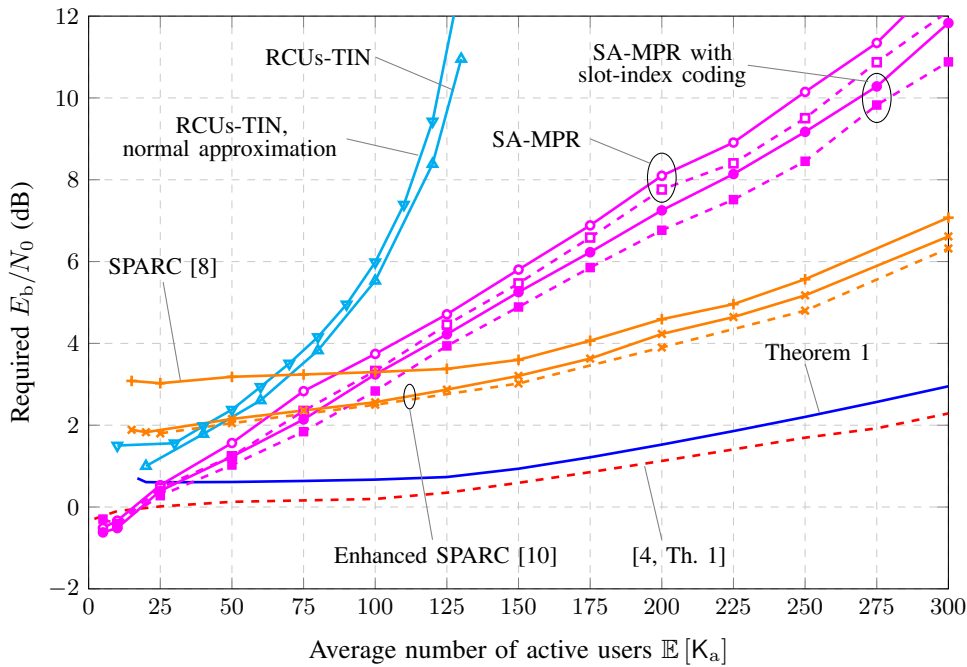}
			\caption{The required $\EbNo$ to achieve $\max\{P_{\rm MD}, P_{\rm FA}\} \le 10^{-1}$ as a function of $\E[\rv{K}_{\rm a}]$ for $k=128$ bits, $n = 19200$ channel uses, and $\rv{K}_{\rm a} \sim \mathrm{Pois}(\E[\rv{K}_{\rm a}])$. 
				Solid lines represent schemes/bounds with $\rv{K}_{\rm a}$ unknown; dashed lines represent schemes/bounds with $\rv{K}_{\rm a}$ known.}
			\label{fig:EbN0_vs_EKa}
		\end{figure*}
		
		In Fig.~\ref{fig:Pe_vs_EbN0}, we plot the bounds on $P_{\rm MD}$ and $P_{\rm FA}$ in Theorem~\ref{thm:RCU_unknownKa} as a function of $\EbNo$ for $\E[\rv{K}_{\rm a}] \in \{50,200\}$, and different decoding radius $r$. We also show the bound in~\cite[Th.~1]{PolyanskiyISIT2017massive_random_access} for $\rv{K}_{\rm a}$ known. 
		The \gls{MD} and \gls{FA} probabilities drop at a certain $\EbNo$, and then saturate to error floors. The existence of a waterfall region and an error-floor region is more evident in the case $\E[\rv{K}_{\rm a}] = 200$. The error floors are due to the initial \glspl{MD} and \glspl{FA}, and are characterized by $\bar{\epsilon}_{\rm MD}$ and $\bar{\epsilon}_{\rm FA}$ in \revisee{Corollary}~\ref{cor:error_floor}. \revise{The ensemble converse in Theorem~\ref{th:converse} computed for the i.i.d. Gaussian ensemble with average power $P$ and for the considered decoding radii is also depicted. We observe that the converses for the MD and FA probabilities are similar, and follow closely our achievability bounds after the waterfall. This explains the waterfall effect: the additional MDs and FAs occurring in the decoding process dominate when the $\EbNo$ is small, while the initial MDs and FAs dominate in the error-floor regime.}
		We further observe that setting $r =0$  leads to $\epsilon_{\rm MD} = \epsilon_{\rm FA}$ and  turns out optimal in the low $\EbNo$ regime, where noise overfitting is the  bottleneck. 
		This is in agreement with the results reported in Fig.~\ref{fig:EbN0_vs_EKa}. An increase of the decoding radius results in a performance improvement in the moderate and high $\EbNo$ regime, where setting $r =0$ yields high error floors. 
		Our numerical results indicate that zero decoding radius should be used when the target \gls{MD} and \gls{FA} probabilities are higher than about \revise{$8\times 10^{-3}$ for $\E[\rv{K}_{\rm a}] = 50$ and $6.5\times 10^{-3}$ for $\E[\rv{K}_{\rm a}] = 200$}. \revise{Otherwise, if the target  \gls{MD} and \gls{FA} probabilities are low, one should increase the decoding radius to lower the error floors and thus meet the requirements in the waterfall regime.} Remarkably, as illustrated in Fig.~\ref{fig:Pe_vs_EbN0}, there is a sharp difference between the case $r = 0$ and the case $r > 0$ in terms of the $\EbNo$ value at which the waterfall region starts. Moreover, this $\EbNo$ value appears to coincide for all positive~$r$.\footnote{Characterizing analytically this $\EbNo$ value is challenging since $\epsilon_{\rm MD}$ and $\epsilon_{\rm FA}$ are given in terms of sums in which each addend involves the optimization of the three intermediate parameters $(\lambda, \rho, \rho_1)$.}

		\begin{figure*}[t!]
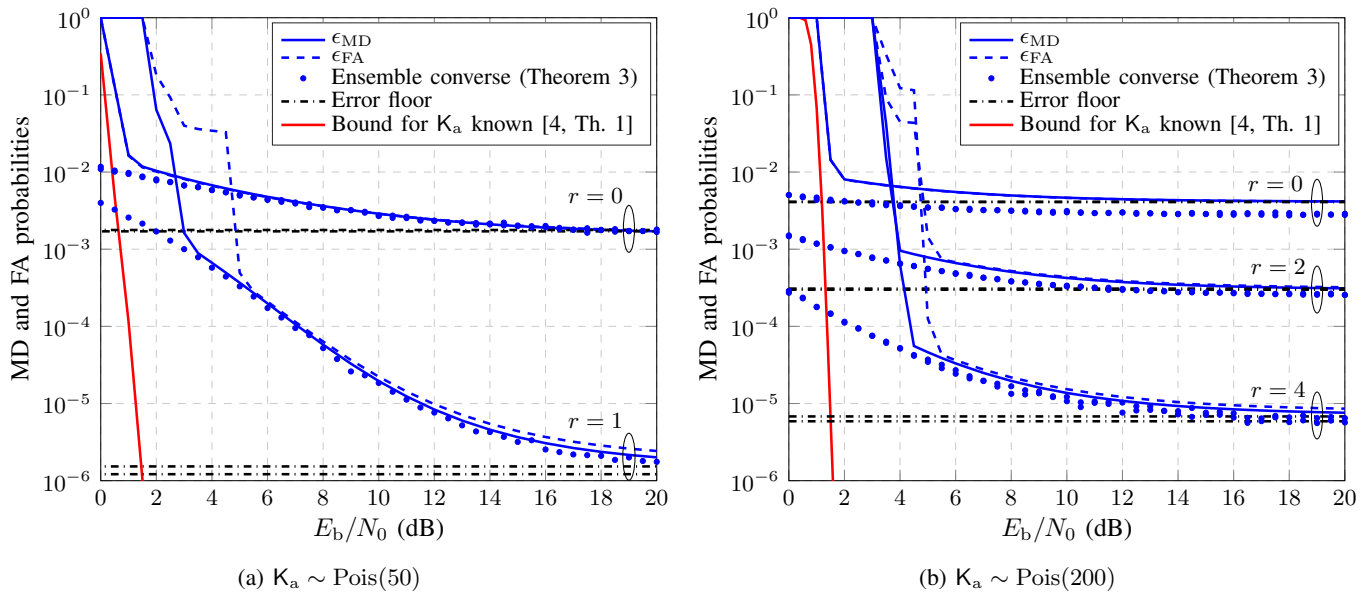

			\centering
            \begin{subfigure}[t]{.48\textwidth}
                \input{fig/Pe_vs_EbN0_journal_k128_EKa50.tex}
                \caption{$\rv{K}_{\rm a} \sim \mathrm{Pois}(50)$}
            \end{subfigure}
            \hspace{.2cm}
            \begin{subfigure}[t]{.48\textwidth}
                \input{fig/Pe_vs_EbN0_journal_k128_EKa200.tex}
                \caption{$\rv{K}_{\rm a} \sim \mathrm{Pois}(200)$}
            \end{subfigure}
			\caption{The bounds~\eqref{eq:eps_MD} and~\eqref{eq:eps_FA} on the \gls{MD} and \gls{FA} probabilities as functions of $\EbNo$ for $k=128$ bits, $n = 19200$ channel uses, and some values of $r$, as well as their large-$P$ limits given in~\eqref{eq:eps_MD_floor} and~\eqref{eq:eps_FA_floor}, \revise{and the ensemble converse given in Theorem~\ref{th:converse}.} 
			}
			\label{fig:Pe_vs_EbN0}
		\end{figure*}
		
		When asymmetric requirements on $P_{\rm MD}$ and $P_{\rm FA}$ are considered, one can adjust the decoding interval $[\underline{K_{\rm a}'}: \overline{K_{\rm a}'}]$ in~\eqref{eq:decoder_Ka'} to achieve different \gls{MD}-\gls{FA} trade-offs. Specifically, one can set $\underline{K_{\rm a}'} =\max\{K_\ell,K_{\rm a}'-r_\ell\}$ and $\overline{K_{\rm a}'}\defeq \min\{K_u,K_{\rm a}'+r_u\}$ where $(r_\ell,r_u)$ is a pair of lower and upper decoding radii. The random-coding bound in Theorem~\ref{thm:RCU_unknownKa} can be extended accordingly. In Fig.~\ref{fig:EbN0_vs_EKa_2}, we consider this extension and show the required $\EbNo$ to achieve more stringent requirements on $P_{\rm MD}$ and $P_{\rm FA}$, namely, $\{P_{\rm MD} \le 10^{-1}, P_{\rm FA} \le 10^{-3}\}$, $\{P_{\rm MD} \le  10^{-3}, P_{\rm FA} \le 10^{-1}\}$, and $\max\{P_{\rm MD}, P_{\rm FA}\} \le  10^{-3}$. We increase $r_\ell$ and $r_u$ such that the error floors lie below these target \gls{MD} and \gls{FA} probabilities, and no significant reduction in $\EbNo$ can be achieved by increasing $r_\ell$ and $r_u$ further. For example, for $\{P_{\rm MD} \le 10^{-1}, P_{\rm FA} \le 10^{-3}\}$, $\{P_{\rm MD} \le  10^{-3}, P_{\rm FA} \le 10^{-1}\}$, and $\max\{P_{\rm MD}, P_{\rm FA}\} \le  10^{-3}$, we respectively set $(r_\ell,r_u) = (\revise{2},0)$, $(0,\revise{2})$, and $(\revise{2,2})$ for $\E[\rv{K}_{\rm a}] = 50$, and $(r_\ell,r_u) = (6,1)$, $(1,6)$, and $(6,6)$ for $\E[\rv{K}_{\rm a}] = 200$. We observe that the required $\EbNo$ is almost the same for the two requirements with $P_{\rm FA} \le 10^{-3}$, and it is slightly lower for the requirement with $P_{\rm FA} \le 10^{-1}$. We also plot the bound for $\rv{K}_{\rm a}$ known~\cite[Th.~1]{PolyanskiyISIT2017massive_random_access}. In contrast to the scenario in Fig.~\ref{fig:EbN0_vs_EKa}, where the gap between our bound and the bound for $\rv{K}_{\rm a}$ known was small, the gap is now significantly larger. Specifically, our bound suggests that one needs additional $3$--$4$~dB in $\EbNo$ when $P_{\rm MD}$ and/or $P_{\rm FA}$ are required to be lower than $10^{-3}$ \revise{when our random-coding scheme is employed}. This gap tends to be smaller as $\E[\rv{K}_{\rm a}]$ increases. \revise{It remains unclear whether this gap is fundamental.} The slope of the curves for the bound in Theorem~\ref{thm:RCU_unknownKa} suggests that the dominating factor for $\E[\rv{K}_{\rm a}] \le 300$ is still the finite-blocklength effect, and the multi-user interference will kick in for higher $\E[\rv{K}_{\rm a}]$.

        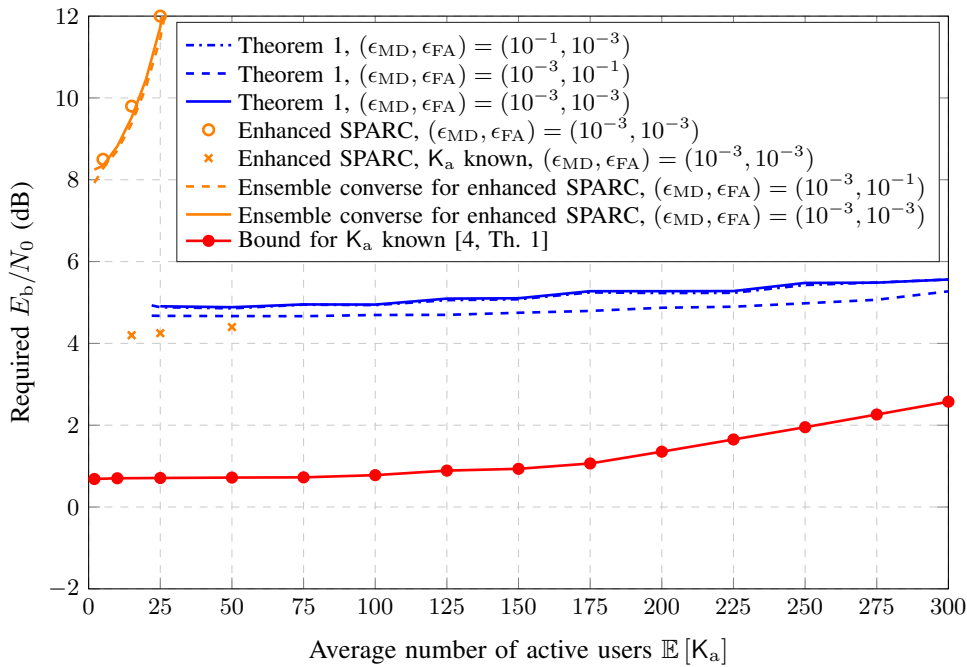
\begin{figure*}[t!]
			\centering
			\definecolor{magenta}{rgb}{1.00000,0.00000,1.00000}%
	\begin{tikzpicture}
		\tikzstyle{every node}=[font=\small]
		\begin{axis}[%
			width=4.5in,
			height=3in,
			at={(0.759in,0.481in)},
			scale only axis,
			xmin=0,
			xmax=300,
            xtick={0, 25, 50, 75, 100, 125, 150, 175,200,225,250,275,300},
			xlabel style={font=\color{black},yshift=0ex},
			xlabel={Average number of active users $\E[\rv{K}_{\rm a}]$},
			ymin=-2,
			ymax=12,
			ytick={-2, 0, 2, 4, 6, 8, 10, 12},
			ylabel style={font=\color{black}, yshift=-2.2ex},
			ylabel={Required $\EbNo$ (dB)},
			axis background/.style={fill=white},
			xmajorgrids,
			ymajorgrids,
			legend style={at={(0.99,0.99)}, anchor=north east, row sep=-2.5pt, legend cell align=left, align=left, draw=white!15!black}
			]
			
			\addplot [color=blue,line width=1pt,dashdotted]
			table[row sep=crcr]{%
				22 4.928994778915966 \\
				25 4.883252396304738 \\
				50 4.856140339766785 \\ 
				75 4.946079896204276 \\ 
				100 4.935671561425632 \\ 
				125 5.056421838879606 \\
				150 5.077527854338077 \\
				175 5.245611704881298 \\ 
				200 5.230443873021846 \\ 
				225 5.235505705256200 \\ 
				250 5.426838602549897 \\
				275 5.487727168692260 \\
				300 5.561874664391777 \\
			};
					\addlegendentry{Theorem~\ref{thm:RCU_unknownKa}, $(\epsilon_{\rm MD},\epsilon_{\rm FA}) = (10^{-1},10^{-3})$}
			
			\addplot [color=blue,line width=1pt,dashed]
			table[row sep=crcr]{%
				22 4.676932132301527 \\
				25 4.672843764355430 \\
				50 4.667531759234577 \\ 
				75 4.664903366173185 \\ 
				100 4.694377737914646 \\ 
				125 4.697463931851004 \\
				150 4.748431792905569 \\
				175 4.793654082315352 \\ 
				200 4.871127632415096 \\ 
				225 4.896526972534520 \\ 
				250 4.979985341559578 \\ 
				275 5.066310360305302 \\
				300 5.272446628963706 \\
			};
			\addlegendentry{Theorem~\ref{thm:RCU_unknownKa}, $(\epsilon_{\rm MD},\epsilon_{\rm FA}) = (10^{-3},10^{-1})$}
		
			\addplot [color=blue,line width=1pt]
			table[row sep=crcr]{%
				25 4.906919347922863 \\
				50 4.884570615982140 \\ 
				75 4.952572460938979 \\ 
				100 4.948678086955543 \\ 
				125 5.094263011065712 \\
				150 5.103680413028469 \\
				175 5.275789365752891 \\ 
				200 5.275789365752891 \\ 
				225 5.279129531637679 \\ 
				250 5.478169756594423 \\ 
				275 5.484543701061485 \\
				300 5.561874664391777 \\
			};
			\addlegendentry{Theorem~\ref{thm:RCU_unknownKa}, $(\epsilon_{\rm MD},\epsilon_{\rm FA}) = (10^{-3},10^{-3})$}
			
			\addplot [color=orange, only marks, mark = o, line width=1pt]
			table[row sep=crcr]{%
				5  8.5 \\   
				15 9.8 \\
				25  12 \\ 
			};
			\addlegendentry{Enhanced SPARC, $(\epsilon_{\rm MD},\epsilon_{\rm FA}) = (10^{-3},10^{-3})$}
			
			\addplot [color=orange, only marks, mark = x, line width=1pt]
			table[row sep=crcr]{%
				15 4.2 \\
				25  4.25 \\
				50 4.4 \\ 
			};
			\addlegendentry{Enhanced SPARC, $\rv{K}_{\rm a}$ known, $(\epsilon_{\rm MD},\epsilon_{\rm FA}) = (10^{-3},10^{-3})$}
			
			\addplot [color=orange, dashed, line width=1pt]
			table[row sep=crcr]{%
				2  7.9215 \\
				5  8.2763 \\   
				10 8.7142 \\
				15 9.3732 \\
				20 10.2603 \\   
				25  11.4632 \\ 
				30 13.2158 \\
			};
			\addlegendentry{\revisee{Ensemble converse} for enhanced SPARC, $(\epsilon_{\rm MD},\epsilon_{\rm FA}) = (10^{-3},10^{-1})$}
			
			\addplot [color=orange,line width=1pt]
			table[row sep=crcr]{%
				2  8.2474 \\
				5  8.3430 \\   
				10 8.8171 \\
				15 9.5198 \\
				20 10.4506 \\   
				25 11.7312 \\ 
				30 13.6187 \\
			};
			\addlegendentry{\revisee{Ensemble converse} for enhanced SPARC, $(\epsilon_{\rm MD},\epsilon_{\rm FA}) = (10^{-3},10^{-3})$}

			\addplot [color=red, line width=1pt, mark=*,mark size = 1.8pt]
			table[row sep=crcr]{%
				2	0.685587245114299 \\
				10  0.702936368329330 \\
				25	0.709835116917573 \\ 
				50	0.720162735783853 \\ 
				75	0.725317350053997 \\ 
				100	0.779180613236050 \\ 
				125	0.890422968939304 \\ 
				150	0.934011404226458 \\ 
				175	1.064594603797970 \\ 
				200	1.352606600228537 \\ 
				225	1.651410782545691 \\ 
				250	1.951414699988889 \\ 
				275	2.261211717237971 \\ 
				300	2.575119173921741 \\
			};
					\addlegendentry{Bound for $\rv{K}_{\rm a}$ known~\cite[Th.~1]{PolyanskiyISIT2017massive_random_access}}

			
			

			
			
		\end{axis}
	\end{tikzpicture}%
			\caption{The required $\EbNo$ to achieve $\{P_{\rm MD} \le 10^{-1}, P_{\rm FA} \le  10^{-3}\}$, $\{P_{\rm MD} \le  10^{-3}, P_{\rm FA} \le 10^{-1}\}$, or $\max\{P_{\rm MD}, P_{\rm FA}\} \le  10^{-3}$
				as a function of $\E[\rv{K}_{\rm a}]$ for $k=128$ bits, $n = 19200$ channel uses, and $\rv{K}_{\rm a} \sim \mathrm{Pois}(\E[\rv{K}_{\rm a}])$. 
			}
			\label{fig:EbN0_vs_EKa_2}
		\end{figure*}
        \begin{figure*}[t!]
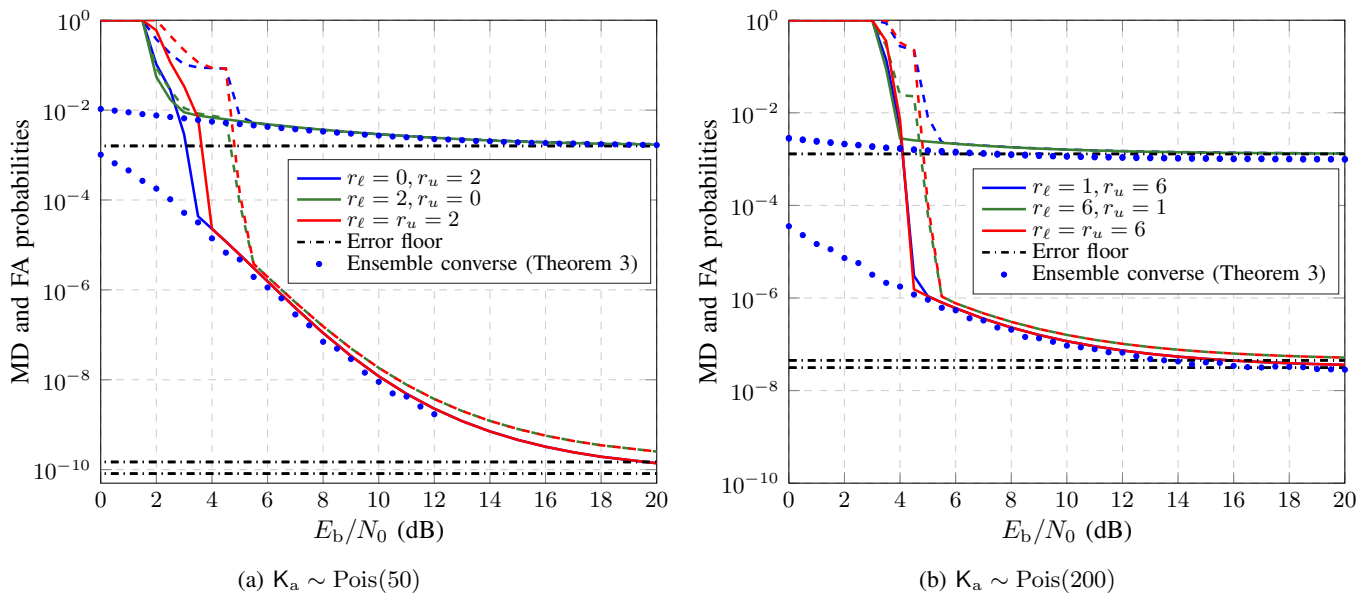

			\centering
            \begin{subfigure}[t]{.48\textwidth}
                \input{fig/Pe_vs_EbN0_journal_k128_EKa50_2.tex}
                \caption{$\rv{K}_{\rm a} \sim \mathrm{Pois}(50)$}
            \end{subfigure}
            \hspace{.2cm}
			\begin{subfigure}[t]{.48\textwidth}
                \input{fig/Pe_vs_EbN0_journal_k128_EKa200_2.tex}
                \caption{$\rv{K}_{\rm a} \sim \mathrm{Pois}(200)$}
            \end{subfigure}
			\caption{The bounds on the \gls{MD} and \gls{FA} probabilities as a function of $\EbNo$ for $k=128$ bits and $n = 19200$ channel uses with $r_\ell$ and $r_u$ selected corresponding to Fig.~\ref{fig:EbN0_vs_EKa_2}. 
				Solid lines represent $\epsilon_{\rm MD}$; dashed lines represent $\epsilon_{\rm FA}$. 
			}
			\label{fig:Pe_vs_EbN0_2}
		\end{figure*}
		In Fig.~\ref{fig:EbN0_vs_EKa_2}, we also show the required $\EbNo$ for the enhanced SPARC scheme with $\rv{K}_{\rm a}$ known/unknown. We further plot a lower bound on this $\EbNo$ value resulting from \revise{evaluating the right-hand sides of of~\eqref{eq:converse_MD} and~\eqref{eq:converse_FA} for the distribution on the received signal $\rvVec{y}$ induced by the SPARC codebook and for $r=0$.\footnote{\revise{Since our ensemble converse bound is not tight in the waterfall regime, the resulting lower bound on the required $\EbNo$ with increased decoding radius is loose and thus not plotted in Fig.~3.}} Recall that the lower bound is obtained by counting the \glspl{MD} and \glspl{FA} that occur in the estimation step only, while assuming that no additional \gls{MD} or \gls{FA} occurs in the decoding step.}
		Recall \revise{also} that we adapt the enhanced SPARC scheme to the case of unknown $\rv{K}_{\rm a}$ by simply treating the estimate of $\rv{K}_{\rm a}$ as the true $\rv{K}_{\rm a}$. The performance of the enhanced SPARC scheme with $\rv{K}_{\rm a}$ unknown is very close to the lower bound and drastically worse than enhanced SPARC with $\rv{K}_{\rm a}$ known. This confirms that the estimation step is indeed the bottleneck. 
		Furthermore, the lower bound for enhanced SPARC exhibits a large gap to the achievability bound in Theorem~\ref{thm:RCU_unknownKa}. For example, the gap is about $7$~dB for only $25$ active users in average. This large gap suggests that this approach, which simply uses the estimate of $\rv{K}_{\rm a}$ to set the decoded list size and relies on existing coding schemes proposed for $\rv{K}_{\rm a}$ known, becomes energy inefficient for stringent requirements on $P_{\rm MD}$ and/or $P_{\rm FA}$. This calls for more sophisticated methods to overcome the bottleneck of estimating $\rv{K}_{\rm a}$ and handle effectively the uncertainty about the number of active users.
		
		In Fig.~\ref{fig:Pe_vs_EbN0_2}, we plot the bounds $\epsilon_{\rm MD}$ and $\epsilon_{\rm FA}$ in Theorem~\ref{thm:RCU_unknownKa} as a function of $\EbNo$ for $\E[\rv{K}_{\rm a}] \in \{50,200\}$ for the decoding radii considered in Fig.~\ref{fig:EbN0_vs_EKa_2}. Solid lines represent $\epsilon_{\rm MD}$, while dashed lines represent $\epsilon_{\rm FA}$. We observe again that the waterfall region of either $\epsilon_{\rm MD}$ or $\epsilon_{\rm FA}$ starts at a similar $\EbNo$ value for various values of $(r_\ell,r_u)$ different from $(0,0)$. \revise{We also observe that, after the waterfall, our achievability bounds approach closely the ensemble converse in Theorem~\ref{th:converse}.} The $\EbNo$ value that satisfies the requirements on both $P_{\rm MD}$ and $P_{\rm FA}$ is dictated by the $\EbNo$ value that satisfy the requirement on $P_{\rm FA}$.
		This explains why the $\EbNo$ values according to Theorem~\ref{thm:RCU_unknownKa}, shown in Fig.~\ref{fig:EbN0_vs_EKa_2}, is similar for the two requirements with $P_{\rm FA} \le 10^{-3}$. 

		\section{Discussion} \label{sec:discussion}
		In this section, we provide some additional remarks on our choice of using a two-step decoder in our random-coding achievability bound, and the challenges involved in obtaining a general converse bound. 
  
		\subsection{The Two-Step Decoder} \label{sec:discussion_decoder}
		An alternative to the proposed two-step decoder, which first estimates the number of active users and then the list of messages, is a joint decoder that estimates both at the same time. Such a joint decoder operates according to the following rule 
		\begin{equation} \label{eq:joint_decoder}
			\widehat{\Wc} = \arg\min_{{\Wc}' \subset [M]: K_\ell \le |{\Wc}'| \le K_u} \|c({\Wc}') - \rvVec{y}\|^2
		\end{equation}   
        where the limits $K_\ell$ and $K_u$ are chosen based on prior knowledge on the distribution of $\rv{K}_{\rm a}$. Note that this decoder is a special case of our two-step decoder when i) the $\rv{K}_{\rm a}$-estimation step~\eqref{eq:est_Ka} is skipped, and ii) $\underline{K_{\rm a}'} = K_\ell$ and $\overline{K_{\rm a}'} = K_u$ in the message-decoding step~\eqref{eq:decoder_Ka'}. Thus, a random-coding bound for the joint decoder~\eqref{eq:joint_decoder} follows directly from Theorem~\ref{thm:RCU_unknownKa}. This bound is stated in the following corollary.
        \begin{corollary}[Random-coding bound for the joint decoder] \label{coro:RCU_joint_decoder}
            Fix $P' < P$, $K_\ell$, and $K_{u}$ ($K_\ell \le K_{u}$). For the $\rv{K}_{\rm a}$-user Gaussian \gls{MAC} with $\rv{K}_{\rm a} \sim P_{\rv{K}_{\rm a}}$, there exists an $(M,n,\epsilon_{\rm MD},\epsilon_{\rm FA})$ random-access code satisfying the power constraint $P$ for which 
    		\begin{align}
    			\epsilon_{\rm MD} &= \sum_{K_{\rm a} =\max\{K_\ell,1\}}^{K_{u}} P_{\rv{K}_{\rm a}}(K_{\rm a}) \sum_{t = 0}^{K_{\rm a}} \frac{t}{K_{\rm a}} 
    			\min\{p_t,q_t\}  + \tilde{p}, \label{eq:eps_MD_jointDec}\\
    			\epsilon_{\rm FA} &= \sum_{K_{\rm a} = K_\ell}^{K_{u}} P_{\rv{K}_{\rm a}}(K_{\rm a}) \sum_{t = 0}^{K_{\rm a}} \sum_{t' \in \Tc_t}  
    			\frac{t'}{K_{\rm a} - t + t'} \min\{p_{t,t'}, q_{t,t'}\}  \notag \\
                &\quad + \tilde{p}, \label{eq:eps_FA_jointDec}
    		\end{align}	
            where $\tilde{p}$, $p_t$, $p_{t,t'}$, $q_t$, and  $q_{t,t'}$ are given by~\eqref{eq:p0}, \eqref{eq:pt}, \eqref{eq:ptt}, \eqref{eq:qt}, and~\eqref{eq:qtt}, respectively, with $\underline{K_{\rm a}'} = K_\ell$ and $\overline{K_{\rm a}'} = K_u$.
        \end{corollary}
        
        Unfortunately, the bound in Corollary~\ref{coro:RCU_joint_decoder} results in a low energy efficiency. For example, for a similar setting as in Fig.~\ref{fig:EbN0_vs_EKa}, i.e., $k = 128$ bits, $n = 19200$ channel uses, and $\rv{K}_{\rm a} \sim {\rm Pois}(\E[\rv{K}_{\rm a}])$, to achieve $\max\{P_{\rm MD}, P_{\rm FA}\} \le 10^{-1}$, the required $\EbNo$ for the joint decoder is between $4.5$~dB and $5$~dB for $\E[\rv{K}_{\rm a}] \in [25:300]$. To understand the drawback of this bound, we shall now inspect the terms $p_t$ and $p_{t,t'}$. 
        Recall that $p_t = \sum_{t' \in \overline{\Tc}_t} p_{t,t'}$ and $p_{t,t'} = e^{-n E(t,t')}$ with the error exponent $E(t,t') = -\rho\rho_1 t' R_1 - \rho_1 R_2 + E_0(\rho,\rho_1)$. Here, $E_0(\rho,\rho_1)$ stems from the Chernoff bound on the probability of the pairwise error event $\|c({\Wc}') - \rvVec{y}\|^2 \le \|c({\Wc}) - \rvVec{y}\|^2$, while $R_1$ and $R_2$ stem from a tightened union bound over all possible sets of $t'$ falsely alarmed messages and $t$ misdetected messages, respectively. These terms scale differently with $t'$, making it nontrivial to understand how $p_{t,t'}$ varies with $t'$. Specifically, for fixed $t,P'$, and $M$, the term $-\rho_1R_2$ is a constant, $E_0(\rho,\rho_1)$ increases logarithmically with $t'$, and $-\rho \rho_1 t' R_1$ decreases linearly with~$t'$.

        \begin{figure*}
            \centering
                \input{fig/jointDec.tex}
            \caption{The bounds $p_{t,t'}$ in~\eqref{eq:ptt} and $p_{t}$ in~\eqref{eq:pt} applied to the joint decoder~\eqref{eq:joint_decoder} for $k = 128$ bits, $n = 19200$ channel uses, $K_{\rm a} = \E[\rv{K}_{\rm a}] = 50$, and $\EbNo = 2$~dB. Here, $K_\ell = (\E[\rv{K}_{\rm a}] - r)^+$ and $K_u = \E[\rv{K}_{\rm a}] + r$; $r_0$ is the largest value of $r$ such that $\P[{\rv{K}_{\rm a} \notin [K_\ell:K_u]}] < 10^{-9}$ with $\rv{K}_{\rm a} \sim {\rm Pois}(50)$.}
            \label{fig:jointDec}
        \end{figure*}

        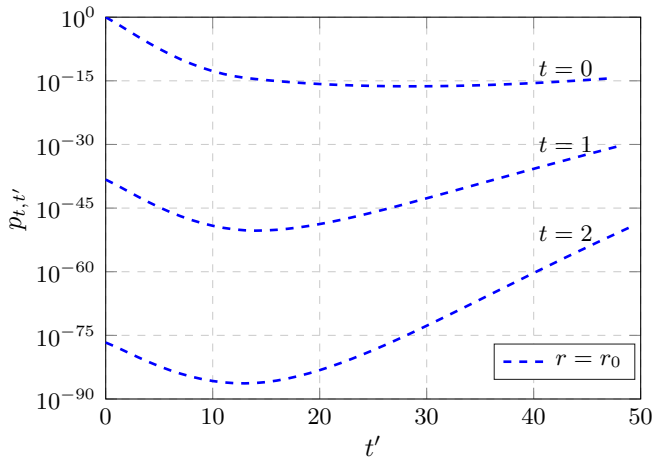
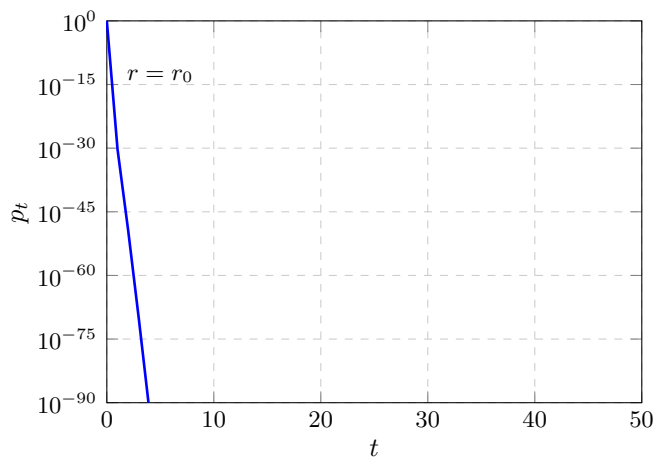
\begin{figure*}
                \begin{subfigure}[t]{.48\textwidth}
            \centering
            \begin{tikzpicture}
        		\tikzstyle{every node}=[font=\small]
        		\begin{axis}[%
        			width=2.8in,
        			height=2in,
        			at={(0.759in,0.481in)},
        			scale only axis,
        			xmin=0,
        			xmax=50,
        			xlabel style={font=\color{black},yshift=1ex},
        			xlabel={$t'$},
        			ymin=1e-90,
        			ymax=1,
        			ymode = log,
        			yminorticks=true,
        			ytick={1e0, 1e-15, 1e-30, 1e-45, 1e-60, 1e-75, 1e-90},
        			ylabel style={font=\color{black}, yshift=-.5ex},
        			ylabel={$p_{t,t'}$},
        			axis background/.style={fill=white},
        			xmajorgrids,
        			ymajorgrids,
        			yminorgrids,
        			legend style={at={(0.99,0.05)}, anchor=south east, row sep=-2.5pt, legend cell align=left, align=left, draw=white!15!black}
        			]		
                    \addplot [color=blue,line width=1pt,dashed]
        			table[row sep=crcr]{%
                        0 1.0000e+00 \\ 
1 3.5258e-02 \\ 
2 9.6449e-04 \\ 
3 2.7180e-05 \\ 
4 8.8408e-07 \\ 
5 3.5260e-08 \\ 
6 1.7890e-09 \\ 
7 1.1830e-10 \\ 
8 1.0368e-11 \\ 
9 1.2188e-12 \\ 
10 1.9382e-13 \\ 
11 4.1964e-14 \\ 
12 1.2427e-14 \\ 
13 5.0500e-15 \\ 
14 2.6291e-15 \\ 
15 1.4654e-15 \\ 
16 8.6777e-16 \\ 
17 5.4024e-16 \\ 
18 3.5444e-16 \\ 
19 2.4387e-16 \\ 
20 1.7553e-16 \\ 
21 1.3189e-16 \\ 
22 1.0321e-16 \\ 
23 8.3983e-17 \\ 
24 7.0948e-17 \\ 
25 6.2121e-17 \\ 
26 5.6306e-17 \\ 
27 5.2761e-17 \\ 
28 5.1051e-17 \\ 
29 5.0950e-17 \\ 
30 5.2395e-17 \\ 
31 5.5470e-17 \\ 
32 6.0379e-17 \\ 
33 6.7543e-17 \\ 
34 7.7553e-17 \\ 
35 9.1349e-17 \\ 
36 1.1029e-16 \\ 
37 1.3636e-16 \\ 
38 1.7255e-16 \\ 
39 2.2330e-16 \\ 
40 2.9528e-16 \\ 
41 3.9871e-16 \\ 
42 5.4941e-16 \\ 
43 7.7209e-16 \\ 
44 1.1057e-15 \\ 
45 1.6126e-15 \\ 
46 2.3937e-15 \\ 
47 3.6138e-15 \\ 
                    };
                    \addlegendentry{$r = r_0$};

                    \addplot [color=blue,line width=1pt,dashed]
        			table[row sep=crcr]{%
                        0 5.1806e-39 \\ 
1 2.8344e-40 \\ 
2 1.1981e-41 \\ 
3 5.1962e-43 \\ 
4 2.5905e-44 \\ 
5 1.5772e-45 \\ 
6 1.2168e-46 \\ 
7 1.2188e-47 \\ 
8 1.6118e-48 \\ 
9 2.8480e-49 \\ 
10 6.7827e-50 \\ 
11 2.1911e-50 \\ 
12 9.6464e-51 \\ 
13 5.8072e-51 \\ 
14 4.7919e-51 \\ 
15 5.4287e-51 \\ 
16 7.6909e-51 \\ 
17 1.3186e-50 \\ 
18 2.6814e-50 \\ 
19 6.3422e-50 \\ 
20 1.7120e-49 \\ 
21 5.1823e-49 \\ 
22 1.7307e-48 \\ 
23 6.2826e-48 \\ 
24 2.4458e-47 \\ 
25 1.0088e-46 \\ 
26 4.3611e-46 \\ 
27 1.9573e-45 \\ 
28 9.0464e-45 \\ 
29 4.2763e-44 \\ 
30 2.0559e-43 \\ 
31 1.0010e-42 \\ 
32 4.9214e-42 \\ 
33 2.4388e-41 \\ 
34 1.2176e-40 \\ 
35 6.1294e-40 \\ 
36 3.0968e-39 \\ 
37 1.5105e-38 \\ 
38 7.4268e-38 \\ 
39 3.6995e-37 \\ 
40 1.8178e-36 \\ 
41 8.6663e-36 \\ 
42 4.2286e-35 \\ 
43 2.0209e-34 \\ 
44 9.4390e-34 \\ 
45 4.5918e-33 \\ 
46 2.0593e-32 \\ 
47 9.6874e-32 \\ 
48 4.3450e-31 \\ 
                    };

                    \addplot [color=blue,line width=1pt,dashed]
        			table[row sep=crcr]{%
                        0 2.0407e-77 \\ 
1 1.7253e-78 \\ 
2 1.1224e-79 \\ 
3 7.4605e-81 \\ 
4 5.6779e-82 \\ 
5 5.2566e-83 \\ 
6 6.1427e-84 \\ 
7 9.2838e-85 \\ 
8 1.8455e-85 \\ 
9 4.8835e-86 \\ 
10 1.7354e-86 \\ 
11 8.3344e-87 \\ 
12 5.4354e-87 \\ 
13 4.8301e-87 \\ 
14 5.8630e-87 \\ 
15 9.7372e-87 \\ 
16 2.2150e-86 \\ 
17 6.5804e-86 \\ 
18 2.4876e-85 \\ 
19 1.1354e-84 \\ 
20 6.2852e-84 \\ 
21 4.1656e-83 \\ 
22 3.2432e-82 \\ 
23 2.9141e-81 \\ 
24 2.9465e-80 \\ 
25 3.2728e-79 \\ 
26 4.0054e-78 \\ 
27 5.3313e-77 \\ 
28 7.6273e-76 \\ 
29 1.1607e-74 \\ 
30 1.8617e-73 \\ 
31 3.0979e-72 \\ 
32 5.1886e-71 \\ 
33 8.9341e-70 \\ 
34 1.5741e-68 \\ 
35 2.8280e-67 \\ 
36 5.1696e-66 \\ 
37 9.0495e-65 \\ 
38 1.5916e-63 \\ 
39 2.8253e-62 \\ 
40 5.0728e-61 \\ 
41 8.4497e-60 \\ 
42 1.4185e-58 \\ 
43 2.4142e-57 \\ 
44 3.8446e-56 \\ 
45 6.0483e-55 \\ 
46 9.7374e-54 \\ 
47 1.4188e-52 \\ 
48 2.1082e-51 \\ 
49 3.0546e-50 \\ 
                    };
                    \node at (axis cs:43,1e-12) {$t = 0$};
                    \node at (axis cs:43,1e-30) {$t = 1$};
                    \node at (axis cs:43,1e-51) {$t = 2$};
        		\end{axis}
        	\end{tikzpicture}%
            \caption{The bound $p_{t,t'}$ on the probability of having $t$ misdetections and $t'$ false alarms for a given $t$ and $t' \in \overline{\Tc}_t = {[(K_{\ell} - K_{\rm a} + t)^+ : (K_u - K_{\rm a} + t)]}$.} 
            \label{fig:jointDec_highSNR_ptt}
            \end{subfigure}
            \hspace{.2cm}
            \begin{subfigure}[t]{.48\textwidth}
                \centering
                \begin{tikzpicture}
        		\tikzstyle{every node}=[font=\small]
        		\begin{axis}[%
        			width=2.8in,
        			height=2in,
        			at={(0.759in,0.481in)},
        			scale only axis,
        			xmin=0,
        			xmax=50,
        			xlabel style={font=\color{black},yshift=1ex},
        			xlabel={$t$},
        			ymin=1e-90,
        			ymax=1,
        			ymode = log,
        			yminorticks=true,
        			ytick={1e0, 1e-15, 1e-30, 1e-45, 1e-60, 1e-75, 1e-90},
        			ylabel style={font=\color{black}, yshift=-.5ex},
        			ylabel={$p_{t}$},
        			axis background/.style={fill=white},
        			xmajorgrids,
        			ymajorgrids,
        			yminorgrids,
        			legend style={at={(0.99,0.99)}, anchor=north east, row sep=-2.5pt, legend cell align=left, align=left, draw=white!15!black}
        			]		
                    \addplot [color=blue,line width=1pt]
        			table[row sep=crcr]{%
                        0 1.0000e+00 \\ 
1 5.5775e-31 \\ 
2 3.2807e-50 \\ 
3 3.8501e-71 \\ 
4 4.6626e-93 \\ 
5 1.6943e-115 \\ 
6 3.4709e-138 \\ 
7 5.7378e-161 \\ 
8 1.0593e-183 \\ 
9 2.5600e-206 \\ 
10 8.2914e-229 \\ 
11 4.8256e-251 \\ 
12 4.4176e-273 \\ 
13 8.1984e-295 \\ 
14 2.4141e-316 \\ 
                    };

                    \node at (axis cs:5,1e-13) {$r = r_0$};
                \end{axis}
                \end{tikzpicture}
                \caption{The bound $p_t = \sum_{t' \in \overline{\Tc}_t} p_{t,t'}$ on the probability of having $t$ misdetections.}
                \label{fig:jointDec_highSNR_pt}
            \end{subfigure}
            \caption{The bounds $p_{t,t'}$ in~\eqref{eq:ptt} and $p_{t}$ in~\eqref{eq:pt} for the setting in Fig.~\ref{fig:jointDec} but with $\EbNo = 4.6$~dB.}
            \label{fig:jointDec_highSNR}
        \end{figure*}

    \reviseee{
    In Fig.~\ref{fig:jointDec} and Fig.~\ref{fig:jointDec_highSNR}, we plot the values of $p_{t,t'}$ and $p_t$ for $k = 128$ bits, $n = 19200$ channel uses, and $K_{\rm a} = \E[\rv{K}_{\rm a}] = 50$. We set $K_\ell = (\E[\rv{K}_{\rm a}] - r)^+$ and $K_u = \E[\rv{K}_{\rm a}] + r$ for a chosen nonnegative integer $r$. We first consider $r = r_0$, where $r_0$ is the largest value of $r$ such that $\P[{\rv{K}_{\rm a} \notin [K_\ell:K_u]}] < 10^{-9}$. In Fig.~\ref{fig:jointDec}, we consider $\EbNo = 2$~dB, which yields a small $P'$. As shown in Fig.~\ref{fig:jointDec_ptt}, $p_{t,t'}$ increases with $t'$ for a fixed $t$. This is because for a small $P'$, the term $E_0(\rho,\rho_1)$ increases slowly with $t'$, and the term $-\rho\rho_1 t' R_1$ dominates for all values of $t'$, driving the error exponent $E(t,t')$ towards~$0$. When $r = r_0$, the interval $[K_\ell:K_u]$ is large. Therefore, both ${\Tc}_t = {[(\max\{K_{\ell},1\} - K_{\rm a} + t)^+ : (K_u - K_{\rm a} + t)]}$ and $\overline{\Tc}_t = {[(K_{\ell} - K_{\rm a} + t)^+ : (K_u - K_{\rm a} + t)]}$ contain large values of $t'$ for which $p_{t,t'}$ is close to $1$. As a consequence, $\epsilon_{\rm FA}$ is large due to the sum over $t'\in \Tc_t$ in~\eqref{eq:eps_FA_jointDec}. Furthermore, $p_t = \sum_{t' \in \overline{\Tc}_t} p_{t,t'}$ is also large, as seen in Fig.~\ref{fig:jointDec_pt}, leading to a large $\epsilon_{\rm MD}$.
    In Fig.~\ref{fig:jointDec_highSNR}, we consider $\EbNo = 4.6$~dB, i.e., a higher $P'$. As shown in Fig.~\ref{fig:jointDec_highSNR_ptt}, for a fixed $t$, the term $p_{t,t'}$ first decreases and then increases with $t'$. This is because when $P'$ is sufficiently high, $E_0(\rho,\rho_1)$ dominates for small $t'$ but $-\rho\rho_1 t' R_1$ eventually dominates as $t'$ grows. However, $p_{t,t'}$ increases rather slowly with $t'$ and remains small for all $t' \in \overline{\Tc}_t$. It follows that  $p_t = \sum_{t' \in \overline{\Tc}_t} p_{t,t'}$ is small, as seen in Fig.~\ref{fig:jointDec_highSNR_pt}.}

        
        \reviseee{The fact that $p_{t,t'}$ increases quickly with $t'$ at low $\EbNo$ seems to indicate that the joint decoder tends to commit many false alarms in this regime. While it remains unclear if this interpretation is correct,\footnote{\reviseee{Indeed, this rapid increase may simply be due to the looseness of Gallager's $\rho$-trick.}} we provide a possible explanation as follows.} At low $\EbNo$, i.e., when the noise dominates, the joint decoder often returns a list (of typically big size) of wrong codewords whose sum is  closer to $\rvVec{y}$ than to the sum of the transmitted codewords. \reviseee{Although it is unlikely that $\|c({\Wc}') - \rvVec{y}\|^2 \le \|c({\Wc}) - \rvVec{y}\|^2$ for a given set $\Wc'$ of large size, the probability that this is true for at least one of $\binom{M}{|\Wc'|}$ possible sets is still significant since $M$ is large.} In essence, by \reviseee{searching over a large set of codewords,} the decoder ends up approximating the additive noise component in the received signal. \reviseee{We refer to this effect as ``noise overfitting''.} 
        
        To mitigate noise overfitting, we simply reduce the feasible set over which the minimization in~\eqref{eq:joint_decoder} is performed. This is similar to using an inductive bias to restrict the hypothesis class in order to overcome overfitting in statistical learning~\cite[Sec.~2.3]{Shalev_ML}. \reviseee{Specifically, we reduce $r$. Indeed, this allows us to avoid large values of $t'$ in $\Tc_t$ and $\overline{\Tc}_t$, and thus reduce~$p_t$. For the scenario in Fig.~\ref{fig:jointDec}, setting $r = 1$ results in $\Tc_t = \overline{\Tc}_t = [(t-1)^+:(t+1)]$. These intervals do not contain large values of $t'$ for which $p_{t,t'}$ is close to $1$, as shown in Fig.~\ref{fig:jointDec_ptt}. It follows that $p_t$ is drastically reduced, as shown in Fig.~\ref{fig:jointDec_pt}.  
        Furthermore, to adapt the inductive bias to the received signal, we choose the feasible set based on the estimate ${K}_{\rm a}'$ of $\rv{K}_{\rm a}$. Specifically,} we replace the feasible set in~\eqref{eq:joint_decoder} with the set $\{\Wc' \subset [M] \colon \underline{K'_{\rm a}} \le |\Wc'| \le \overline{K'_{\rm a}}\}$ where $[\underline{K'_{\rm a}}:\overline{K'_{\rm a}}]$ is an interval around ${K}_{\rm a}'$. 
        We control the size of this interval via the decoding radius $r$. 
        It turns out that, to satisfy mild requirements on the MD and FA probabilities, setting $\underline{K'_{\rm a}} = \overline{K'_{\rm a}} = K'_{\rm a}$, i.e., $r = 0$, leads to high energy efficiency, whereas for more stringent requirements, the interval $[\underline{K'_{\rm a}}:\overline{K'_{\rm a}}]$ should be progressively enlarged. 

        An alternative method to overcome noise overfitting is to introduce a regularization term in~\eqref{eq:joint_decoder} that penalizes a choice of $|\widehat{\Wc}|$ far from $\E[\rv{K}_{\rm a}]$. We have tried this method but did not obtain a better bound than the one provided by the two-step decoder.
        
		In short, although our two-step decoder might be suboptimal, it effectively mitigates noise overfitting and achieves the highest energy efficiency among the approaches that we have considered.

		\subsection{A General Converse} \label{sec:discussion_converse} 
		A general converse bound on the \gls{MD} and \gls{FA} probabilities appears difficult to obtain. A possible approach is to assume  that $\rv{K}_{\rm a}$ is known to the receiver. However, even for this case, a tight converse bound is not available in the literature. In~\cite{PolyanskiyISIT2017massive_random_access}, only a conjectured converse bound was provided. A converse bound on the required power to achieve a target \gls{MD} probability $\epsilon_{\rm MD}$ for fixed and known $K_{\rm a}$ was reported in~\cite{Polyanskiy_ISIT2021_tutorial}. This converse is based on two different approaches. In the first approach, one casts a {UMA} code as a single-user code with list decoding, and applies the result on minimum energy to send $k$ bits through the Gaussian channel~\cite{Polyanskiy2011minimum}. Here, the list size is $K_{\rm a}$. 
        In the second approach, one computes the rate-distortion function between two binary vectors that indicate the transmitted and decoded messages. The average distortion between these vectors is bounded by $2K_{\rm a} \epsilon_{\rm MD}$. Then, this rate-distortion function is upper-bounded by the sum-capacity $n\log(1 + K_{\rm a} P)$ of the Gaussian MAC. The bound obtained via the first approach dominates when $K_{\rm a}$ is small whereas the one obtained via the second approach dominates when $K_{\rm a}$ is large. However, this second bound holds only for $\epsilon_{\rm MD} \le \frac{1}{e K_{\rm a}}$.\footnote{\revisee{In~\cite[Slide~31]{Polyanskiy_ISIT2021_tutorial}, the step $\E[\rv{S} \log \rv{S}]\ge t \log t$ holds for $t \in (0,1/e]$ only. Here, $\rv{S}$ is the number of misdetected messages and $t = K_{\rm a} \epsilon_{\rm MD}$. Since $\epsilon_{\rm MD}$ is the target \gls{MD} probability, we have that $\E[\rv{S}] \le t$. Jensen's inequality implies that $\E[\rv{S} \log \rv{S}] \ge \E[\rv{S}] \log \E[\rv{S}]$, but $\E[\rv{S}] \log \E[\rv{S}] \ge t \log t$ only if $t \in (0,1/e]$. Therefore, the second bound in the converse reported in~\cite{Polyanskiy_ISIT2021_tutorial} holds only for $\epsilon_{\rm MD} \le \frac{1}{e K_{\rm a}}$.}} For $\epsilon_{\rm MD} > \frac{1}{e K_{\rm a}}$, which typically holds in massive \gls{IoT} applications, one needs to rely on the bound obtained from the first approach, which exhibits a large gap from the achievability bound for large $K_{\rm a}$.
		
		The converse in~\cite{Polyanskiy_ISIT2021_tutorial} does not generalize naturally to the case of unknown $\rv{K}_{\rm a}$, where both the \gls{MD} and \gls{FA} probabilities need to be considered, for the following two reasons. First, in the list-decoding-based approach, the list size $\rv{K}_{\rm a}$ is not fixed and known, and the error event associated with single-user list decoding accounts for \gls{MD} only. Second, in the rate-distortion-based approach, it is nontrivial to express/bound the average distortion between the aforementioned binary vectors in terms of the target \gls{MD} and \gls{FA} probabilities.
		
		
		\section{Conclusion} \label{sec:conclusions}
		To account for the random user activity in the \gls{IoT}, we proposed a formulation for unsourced multiple access where both the identity and the number of active users are unknown. We derived a random-coding bound for the Gaussian MAC that reveals a trade-off between misdetection and false alarm. Our bound provides an estimate of the penalty in terms of energy efficiency due to the lack of knowledge of the number of active users, and serves as a benchmark to assess the performance of practical schemes. Numerical results show that for the Gaussian MAC, if the target misdetection and false-alarm probabilities are sufficiently high, e.g., $10^{-1}$, the lack of knowledge of the number of active users entails a small loss. In this case, it is effective to adapt a coding scheme that performs well for the case of known number of active users, by first estimating the number of active users and then treating this estimate as the true value in the decoding process. However, for stringent target misdetection and false-alarm probabilities, e.g., $10^{-3}$, numerical results \revise{suggest} that the loss due to the lack of knowledge of the number of active users \revise{might be} significant. \revise{It remains unclear if this is loss is fundamental or pertains to the considered random-coding scheme only.} \revise{For stringent requirements}, adapting existing coding schemes proposed for known number of active users by simply treating the estimate of the number of active users as perfect is energy inefficient even for a small number of active users. Therefore, more sophisticated methods are needed to handle effectively the uncertainty about the number of active users. 
		
		\section*{Acknowledgement}
		Khac-Hoang Ngo has received funding from the European Union’s Horizon 2020 research and innovation programme under the Marie Skłodowska-Curie grant agreement No 101022113. Alejandro Lancho has received funding from the European Union’s Horizon 2020 research and innovation programme under the Marie Skłodowska-Curie grant agreement No 101024432. Giuseppe Durisi has received funding from the Swedish Research Council under grant 2021-04970. The authors would like to thank Y. Polyanskiy for fruitful discussions, as well as V. K. Amalladinne, J.-F. Chamberland, and K. R.~Narayanan for providing their simulation codes. They would like to further thank the Associate Editor, Or Ordentlich, as well as the anonymous reviewers for their valuable comments.

		\appendices
		\section{Proof of Theorem~\ref{thm:RCU_unknownKa}} \label{app:proof}
		The following well-known results will be used in the proof.
		
		\begin{lemma}[{Change of measure~\cite[Lemma~4]{Ohnishi2021}}] \label{lem:change_measure}
			Let $p$ and $q$ be two probability measures. Consider a random variable $\rv{x}$ supported on $\Hc$ and a function $f \colon \Hc \to [0,1]$. It holds that 
			\begin{align}
				\E_p[f(\rv{x})] \le \E_q[f(\rv{x})] + d_{\rm TV}(p,q)
			\end{align}
			where $d_{\rm TV}(p,q)$ denotes the total variation distance between $p$ and $q$.
		\end{lemma}

		\begin{lemma}[{Chernoff bound~\cite[Th. 6.2.7]{DeGroot2012ProbStats}}] \label{lem:Chernoff}
			For a random variable $\rv{x}$ with moment-generating function $\E[e^{t \rv{x}}]$ defined for all $|t| \le b$, it holds for all $\lambda \in [0,b]$ that
			\begin{align}
				\P[\rv{x} \le x] \le e^{\lambda x} \E[e^{-\lambda \rv{x}}].
			\end{align}
		\end{lemma}
		
		\begin{lemma} [{Gallager's $\rho$-trick~\cite[p.~136]{Gallager1968information}}] \label{lem:Gallager}
			It holds that $\P[\cup_i A_i] \le (\sum_{i} \P[A_i])^\rho$ for every $\rho \in [0,1]$.
		\end{lemma}
		\begin{lemma} \label{lem:chi2}
			Let $\rvVec{x} \sim \Cc\Nc(\muv,\sigma^2\Id_n)$. It holds that 
			\begin{align}
				\E[e^{-\gamma \|\rvVec{x}\|^2}] = (1+\gamma\sigma^2)^{-n} \exp\bigg(-\frac{\gamma\|\muv\|^2}{1+\gamma\sigma^2}\bigg), ~ \forall \gamma > -\frac{1}{\sigma^2}. \label{eq:tmp363}
			\end{align} 
		\end{lemma}
		\revise{\begin{proof}
				We write 
				$\E\big[e^{-\gamma \|\rvVec{x}\|^2}\big] = \E[\exp\Big(-\frac{\gamma \sigma^2}{2} \big\|\frac{\sqrt{2}}{\sigma} \rvVec{x}\big\|^2\Big)]$. Notice that $ \big\|\frac{\sqrt{2}}{\sigma} \rvVec{x}\big\|^2$ follows the noncentral chi-square distribution with $2n$ degrees of freedom and noncentrality parameter $\frac{2\|\muv\|^2}{\sigma^2}$. We obtain~\eqref{eq:tmp363} by using that the moment-generating function of this distribution is known in closed form.
		\end{proof}	}
		
		We present next an error analysis of the random-coding scheme introduced in Section~\ref{sec:RCU}. 
		Denote by $\revise{\Wc_{\rm MD}}$ the set of misdetected messages, i.e., $\revise{\Wc_{\rm MD}} \defeq \widetilde{\Wc} \setminus \widehat{\Wc}$, and by \revise{$\Wc_{\rm FA}$} the set of falsely alarmed messages, i.e., $\revise{\Wc_{\rm FA}} \defeq \widehat{\Wc} \setminus \widetilde{\Wc}$. The \gls{MD} and \gls{FA} probabilities, defined respectively in~\eqref{eq:def_pMD} and \eqref{eq:def_pFA}, can be expressed as the average fraction of misdetected and falsely alarmed messages as 
		\begin{align}
			P_{\rm MD} &= \revise{\E[\md]}, \label{eq:pMD}\\
			P_{\rm FA} &= \revise{\E[\fa]}. \label{eq:pFA}
		\end{align}

		\subsection{Change of Measure} \label{sec:change_measure}
		Since \revise{$\md$} and \revise{$\fa$} are nonnegative random variables that are upper-bounded by one, we can apply Lemma~\ref{lem:change_measure}. Specifically, we replace the measure over which the expectation is taken by the one under which: 
		i) there are at least $K_\ell$ and at most $K_{u} \ge \overline{K_{\rm a}'}$ active users, i.e., $K_\ell\le\rv{K}_{\rm a} \le K_{u}$; ii) the active users transmit distinct messages, i.e., $|\widetilde{\Wc}| = \rv{K}_{\rm a}$ and $\widetilde{\rv{w}}_1,\dots,\widetilde{\rv{w}}_{\rv{K}_{\rm a}}$ are sampled uniformly without replacement from $[M]$;  iii) $\rvVec{x}_i = \cv_{\rv{w}_i}, \forall i$, instead of $\rvVec{x}_i = \cv_{\rv{w}_i} \ind{\|\cv_{\rv{w}_i}\|^2 \le nP}$. 
		It then follows from \cite[Eq. (41)]{Kowshik2020fundamental} that the total variation between the original measure and the new one is
		upper-bounded by $\P[{\rv{K}_{\rm a} \notin [K_\ell:K_u]}] + \P[ |\widetilde{\Wc}| < \rv{K}_{\rm a}] + \P[\overline{U}]$, where
		$U \defeq \{\|\cv_{\rv{w}_i}\|^2 \le nP, \forall i \in [\rv{K}_{\rm a}] \}$ and $\overline{U}$ denotes the complement of~$U$.
		We compute these probabilities 
		as follows:
		\begin{itemize}
			\item To compute the first probability, we simply use that $$\P[{\rv{K}_{\rm a} \notin [K_\ell:K_u]}] = 1 - \sum_{K_{\rm a} = K_\ell}^{K_{u}}P_{\rv{K}_{\rm a}}(K_{\rm a}).$$
			
			\item To evaluate $\P[ |\widetilde{\Wc}| < \rv{K}_{\rm a}]$, consider a given $\rv{K}_{\rm a} = K_{\rm a}$. Since ${\rv{w}}_1,\dots,{\rv{w}}_{K_{\rm a}}$ are drawn uniformly and independently from $[M]$, there are $M^{K_{\rm a}}$ possible $K_{\rm a}$-tuples. Among them, $\frac{M!}{(M-K_{\rm a})!}$ tuples have nonduplicate elements. Therefore, $\P[ |\widetilde{\Wc}| = K_{\rm a} \cond \rv{K}_{\rm a} = K_{\rm a}] = \frac{M!}{(M-K_{\rm a})!} \frac{1}{M^{K_{\rm a}}}$. As a consequence, 
			\begin{align} \label{eq:P_msg_collision}
				\P[ |\widetilde{\Wc}| < \rv{K}_{\rm a}] &= 1 - \P[ |\widetilde{\Wc}| = \rv{K}_{\rm a}] \\
                &= 1- \E_{\rv{K}_{\rm a}}\bigg[\frac{M!}{M^{\rv{K}_{\rm a}}(M-\rv{K}_{\rm a})!}\bigg].
			\end{align} 
		
		\item The probability $\P[\overline{U}]$ can be finally evaluated as
		\begin{align}
			\P[\overline{U}] &= \E_{\rv{K}_{\rm a}}\Bigg[{\P[\bigcup_{i=1}^{\rv{K}_{\rm a}} \|\cv_{\rv{w}_i}\|^2 > nP \revise{\;\Big\vert\; \rv{K}_{\rm a}}]}\Bigg] \\
			&\le \E_{\rv{K}_{\rm a}} \Bigg[\sum_{i=1}^{\rv{K}_{\rm a}}{\P[\|\cv_{\rv{w}_i}\|^2 > nP]}\Bigg] \label{eq:tmp675}\\
			&= \E[\rv{K}_{\rm a}]  \frac{\Gamma(n,nP/P')}{\Gamma(n)}, \label{eq:tmp676}
		\end{align}
		where \eqref{eq:tmp675} follows from the union bound and \eqref{eq:tmp676} holds since $\|\cv_{\rv{w}_i}\|^2$ follows the Gamma distribution with shape $n$ and scale $P'$.
	\end{itemize}
	From the above calculations, we deduce that the total variation between the two measures is upper-bounded by $\tilde{p}$ defined in~\eqref{eq:p0}.
	Hence, applying Lemma~\ref{lem:change_measure} to the random quantities $\md$ and $\fa$, we consider implicitly the new measure from now on at a cost of adding $\tilde{p}$ to the original expectations. 
	
	It remains to bound the \gls{MD} and \gls{FA} probabilities given in~\eqref{eq:pMD} and \eqref{eq:pFA}, respectively, under the new measure. For the sake of clarity, in Appendix~\ref{sec:special_case}, we shall prove a bound on $P_{\rm MD}$ and $P_{\rm FA}$ for a special case where i) $\rv{K}_{\rm a}$ and $\rv{K}'_{\rm a}$ are fixed and $r = 0$, i.e., there are always $K_{\rm a}$ users transmitting and the decoder always outputs a list of size $K'_{\rm a}$; ii) $K'_{\rm a} < \min\{K_{\rm a}, M-K_{\rm a}\}$. Later, in Appendix~\ref{sec:general_case}, we shall show how to extend the proof to the general case where $\rv{K}_{\rm a}$ and $\rv{K}'_{\rm a}$ are random and $r \ge 0$.
	
	\vspace{-.2cm}
	\subsection{A Special Case} \label{sec:special_case}
	In the aforementioned special case, \eqref{eq:eps_MD} and \eqref{eq:eps_FA} become   
	\begin{align}
		\epsilon_{\rm MD} &= \sum_{t = 0}^{K_{\rm a}'}\frac{t+K_{\rm a}-K_{\rm a}'}{K_{\rm a}} \min\{p_{t,t},q_{t,t}\} + \tilde{p}, \label{eq:eps_MD_simp}\\
		\epsilon_{\rm FA} &= \sum_{t = 0}^{K_{\rm a}'}  \frac{t}{K'_{\rm a}} \min\{p_{t,t}, q_{t,t}\} + \tilde{p}, \label{eq:eps_FA_simp}
	\end{align}		
	where $p_{t,t}$ and $q_{t,t}$ will be given shortly. The task is to show that $\epsilon_{\rm MD}$ and $\epsilon_{\rm FA}$ are indeed upper bounds of $P_{\rm MD}$ and $P_{\rm FA}$, respectively, in this special case.
	
	Since the decoded list size $K'_{\rm a}$ is smaller than the number of transmitted messages~$K_{\rm a}$, $K_{\rm a} - {K_{\rm a}'}$ messages are initially misdetected, and there can be $t \in [0:K_{\rm a}']$ additional \glspl{MD} occurring during the decoding process. Exploiting symmetry, we assume \gls{wlog} that $\widetilde{\Wc} = [K_{\rm a}]$ and that the list of messages that are initially misdetected due to insufficient decoded list size is $\revise{\Wc_{\rm iMD}} = [{K}_{\rm a} - {K_{\rm a}'}]$. 
	Furthermore, let $\revise{\Wc_{\rm aMD}} = \revise{\Wc_{\rm MD}} \setminus \revise{\Wc_{\rm iMD}}$ denote the set of $t$ additional \glspl{MD}.  Note that $\revise{\Wc_{\rm aMD}}$ is a generic subset of $[K_{\rm a} - {K_{\rm a}'} + 1:K_{\rm a}]$. Note also that $t$ is the number of \glspl{FA}, i.e., $|\revise{\Wc_{\rm FA}}| = t$. The relation between these sets of messages is depicted in Fig.~\ref{fig:venn_special_case}.
	\begin{figure}
		\centering
		\begin{tikzpicture}[thick,scale=.9, every node/.style={scale=.9}]
			\def\radiusB{2cm}
			\def\radius{0.9*\radiusB}
			\def\mycolorbox#1{\textcolor{#1}{\rule{2ex}{2ex}}}
			\colorlet{colori}{gray!80}
			\colorlet{colorii}{gray!20}
			
			\coordinate (ceni) at (0,0);
			\coordinate[xshift=.95*\radiusB] (cenii);
			
			\coordinate (edge1a) at (-\radiusB,0.1cm);
			
			\coordinate (edge2a) at (\radiusB-\radius-.1cm,.1cm);
			
			\draw[fill=colori,fill opacity=0.5] (ceni) circle (\radiusB);
			\draw[fill=colorii,fill opacity=0.5] (cenii) circle (\radius);
			\draw (ceni) circle (\radiusB);
			
			\draw (edge1a) to (edge2a);
			
			\draw[-latex] (-1.35*\radiusB,-0.6*\radiusB) node[below,xshift=-.4cm,text width=3cm,font=\linespread{1}\selectfont,align=center] {\small Distinct transmitted messages $\widetilde{\Wc}$} -- (-.87*\radiusB,-0.5*\radiusB);
			\draw[-latex] (2*\radiusB,-0.6*\radiusB) node[below,text width=2cm,font=\linespread{1}\selectfont,align=center] {\small Decoded messages $\widehat{\Wc}$} -- (1.7*\radiusB,-0.5*\radiusB);
			
			\node[yshift=1.25*\radiusB,xshift=-1.25cm,text width=4.2cm,font=\linespread{1}\selectfont,align=center] {\small \glspl{MD}: \\ $\revise{\Wc_{\rm MD}} = \revise{\Wc_{\rm iMD}} \cup \revise{\Wc_{\rm aMD}}$ \\ $= \widetilde{\Wc} \setminus \widehat{\Wc}$};
			
			\node[yshift=1.15*\radiusB,xshift=1cm,text width=4cm,font=\linespread{1}\selectfont,align=center] at (cenii) {\small \glspl{FA}: \\ $\revise{\Wc_{\rm FA}} = \widehat{\Wc} \setminus \widetilde{\Wc}$};
			
			\node[xshift=1cm,text width=1.5cm,font=\linespread{1}\selectfont,align=center] at (ceni) {\small correctly decoded messages $\widetilde{\Wc} \cap \widehat{\Wc}$};
			
			\node[yshift=.43*\radiusB,xshift=-.7cm,text width=2.3cm,font=\linespread{1}\selectfont,align=center] at (ceni) {\small $~~~~{K}_{\rm a}-K_{\rm a}'$ initial \glspl{MD} $~~~\revise{\Wc_{\rm iMD}}~~~$};
			
			\node[yshift=-.45*\radiusB,xshift=-.4*\radiusB,text width=1.8cm,font=\linespread{1}\selectfont,align=center] at (ceni) {\small $t$ additional \glspl{MD} $~~~~~~\revise{\Wc_{\rm aMD}}~~~$};
			
			
			\node[yshift=0cm,xshift=.45*\radiusB,text width=1.5cm,font=\linespread{1}\selectfont,align=center] at (cenii) {\small $t$ \glspl{FA} $~~~\revise{\Wc_{\rm FA}}~~~$};
		\end{tikzpicture}
		\caption{A diagram depicting the relation between the sets of messages defined for the special case in Appendix~\ref{sec:special_case}.}
		\label{fig:venn_special_case}
	\end{figure}
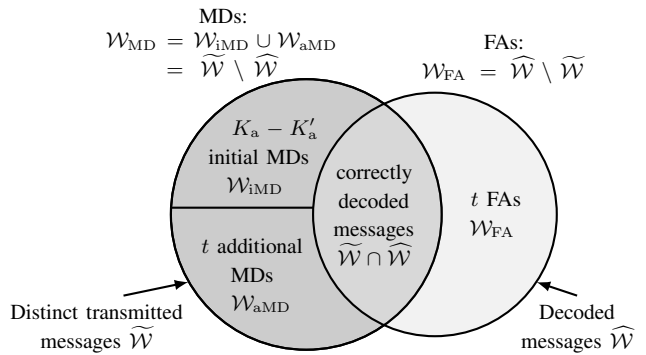
	
	Using the above definitions, the set of transmitted messages can be expressed as $$\Wc = \widetilde{\Wc} = \revise{\Wc_{\rm iMD}} \cup \revise{\Wc_{\rm aMD}} \cup (\widetilde{\Wc} \setminus \revise{\Wc_{\rm MD}}),$$ and the received signal is $$\rvVec{y} = c(\revise{\Wc_{\rm iMD}}) + c(\revise{\Wc_{\rm aMD}}) + c(\widetilde{\Wc} \setminus \revise{\Wc_{\rm MD}}) + \rvVec{z}.$$ Since the messages in $\revise{\Wc_{\rm iMD}}$ are always misdetected, the best approximation of $\Wc$ that the decoder can produce is $\revise{\Wc_{\rm aMD}} \cup (\widetilde{\Wc} \setminus \revise{\Wc_{\rm MD}})$. However, under the considered error event $\widetilde{\Wc} \to \widehat{\Wc}$, the messages in $\revise{\Wc_{\rm aMD}}$ are misdetected by the ones in \revise{$\Wc_{\rm FA}$}, and thus the actual decoded list is $\revise{\Wc_{\rm FA}} \cup (\widetilde{\Wc} \setminus \revise{\Wc_{\rm MD}})$. Therefore, $\widetilde{\Wc} \to \widehat{\Wc}$ implies that $\|\rvVec{y} - c(\revise{\Wc_{\rm FA}}) - c(\Wc \setminus \revise{\Wc_{\rm MD}})\|^2 < \|\rvVec{y} - c(\revise{\Wc_{\rm aMD}}) - c(\Wc \setminus \revise{\Wc_{\rm MD}}))\|^2$, which is equivalent to
	\begin{align}
		\|c(\revise{\Wc_{\rm iMD}}) + c(\revise{\Wc_{\rm aMD}})- c(\revise{\Wc_{\rm FA}}) + \rvVec{z}\|^2 
		< \|c(\revise{\Wc_{\rm iMD}}) + \rvVec{z}\|^2. \label{eq:tmp1798}
	\end{align}
	Let $F(\revise{\Wc_{\rm iMD}},\revise{\Wc_{\rm aMD}},\revise{\Wc_{\rm FA}})$ denote the set of $(\revise{\Wc_{\rm iMD}},\revise{\Wc_{\rm aMD}},\revise{\Wc_{\rm FA}})$ such that~\eqref{eq:tmp1798} holds. 
	
	We now compute the expectations in~\eqref{eq:pMD} and \eqref{eq:pFA}. 
	Recall that, under assumptions just stated, we have $|\revise{\Wc_{\rm MD}}| =  t + K_{\rm a} - K'_{\rm a}$, $|\revise{\Wc_{\rm FA}}| = |\revise{\Wc_{\rm aMD}}| = t$, and $|\widehat{\Wc}| = K'_{\rm a}$. 
	It follows from \eqref{eq:pMD} and \eqref{eq:pFA} that, after the change of measure in Appendix~\ref{sec:change_measure}, $P_{\rm MD}$ and $P_{\rm FA}$ can be bounded as
	\begin{align}
		P_{\rm MD} &\le \sum_{t=0}^{K'_{\rm a}} \frac{t+K_{\rm a}-{K_{\rm a}'}}{K_{\rm a}} \P[|\revise{\Wc_{\rm aMD}}| = t] + \tilde{p},  \label{eq:tmp850_simp}\\
		P_{\rm FA} &\le \sum_{t=0}^{K'_{\rm a}} \frac{t}{K_{\rm a}'} \P[|\revise{\Wc_{\rm aMD}}| = t] + \tilde{p}. \label{eq:tmp853_simp}
	\end{align} 
	Next, we proceed to bound $\P[|\revise{\Wc_{\rm aMD}}| = t]$ 
	following two approaches. The first approach is based on error exponent analyses, resulting in the term $p_{t,t}$ in~\eqref{eq:eps_MD_simp}. The second approach is a variation of the DT bound \cite[Th.~17]{Polyanskiy2010}, resulting in $q_{t,t}$ in~\eqref{eq:eps_MD_simp}.
	
	\subsubsection{The Error-Exponent-Based Approach}  \label{sec:1st_approach}
	By writing the event $|\revise{\Wc_{\rm aMD}}| = t$ as a union of the pairwise error events $F(\revise{\Wc_{\rm iMD}},\revise{\Wc_{\rm aMD}},\revise{\Wc_{\rm FA}})$, we have that 
	\begin{align}
		&\P[{|\revise{\Wc_{\rm aMD}}| = t}] = \notag \\ &
		\P[\!\bigcup_{\revise{\Wc_{\rm aMD}} \subset [K_{\rm a} - {K_{\rm a}'} + 1:K_{\rm a}] \atop |\revise{\Wc_{\rm aMD}}| = t} \bigcup_{\revise{\Wc_{\rm FA}} \subset [K_{\rm a}+1:M] \atop |\revise{\Wc_{\rm FA}}| = t} \!\!\!F(\revise{\Wc_{\rm iMD}},\revise{\Wc_{\rm aMD}},\revise{\Wc_{\rm FA}})]\!. \label{eq:tmp901_simp} 
	\end{align}
	Next, given $c(\revise{\Wc_{\rm iMD}})$, $c(\revise{\Wc_{\rm aMD}})$, and $\rvVec{z}$, it holds for every $\lambda > -\frac{1}{tP'}$ that
	\begin{align}
		&\P[F(\revise{\Wc_{\rm iMD}},\revise{\Wc_{\rm aMD}},\revise{\Wc_{\rm FA}})] \notag \\
		&\le e^{\lambda \|c(\revise{\Wc_{\rm iMD}}) + \rvVec{z}\|^2} \notag \\&\quad \cdot 
		\E_{c(\revise{\Wc_{\rm FA}})}\Big[e^{-\lambda \|c(\revise{\Wc_{\rm iMD}}) + c(\revise{\Wc_{\rm aMD}})- c(\revise{\Wc_{\rm FA}}) + \rvVec{z}\|^2}\Big] \label{eq:tmp766_simp}\\
		&= e^{\lambda \|c(\revise{\Wc_{\rm iMD}}) + \rvVec{z}\|^2} (1+\lambda tP')^{-n} \notag \\ 	&\quad \cdot 
		\exp\bigg(-\frac{\lambda\|c(\revise{\Wc_{\rm iMD}}) + c(\revise{\Wc_{\rm aMD}}) + \rvVec{z}\|^2}{1+\lambda t P'}\bigg), \label{eq:tmp768_simp}
	\end{align}
	where \eqref{eq:tmp766_simp} follows from the Chernoff bound in Lemma~\ref{lem:Chernoff}, and \eqref{eq:tmp768_simp} follows by computing the expectation in~\eqref{eq:tmp766_simp} using Lemma~\ref{lem:chi2}.
	Next, we apply Gallager's $\rho$-trick in Lemma~\ref{lem:Gallager} and conclude that, given $c(\revise{\Wc_{\rm iMD}})$, $c(\revise{\Wc_{\rm aMD}})$, and $\rvVec{z}$, it holds for every $\rho \in [0,1]$ that
	\begin{align}
		&\P[\bigcup_{\revise{\Wc_{\rm FA}} \subset [K_{\rm a}+1:M] \atop |\revise{\Wc_{\rm FA}}| = t} F(\revise{\Wc_{\rm iMD}},\revise{\Wc_{\rm aMD}},\revise{\Wc_{\rm FA}})] \\
		&\le \binom{M-K_{\rm a}}{t}^\rho (1+\lambda tP')^{-n\rho} 
		\exp\Bigg(\lambda \rho \bigg(\|c(\revise{\Wc_{\rm iMD}}) + \rvVec{z}\|^2 \notag \\
        &\qquad -\frac{\|c(\revise{\Wc_{\rm iMD}}) + c(\revise{\Wc_{\rm aMD}}) + \rvVec{z}\|^2}{1+\lambda t P'}\bigg)\Bigg). \label{eq:tmp803_simp}
	\end{align}
	Taking the expectation over $c(\revise{\Wc_{\rm aMD}})$ using Lemma~\ref{lem:chi2}, we obtain for given $c(\revise{\Wc_{\rm iMD}})$ and $\rvVec{z}$ that
	\begin{align}
		&\P[\bigcup_{\revise{\Wc_{\rm FA}} \subset [K_{\rm a}+1:M] \atop |\revise{\Wc_{\rm FA}}| = t} F(\revise{\Wc_{\rm iMD}},\revise{\Wc_{\rm aMD}},\revise{\Wc_{\rm FA}})] \notag \\
		&\le \binom{M-K_{\rm a}}{t}^\rho (1+\lambda tP')^{-n\rho} \Big(1+\frac{\lambda \rho t P'}{1+\lambda tP'}\Big)^{-n} \notag \\
		&\quad \cdot \exp\Bigg(\lambda\rho \bigg(1-\frac{1}{1+\lambda P't(1+\rho)}\bigg)\|c(\revise{\Wc_{\rm iMD}}) + \rvVec{z}\|^2\Bigg) \\
		&= \binom{M-K_{\rm a}}{t}^\rho \exp\left(b_0\|c(\revise{\Wc_{\rm iMD}})  + \rvVec{z}\|^2 - na_0\right), \label{eq:tmp811}
	\end{align}
	where $a_0$ and $b_0$ are obtained by taking $t' = t$ in~\eqref{eq:a} and \eqref{eq:b}, respectively. 
	Now applying Gallager's $\rho$-trick again, we obtain that, for every $\rho_1 \in [0,1]$,
	\begin{align}
		&\P[\bigcup_{\revise{\Wc_{\rm aMD}} \subset [K_{\rm a} - {K_{\rm a}'} + 1:K_{\rm a}] \atop |\revise{\Wc_{\rm aMD}}| = t} \bigcup_{\revise{\Wc_{\rm FA}} \subset [K_{\rm a}+1:M] \atop |\revise{\Wc_{\rm FA}}| = t} \!\!F(\revise{\Wc_{\rm iMD}},\revise{\Wc_{\rm aMD}},\revise{\Wc_{\rm FA}})] \notag \\
		&\le \binom{K_{\rm a}'}{t}^{\rho_1} \binom{M-K_{\rm a}}{t}^{\rho\rho_1}  \notag \\&\quad \cdot 
		\E[\exp\left(\rho_1 b_0\|c(\revise{\Wc_{\rm iMD}}) + \rvVec{z}\|^2 - n\rho_1 a_0\right)]  \label{eq:tmp797_simp} \\
		&= \binom{K_{\rm a}'}{t}^{\rho_1} \binom{M-K_{\rm a}}{t}^{\rho\rho_1} e^{-n\rho_1 a_0} \big(1-\rho_1P_2b_0\big)^{-n}, \label{eq:tmp800_simp}
	\end{align}
	where the last equality follows by computing the expectation in~\eqref{eq:tmp797_simp} jointly over $c(\revise{\Wc_{\rm iMD}})$ and $\rvVec{z}$ using Lemma~\ref{lem:chi2}, and by setting $P_2= 1+(K_{\rm a} - K_{\rm a}')P'$. Finally, substituting the result into \eqref{eq:tmp901_simp}, we obtain 
	\begin{align}
		&\P[|\revise{\Wc_{\rm aMD}}| = t] \notag \\
		&\le\binom{K_{\rm a}'}{t}^{\rho_1} \binom{M-K_{\rm a}}{t}^{\rho\rho_1} e^{-n\rho_1 a_0} \big(1-\rho_1P_2b_0\big)^{-n} \\
        &\defeq p_{t,t}. \label{eq:tmp1148_simp}
	\end{align}
	
	\subsubsection{The DT-Based Approach} \label{sec:2nd_approach}
	Next, we present an alternative bound on $\P[{|\revise{\Wc_{\rm aMD}}| = t}]$. Consider the channel law $P_{\rvVec{y} \cond c(\revise{\Wc_{\rm MD}}), c(\Wc \setminus \revise{\Wc_{\rm MD}})}$ with input $c(\revise{\Wc_{\rm MD}})$ and output $\rvVec{y}$ where $|\revise{\Wc_{\rm aMD}}| = t$. The corresponding information density~\cite[Def. 17.1]{Polyanskiy2019lecture} 
	is given by
	\begin{align}
		&\imath_t(c(\revise{\Wc_{\rm MD}});\rvVec{y} \cond c(\Wc \setminus \revise{\Wc_{\rm MD}})) \notag \\
		&= n \ln(1+(t+K_{\rm a}-K_{\rm a}')P') + \frac{\|\rvVec{y} - c(\Wc \setminus \revise{\Wc_{\rm MD}})\|^2}{1+(t+K_{\rm a}-K_{\rm a}')P'} \notag \\
		&\quad - \|\rvVec{y} - c(\revise{\Wc_{\rm MD}}) - c(\Wc \setminus \revise{\Wc_{\rm MD}})\|^2. 
	\end{align}
	Note that~\eqref{eq:tmp1798} is equivalent to $$\imath_t(c(\revise{\Wc_{\rm FA}});\rvVec{y} \cond c(\Wc \setminus \revise{\Wc_{\rm MD}})) > \imath_t(c(\revise{\Wc_{\rm aMD}});\rvVec{y} \cond c(\Wc \setminus \revise{\Wc_{\rm MD}})).$$ 
	Let
	$\rv{I}_t \defeq \displaystyle\min_{\revise{\Wc_{\rm aMD}} \subset [K_{\rm a} - {K_{\rm a}'} + 1:K_{\rm a}] \atop |\revise{\Wc_{\rm aMD}}| = t} \imath_t(c(\revise{\Wc_{\rm aMD}});\rvVec{y} \cond c(\Wc \setminus \revise{\Wc_{\rm MD}})).$
	For a fixed arbitrary $\gamma$, it follows that
	\begin{align}
		&\P[{|\revise{\Wc_{\rm aMD}}| = t}] \notag \\
		&=\P[I_{t} \le \gamma]\P[{|\revise{\Wc_{\rm aMD}}| = t \;\big|\; I_{t} \le \gamma}] \notag \\ &\quad 
		+ \P[I_{t} > \gamma]\P[{|\revise{\Wc_{\rm aMD}}| = t \;\big|\; I_{t} > \gamma}] \\
		&\le \P[I_{t} \le \gamma] + \P[{|\revise{\Wc_{\rm aMD}}| = t \;\big|\; I_{t} > \gamma}] \label{eq:tmp838_simp}\\
		&= \P[I_{t} \le \gamma] \notag \\
		&\quad+ \P\Bigg[ \bigcup_{\revise{\Wc_{\rm aMD}} \subset [K_{\rm a} - {K_{\rm a}'} + 1:K_{\rm a}] \atop |\revise{\Wc_{\rm aMD}}| = t} \bigcup_{\revise{\Wc_{\rm FA}} \subset [K_{\rm a}+1:M] \atop |\revise{\Wc_{\rm FA}}| = t} \notag \\
		&\qquad \qquad  \big\{\imath_t(c(\revise{\Wc_{\rm FA}});\rvVec{y} \cond c(\Wc \setminus \revise{\Wc_{\rm MD}})) \notag \\ &\qquad \qquad ~ 
		> \imath_t(c(\revise{\Wc_{\rm aMD}});\rvVec{y} \cond c(\Wc \setminus
		\revise{\Wc_{\rm MD}}))\big\} \Big| I_{t} > \gamma\Bigg] \label{eq:tmp814_simp} \\
		&\le  \P[I_{t} \le \gamma] \notag \\
		&\quad+ \P\Bigg[\bigcup_{\revise{\Wc_{\rm aMD}} \subset [K_{\rm a} - {K_{\rm a}'} + 1:K_{\rm a}] \atop |\revise{\Wc_{\rm aMD}}| = t} \bigcup_{\revise{\Wc_{\rm FA}} \subset [K_{\rm a}+1:M] \atop |\revise{\Wc_{\rm FA}}| = t}  \notag \\ &\qquad \qquad 
		\big\{\imath_t(c(\revise{\Wc_{\rm FA}});\rvVec{y} \cond c(\Wc \setminus \revise{\Wc_{\rm MD}})) > \gamma\big\}\Bigg]. \label{eq:tmp383_simp}
	\end{align}
	Here, \eqref{eq:tmp814_simp} follows by writing explicitly the event $\{|\revise{\Wc_{\rm aMD}}| = t\}$, and \eqref{eq:tmp383_simp} by relaxing the inequality inside the second probability.
	Using that $\P[\imath(x;\rv{y}) > \gamma] \le e^{-\gamma}, \forall x$~\cite[Cor.~17.1]{Polyanskiy2019lecture}, we obtain
	\begin{align}
		\P[\imath_t(c(\revise{\Wc_{\rm FA}});\rvVec{y} \cond c(\Wc \setminus \revise{\Wc_{\rm MD}})) > \gamma] \le e^{-\gamma}.
	\end{align}
	Then, by applying the union bound and taking the infimum over $\gamma$, we conclude that
	\begin{align}
		&\P[{|\revise{\Wc_{\rm aMD}}| = t}] \notag \\
		&\le \inf_{\gamma} \bigg( \P[\rv{I}_t \le \gamma] +  \binom{K_{\rm a}'}{t}  \binom{M-K_{\rm a}}{t}  e^{-\gamma} \bigg) \\ 	&
		\defeq q_{t,t}. \label{eq:tmp1201_simp}
	\end{align}
	This concludes the DT-based approach.
	
	It follows from \eqref{eq:tmp1148_simp} and \eqref{eq:tmp1201_simp} that
	$\P[|\revise{\Wc_{\rm aMD}}| = t] \le \min\left\{p_{t,t}, q_{t,t} \right\}$.
	Introducing this bound into \eqref{eq:tmp850_simp} and \eqref{eq:tmp853_simp}, we obtain that the \gls{MD} and \gls{FA} probabilities, averaged over the Gaussian codebook ensemble, are upper-bounded by $\epsilon_{\rm MD}$ and $\epsilon_{\rm FA}$ given in~\eqref{eq:eps_MD_simp} and \eqref{eq:eps_FA_simp}, respectively. 
	

	\subsection{The General Case}  \label{sec:general_case}
	We now explain how the result in the special case considered in the previous subsection can be extended to the general case where $\rv{K}_{\rm a}$ and $\rv{K}'_{\rm a}$ are random and $r \ge 0$. 
	For random $\rv{K}_{\rm a}$ and $\rv{K}'_{\rm a}$, one has to take into account all the possible combinations of the number of transmitted messages and decoded messages when computing the expectations in~\eqref{eq:pMD} and \eqref{eq:pFA}. Consider the event that $K_{\rm a}$ users are active and the estimation of $K_{\rm a}$ results in $K_{\rm a}'$, which we denote by $K_{\rm a} \to K_{\rm a}'$. As in the special case, we assume \gls{wlog} that $\widetilde{\Wc} = [{K}_{\rm a}]$. Furthermore, exploiting symmetry, we let $\revise{\Wc_{\rm iMD}} = [({K}_{\rm a} - \overline{K_{\rm a}'})^+]$ the list of $({K}_{\rm a} - \overline{K_{\rm a}'})^+$ initial \glspl{MD} due to insufficient decoded list size, and $\revise{\Wc_{\rm aMD}} = \revise{\Wc_{\rm MD}} \setminus \revise{\Wc_{\rm iMD}}$ denote the list of $t$ additional \glspl{MD} occurring during the decoding process.  
	Note also that,
	if $\underline{K_{\rm a}'} > K_{\rm a}$, the decoder always outputs more than $K_{\rm a}$ messages. Hence, at least $\underline{K_{\rm a}'} - K_{\rm a}$ decoded messages are falsely alarmed. Exploiting symmetry, we let \gls{wlog} $\revise{\Wc_{\rm iFA}} = [{K}_{\rm a} + 1: \underline{K_{\rm a}'}]$ the list of $(\underline{K_{\rm a}'}-{K}_{\rm a})^+$ initial \glspl{FA} due to excessive decoded list size, and $\revise{\Wc_{\rm aFA}} = \revise{\Wc_{\rm FA}} \setminus \revise{\Wc_{\rm iFA}}$ denote the list of $t'$ additional \glspl{FA} occurring during the decoding process. In Fig.~\ref{fig:venn}, we depict the relation between these sets of messages. Under these assumptions, $\revise{\Wc_{\rm aMD}}$ and $\revise{\Wc_{\rm aFA}}$ are generic subsets of $[({K}_{\rm a} - \overline{K_{\rm a}'})^+ + 1:{K}_{\rm a}]$  and $[\max\{{K}_{\rm a},\underline{K_{\rm a}'}\}+1 : M]$, respectively.
	
	\begin{figure}
		\centering
		\begin{tikzpicture}[thick,scale=.9, every node/.style={scale=.9}]
			\def\radius{2.15cm}
			\def\radiusB{0.9*\radius}
			\def\mycolorbox#1{\textcolor{#1}{\rule{2ex}{2ex}}}
			\colorlet{colori}{gray!80}
			\colorlet{colorii}{gray!20}
			
			\coordinate (ceni) at (0,0);
			\coordinate[xshift=1.05*\radius] (cenii);
			
			\coordinate (edge1a) at (-\radiusB,0.1cm);
			\coordinate(edge1b) at (\radiusB,-.2cm);
			
			\coordinate (edge2a) at (\radius-\radiusB-.1cm,.1cm);
			\coordinate (edge2b) at (\radius+\radiusB+.3cm,-.2cm);
			
			\draw[fill=colori,fill opacity=0.5] (ceni) circle (\radiusB);
			\draw[fill=colorii,fill opacity=0.5] (cenii) circle (\radius);
			\draw (ceni) circle (\radiusB);
			
			\draw (edge1a) to (edge2a);
			\draw (edge1b) to (edge2b);
			
			\draw[-latex] (-1.1*\radius,-0.6*\radius) node[below,xshift=-.4cm,text width=3cm,font=\linespread{1}\selectfont,align=center] {\small Distinct transmitted messages $\widetilde{\Wc}$} -- (-0.77*\radius,-0.5*\radius);
			\draw[-latex] (2.13*\radius,-0.6*\radius) node[below,text width=2cm,font=\linespread{1}\selectfont,align=center] {\small Decoded messages $\widehat{\Wc}$} -- (1.93*\radius,-0.5*\radius);
			
			\node[yshift=1.2*\radius,xshift=-1.25cm,text width=4.3cm,font=\linespread{1}\selectfont,align=center] {\small \glspl{MD}: \\ $\revise{\Wc_{\rm MD}} = \revise{\Wc_{\rm iMD}} \cup \revise{\Wc_{\rm aMD}}$ \\ $= \widetilde{\Wc} \setminus \widehat{\Wc}$};
			
			\node[yshift=1.2*\radius,xshift=1.5cm,text width=4.3cm,font=\linespread{1}\selectfont,align=center] at (cenii) {\small \glspl{FA}: \\ $\revise{\Wc_{\rm FA}} = \revise{\Wc_{\rm iFA}} \cup \revise{\Wc_{\rm aFA}}$ \\ $\ = \widehat{\Wc} \setminus \widetilde{\Wc}$};
			
			\node[xshift=1cm,text width=1.5cm,font=\linespread{1}\selectfont,align=center] at (ceni) {\small correctly decoded messages $\widetilde{\Wc} \cap \widehat{\Wc}$};
			\node[yshift=.43*\radius,xshift=-.68cm,text width=2.3cm,font=\linespread{1}\selectfont,align=center] at (ceni) {\small $~~~~({K}_{\rm a}-\overline{K_{\rm a}'})^+$ initial \glspl{MD} $~~~\revise{\Wc_{\rm iMD}}~~~$};
			\node[yshift=-.4*\radius,xshift=-.37*\radius,text width=1.8cm,font=\linespread{1}\selectfont,align=center] at (ceni) {\small $t$ additional \glspl{MD} $~~~~~~\revise{\Wc_{\rm aMD}}~~~$};
			\node[yshift=.4*\radius,xshift=.27*\radius,text width=1.7cm,font=\linespread{1}\selectfont,align=center] at (cenii) {\small $(\underline{K_{\rm a}'}-{K}_{\rm a})^+$ initial \glspl{FA} $~~~\revise{\Wc_{\rm iFA}}~~~$};
			\node[yshift=-1.2cm,xshift=.28*\radius,text width=1.9cm,font=\linespread{1}\selectfont,align=center] at (cenii) {\small $t'$ additional \glspl{FA} $~~~\revise{\Wc_{\rm aFA}}~~~$};
		\end{tikzpicture}
		\caption{A diagram depicting the relation between the sets of messages defined for the general case in Appendix~\ref{sec:general_case}.}
		\label{fig:venn}
	\end{figure}
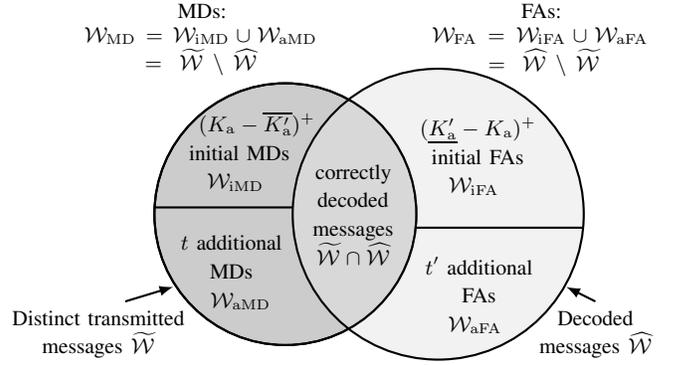
	
	Note that in the special case considered in Appendix~\ref{sec:special_case}, $t$ can take value from $0$ to $K'_{\rm a}$ while $t' = t$. In the general case, instead:
	\begin{itemize}
		\item The possible values of $t$ belong to the set $\Tc$ defined in~\eqref{eq:T}. This is because the number of \glspl{MD}, given by $t+{(K_{\rm a}-\overline{K_{\rm a}'})}^+$, is upper-bounded by the total number $K_{\rm a}$ of transmitted messages, and by $M-\underline{K_{\rm a}'}$ (since at least $\underline{K_{\rm a}'}$ messages are returned).
		
		\item Given $t$, the integer $t'$ takes value in $\overline{\Tc}_t$ defined in~\eqref{eq:Tbart} because: i) the decoded list size, given by $K_{\rm a} - t - {(K_{\rm a} - \overline{K_{\rm a}'})}^+ + t' + {(\underline{K_{\rm a}'}-K_{\rm a})}^+$, must be in $[\underline{K_{\rm a}'} : \overline{K_{\rm a}'}]$; ii) the number of \glspl{FA}, given by $t'+ (\underline{K_{\rm a}'}-K_{\rm a})^+$, is upper-bounded by the number $M-K_{\rm a}$ of messages that are not transmitted, and by the maximal number $\overline{K_{\rm a}'}$ of decoded messages.
		
		\item If the decoded list size is further required to be strictly positive, 
		then $t'$ takes value in $\Tc_t$ defined in~\eqref{eq:Tt}.
	\end{itemize}
	
	Using the above definitions, the best approximation of $\widetilde{\Wc}$ that the decoder can produce is $\revise{\Wc_{\rm iFA}} \cup \revise{\Wc_{\rm aMD}} \cup (\Wc \setminus \revise{\Wc_{\rm MD}})$, while the actual decoded list, under $\widetilde{\Wc} \to \widehat{\Wc}$, is $\revise{\Wc_{\rm iFA}} \cup \revise{\Wc_{\rm aFA}} \cup (\Wc \setminus \revise{\Wc_{\rm MD}})$. Therefore, $\widetilde{\Wc} \to \widehat{\Wc}$ implies that $\|\rvVec{y} - c(\revise{\Wc_{\rm iFA}}) - c(\revise{\Wc_{\rm aFA}}) - c(\Wc \setminus \revise{\Wc_{\rm MD}})\|^2 < \|\rvVec{y} - c(\revise{\Wc_{\rm iFA}}) - c(\revise{\Wc_{\rm aMD}}) - c(\Wc \setminus \revise{\Wc_{\rm MD}})\|^2$, which is equivalent to
	\begin{align}
		&\|c(\revise{\Wc_{\rm iMD}}) + c(\revise{\Wc_{\rm aMD}})- c(\revise{\Wc_{\rm iFA}}) - c(\revise{\Wc_{\rm aFA}}) + \rvVec{z}\|^2 \notag \\ 	&\qquad 
		< \|c(\revise{\Wc_{\rm iMD}}) - c(\revise{\Wc_{\rm iFA}}) + \rvVec{z}\|^2. \label{eq:eventF}
	\end{align}
	Let $F(\revise{\Wc_{\rm iMD}},\revise{\Wc_{\rm aMD}},\revise{\Wc_{\rm iFA}},\revise{\Wc_{\rm aFA}})$ denote the set of $(\revise{\Wc_{\rm iMD}},\revise{\Wc_{\rm aMD}},\revise{\Wc_{\rm iFA}},\revise{\Wc_{\rm aFA}})$ such that~\eqref{eq:eventF} holds.

	We now compute the expectations in $P_{\rm MD}$ given by~\eqref{eq:pMD} and $P_{\rm FA}$ given by~\eqref{eq:pFA}. 
	Given $|\revise{\Wc_{\rm aMD}}| = t$ and $|\revise{\Wc_{\rm aFA}}| = t'$, we have that $|\revise{\Wc_{\rm MD}}| = t+(\rv{K}_{\rm a} - \overline{K_{\rm a}'})^+$, $|\revise{\Wc_{\rm FA}}| = t + (\underline{K_{\rm a}'}-\rv{K}_{\rm a})^+$, and $|\widehat{\Wc}| = \rv{K}_{\rm a} - t - (\rv{K}_{\rm a} - \overline{K_{\rm a}'})^+ + t' + (\underline{K_{\rm a}'}-\rv{K}_{\rm a})^+$. 
	It follows from \eqref{eq:pMD} and \eqref{eq:pFA} that, after the change of measure in Appendix~\ref{sec:change_measure}, $P_{\rm MD}$ and $P_{\rm FA}$ can be bounded as
	\begin{align}
		P_{\rm MD} &\le \sum_{K_{\rm a} =\max\{K_\ell,1\}}^{K_{u}} \Bigg(P_{\rv{K}_{\rm a}}(K_{\rm a}) \sum_{K_{\rm a}' = K_\ell}^{K_{u}} \sum_{t\in \Tc} \notag \\	&\qquad \frac{t+(K_{\rm a}-\overline{K_{\rm a}'})^+}{K_{\rm a}} 
		\P[|\revise{\Wc_{\rm aMD}}| = t, K_{\rm a} \to K_{\rm a}'] \Bigg)
		+ \tilde{p},  \label{eq:tmp850}\\ 
		P_{\rm FA} &\le \sum_{K_{\rm a} =K_\ell}^{K_{u}} \Bigg(P_{\rv{K}_{\rm a}}(K_{\rm a}) \sum_{K_{\rm a}' = K_\ell}^{K_{u}}  \sum_{t\in \Tc} \sum_{t' \in \Tc_t} \notag \\	&\qquad \quad
		\frac{t'+(\underline{K_{\rm a}'} - K_{\rm a})^+}{K_{\rm a} - t - {(K_{\rm a} - \overline{K_{\rm a}'})}^+ + t' + {(\underline{K_{\rm a}'}-K_{\rm a})}^+} \notag \\
		&\qquad \quad \cdot \P[|\revise{\Wc_{\rm aMD}}| = t, |\revise{\Wc_{\rm aFA}}| = t', K_{\rm a} \to K_{\rm a}'] \Bigg) + \tilde{p}. \label{eq:tmp853}
	\end{align} 
	
	Next, we proceed to bound the joint probability $\P[|\revise{\Wc_{\rm aMD}}| = t,K_{\rm a} \to K_{\rm a}']$ in \eqref{eq:tmp850} and the joint probability $\P[|\revise{\Wc_{\rm aMD}}| = t,|\revise{\Wc_{\rm aFA}}| = t',K_{\rm a} \to K_{\rm a}']$ in \eqref{eq:tmp853}. Let
	$A(K_{\rm a},K_{\rm a}') \defeq 
	\{m(\rvVec{y},K_{\rm a}') > m(\rvVec{y},K), \forall K \ne K_{\rm a}'\}$.
	Since the event $K_{\rm a} \to K_{\rm a}'$ implies that $|\widehat{\Wc}| \in [\underline{K_{\rm a}'}:\overline{K_{\rm a}'}]$ and that $A(K_{\rm a},K_{\rm a}')$ occurs,
	we have 
	\begin{align}
		&\P[|\revise{\Wc_{\rm aMD}}| = t, K_{\rm a} \to K_{\rm a}'] \notag \\
		&\le \P[{{|\revise{\Wc_{\rm aMD}}| = t, |\widehat{\Wc}|\in [\underline{K_{\rm a}'}:\overline{K_{\rm a}'}}],A(K_{\rm a},K_{\rm a}')}] \\
		&\le \min\Big\{\P[{|\revise{\Wc_{\rm aMD}}| = t, |\widehat{\Wc}|\in [\underline{K_{\rm a}'}:\overline{K_{\rm a}'}]}],  \notag \\
        &\qquad \qquad  \P[A(K_{\rm a},K_{\rm a}')] \Big\}, \label{eq:tmp883}
	\end{align}
	where \eqref{eq:tmp883} follows from the fact that the joint probability is upper-bounded by each of the individual probabilities. 
	Similarly, it follows that
	\begin{align}
		&\P[|\revise{\Wc_{\rm aMD}}| = t,|\revise{\Wc_{\rm aFA}}| = t', K_{\rm a} \to K_{\rm a}'] \notag \\
		 &\le \min\Big\{\P[{|\revise{\Wc_{\rm aMD}}| = t, |\revise{\Wc_{\rm aFA}}| = t', |\widehat{\Wc}|\in [\underline{K_{\rm a}'}:\overline{K_{\rm a}'}]}], \notag \\ 	&\qquad\qquad
		\P[A(K_{\rm a},K_{\rm a}')] \Big\}. \label{eq:tmp1054}
	\end{align}
	We next present the bounds on the probabilities $\P[A(K_{\rm a},K_{\rm a}')]$, $\P[{|\revise{\Wc_{\rm aMD}}| = t,  |\widehat{\Wc}| \in [\underline{K_{\rm a}'}:\overline{K_{\rm a}'}]}]$, and $\P[{|\revise{\Wc_{\rm aMD}}| = t, |\revise{\Wc_{\rm aFA}}| = t',  |\widehat{\Wc}| \in [\underline{K_{\rm a}'}:\overline{K_{\rm a}'}]}]$.
	
	\subsubsection{Bound on $\P[A(K_{\rm a},K_{\rm a}')]$} \label{sec:bound_xi}
	We have
	\begin{align}
		&\P[A(K_{\rm a},K_{\rm a}')] \notag \\
		&= \P[m(\rvVec{y},K_{\rm a}') > m(\rvVec{y},K), \forall K \ne K_{\rm a}'] \\ 
		&\le\min_{K \in [K_\ell:K_u] \colon K \ne K_{\rm a}'}\P[m(\rvVec{y},K_{\rm a}') > m(\rvVec{y},K)] \\
		&= \xi(K_{\rm a},K_{\rm a}'). \label{eq:tmp1077}
	\end{align}
	Note that under the new measure, $\rvVec{y} \sim \Cc\Nc(\mathbf{0},(1+K_{\rm a} P')\Id_n)$. 
	
	\subsubsection{Bounds on $\P[{|\revise{\Wc_{\rm aMD}}| = t,  |\widehat{\Wc}| \in [\underline{K_{\rm a}'}:\overline{K_{\rm a}'}]}]$} \label{sec:bound_tMDs}
	As in Appendix~\ref{sec:special_case}, we follow two approaches to bound $\P[{|\revise{\Wc_{\rm aMD}}| = t,  |\widehat{\Wc}| \in [\underline{K_{\rm a}'}:\overline{K_{\rm a}'}]}]$. The first approach is based on error exponent analyses and the second approach is based on the DT bound. In the first approach, we write the event $\{|\revise{\Wc_{\rm aMD}}| = t, |\widehat{\Wc}| \in [\underline{K_{\rm a}'}:\overline{K_{\rm a}'}]\}$ as a union of the pairwise events and obtain
	\begin{align}
		&\P[{|\revise{\Wc_{\rm aMD}}| = t, |\widehat{\Wc}| \in [\underline{K_{\rm a}'}:\overline{K_{\rm a}'}]}] \notag \\
		&= \P\Bigg[\bigcup_{t' \in \overline{\Tc}_t} \bigcup_{\revise{\Wc_{\rm aMD}} \subset [(K_{\rm a} - \overline{K_{\rm a}'})^+ + 1:K_{\rm a}] \atop |\revise{\Wc_{\rm aMD}}| = t}
		\bigcup_{\revise{\Wc_{\rm aFA}} \subset [\max\{K_{\rm a},\underline{K_{\rm a}'}\}+1:M] \atop |\revise{\Wc_{\rm aFA}}| = t'}  \notag \\ 	&\qquad \qquad
		F(\revise{\Wc_{\rm iMD}},\revise{\Wc_{\rm aMD}},\revise{\Wc_{\rm iFA}},\revise{\Wc_{\rm aFA}}) \Bigg]. \label{eq:tmp901} 
	\end{align}
	Then, by applying the Chernoff bound, Gallager's $\rho$-trick, and Lemma~\ref{lem:chi2}  following similar steps as in Appendix~\ref{sec:1st_approach}, we obtain
	\begin{align}
		\P[{|\revise{\Wc_{\rm aMD}}| = t, |\widehat{\Wc}| \in [\underline{K_{\rm a}'}:\overline{K_{\rm a}'}]}] \le p_t \label{eq:tmp1148}
	\end{align}
	with $p_t$ given by~\eqref{eq:pt}.  In the second approach, we consider the channel law $P_{\rvVec{y} \cond c(\revise{\Wc_{\rm MD}}), c(\Wc \setminus \revise{\Wc_{\rm MD}})}$ with input $c(\revise{\Wc_{\rm MD}})$ and output $\rvVec{y}$ where $|\revise{\Wc_{\rm aMD}}| = t$. The corresponding information density $\imath_t(c(\revise{\Wc_{\rm MD}});\rvVec{y} \cond c(\Wc \setminus \revise{\Wc_{\rm MD}}))$ is defined in~\eqref{eq:infor_den}. Note that~\eqref{eq:eventF} is equivalent to 
    \begin{align}
    &\imath_t(c(\revise{\Wc_{\rm iFA}})+c(\revise{\Wc_{\rm aFA}});\rvVec{y} \cond c(\Wc \setminus \revise{\Wc_{\rm MD}})) \notag \\
    &> \imath_t(c(\revise{\Wc_{\rm iFA}}) + c(\revise{\Wc_{\rm aMD}});\rvVec{y} \cond c(\Wc \setminus \revise{\Wc_{\rm MD}})).
    \end{align}
	Then, by proceeding as in Appendix~\ref{sec:2nd_approach}, we obtain 
	\begin{align}
		\P[{|\revise{\Wc_{\rm aMD}}| = t, |\widehat{\Wc}| \in [\underline{K_{\rm a}'}:\overline{K_{\rm a}'}]}] \le q_t \label{eq:tmp1201}
	\end{align}
	with $q_t$ given by~\eqref{eq:qt}.
	
	\subsubsection{Bounds on $\P\Big[|\revise{\Wc_{\rm aMD}}| = t, |\revise{\Wc_{\rm aFA}}| = t', |\widehat{\Wc}| \in [\underline{K_{\rm a}'}:\overline{K_{\rm a}'}]\Big]$} 
	First, we have that 
	\begin{align}
		&\P[{|\revise{\Wc_{\rm aMD}}| = t, |\revise{\Wc_{\rm aFA}}| = t', |\widehat{\Wc}| \in [\underline{K_{\rm a}'}:\overline{K_{\rm a}'}]}] \notag \\
		&= \P\Bigg[\bigcup_{\revise{\Wc_{\rm aMD}} \subset [(K_{\rm a} - \overline{K_{\rm a}'})^+ + 1:K_{\rm a}] \atop |\revise{\Wc_{\rm aMD}}| = t} 
		\bigcup_{\revise{\Wc_{\rm aFA}} \subset [\max\{K_{\rm a},\underline{K_{\rm a}'}\}+1:M] \atop |\revise{\Wc_{\rm aMD}}| = t'} \notag \\ 	&\qquad \quad
		F(\revise{\Wc_{\rm iMD}},\revise{\Wc_{\rm aMD}},\revise{\Wc_{\rm iFA}},\revise{\Wc_{\rm aFA}})\Bigg]. \label{eq:tmp365}
	\end{align}
	Notice that the probability $\P[{|\revise{\Wc_{\rm aMD}}| = t, |\revise{\Wc_{\rm aFA}}| = t', |\widehat{\Wc}| \in \big[\underline{K_{\rm a}'}:\overline{K_{\rm a}'}\big]}]$ differs from the probability $\P[{|\revise{\Wc_{\rm aMD}}| = t, |\widehat{\Wc}| \in [\underline{K_{\rm a}'}:\overline{K_{\rm a}'}]}]$ in~\eqref{eq:tmp901} only in that the union over ${t'\in \overline{\Tc}_t}$ is absent. By applying the Chernoff bound, Gallager's $\rho$-trick, and Lemma~\ref{lem:chi2} following similar steps as in Appendix~\ref{sec:1st_approach}, we conclude that
	\begin{align}
		\P[{|\revise{\Wc_{\rm aMD}}| = t, |\revise{\Wc_{\rm aFA}}| = t', |\widehat{\Wc}| \in [\underline{K_{\rm a}'}:\overline{K_{\rm a}'}]}] \le  p_{t,t'} \label{eq:tmp1217}
	\end{align}
	with $p_{t,t'}$ given by~\eqref{eq:ptt}.
	Alternatively, bounding $\P\Big[|\revise{\Wc_{\rm aMD}}| = t, |\revise{\Wc_{\rm aFA}}| = t', |\widehat{\Wc}| \in [\underline{K_{\rm a}'}:\overline{K_{\rm a}'}]\Big]$ as in Appendix~\ref{sec:2nd_approach}, we obtain 
	\begin{align}
		\P[{|\revise{\Wc_{\rm aMD}}| = t, |\revise{\Wc_{\rm aFA}}| = t', |\widehat{\Wc}| \in [\underline{K_{\rm a}'}:\overline{K_{\rm a}'}]}] \le q_{t,t'} \label{eq:tmp1226}
	\end{align}
	with $p_{t,t'}$ given by~\eqref{eq:qtt}.
	
	\vspace{.3cm}
	It now follows from \eqref{eq:tmp883}, \eqref{eq:tmp1077}, \eqref{eq:tmp1148}, and \eqref{eq:tmp1201} that 
	\begin{align}
		\P[|\revise{\Wc_{\rm aMD}}| = t, K_{\rm a} \to K_{\rm a}'] \le \min\left\{p_t, q_t, \xi(K_{\rm a},K_{\rm a}') \right\}.
	\end{align}
	From \eqref{eq:tmp1054}, \eqref{eq:tmp1077}, \eqref{eq:tmp1217}, and \eqref{eq:tmp1226}, we obtain that 
	\begin{align}
		&\P[|\revise{\Wc_{\rm aMD}}| = t, |\revise{\Wc_{\rm aFA}}| = t', K_{\rm a} \to K_{\rm a}'] \notag \\ 
		&\le \min\left\{p_{t,t'}, q_{t,t'}, \xi(K_{\rm a},K_{\rm a}') \right\}.
	\end{align}
	Substituting these bounds on $\P[|\revise{\Wc_{\rm aMD}}| = t, K_{\rm a} \to K_{\rm a}']$ and $\P[|\revise{\Wc_{\rm aMD}}| = t, |\revise{\Wc_{\rm aFA}}| = t', K_{\rm a} \to K_{\rm a}']$ into \eqref{eq:tmp850} and \eqref{eq:tmp853}, we deduce that the \gls{MD} and \gls{FA} probabilities, averaged over the Gaussian codebook ensemble, are upper-bounded by $\epsilon_{\rm MD}$ and $\epsilon_{\rm FA}$ given in~\eqref{eq:eps_MD} and \eqref{eq:eps_FA}, respectively. Finally, by proceeding as in \cite[Th.~19]{Polyanskiy2011feedback}, one can show that there exists a randomized coding strategy that achieves \eqref{eq:eps_MD} and \eqref{eq:eps_FA} and involves time-sharing among at most three deterministic codes, as explained in Remark~\ref{remark:U}.

	\section{Proof of Theorem~\ref{thm:xi}}
	\label{proof:xi}
	Let $\rvVec{y}_0 \sim \Cc\Nc(\mathbf{0},(1+K_{\rm a}P')\Id_n)$.
	The probability density function of $\rvVec{y}_0$ is given by
	$
	p_{\rvVec{y}_0}(\yv_0) = \frac{1}{\pi^n (1+K_{\rm a} P')^n} \exp\left(-\frac{\|\yv_0\|^2}{1+K_{\rm a} P'}\right).
	$
	Therefore, with \gls{ML} estimation of $\rv{K}_{\rm a}$, we have that
	$
	m(\rvVec{y}_0,K) = \ln p_{\rvVec{y}_0}(\yv_0) = - \frac{\|\yv_0\|^2}{1+K P'} - n\ln(1+K P') - n \ln \pi.
	$
	As a consequence, the event $m\left(\rvVec{y}_0,K_{\rm a}'\right) > m\left(\rvVec{y}_0,K\right)$ can be written as $\frac{\|\rvVec{y}_0\|^2}{1+K_{\rm a}' P'} + n\ln(1+K_{\rm a}' P') < \frac{\|\rvVec{y}_0\|^2}{1+K P'} + n\ln(1+K P')$, or equivalently, 
	\begin{equation}
		\|\rvVec{y}_0\|^2 \left(\frac{1}{1+K_{\rm a}'P'} - \frac{1}{1+KP'}\right) < n \ln\left(\frac{1+KP'}{1+K_{\rm a}'P'}\right). \label{eq:eventKa}
	\end{equation}
	Using the fact that $\|\rvVec{y}_0\|^2$ follows a Gamma distribution with shape $n$ and scale $1+K_{\rm a} P'$, we deduce that $\xi(K_{\rm a},K_{\rm a}')$ is given by~\eqref{eq:xi_ML} 	with $
	\zeta(K,K_{\rm a},K_{\rm a}')$ given by~\eqref{eq:zeta_ML}.
	
	For energy-based estimation, i.e., $m(\yv,K) = -\big|\|\yv\|^2 - n(1 + KP') \big|$, we can show after some manipulations that the event  $m\left(\rvVec{y}_0,K_{\rm a}'\right) > m\left(\rvVec{y}_0,K\right)$ is equivalent to
	\begin{align}
		\begin{cases}
			\|\rvVec{y}_0\|^2 > n\left(1 + \frac{K_{\rm a} + K_{\rm a}'}{2}P'\right), &\text{if~} K_{\rm a}' < K_{\rm a}, \\
			\|\rvVec{y}_0\|^2 < n\left(1 + \frac{K_{\rm a} + K_{\rm a}'}{2}P'\right), &\text{if~} K_{\rm a}' > K_{\rm a}. 
		\end{cases}
	\end{align}
	Thus, using that $\|\rvVec{y}_0\|^2$ is Gamma distributed, we deduce that $\xi(K_{\rm a},K_{\rm a}')$ is given by~\eqref{eq:xi_ML} with $\zeta(K,K_{\rm a},K_{\rm a}')$ given by~\eqref{eq:zeta_energy}.

	\section{Proof of \revisee{Corollary}~\ref{cor:error_floor}} \label{proof:error_floor}
	\revisee{We evaluate the bounds $\epsilon_{\rm MD}$ and $\epsilon_{\rm FA}$ given in~\eqref{eq:eps_MD} and~\eqref{eq:eps_FA}, respectively, in the limit $P\to\infty$.}  First, the optimal value of $P'$ minimizing \revisee{these} bounds must grow with $P$ since otherwise $\tilde{p}$ will be large. Therefore, as $P\to \infty$, we can assume without loss of optimality that $P' \to \infty$. Next, when $t = t' = 0$, we can verify that $a = b = 0$, thus $E_0(\rho,\rho_1) = 0$ and $E(0,0) = 0$, achieved with $\rho = \rho_1 = 0$. Therefore, $p_0 = p_{0,0} = e^{-n\cdot 0} = 1$. We can also verify that $q_0$ and $q_{0,0}$ both converge to $1$ as $P' \to \infty$.
	When $P' \to \infty$, $\xi(K_{\rm a},K_{\rm a}')$ given in Theorem~\ref{thm:xi} converges to the right-hand side of \eqref{eq:xi_ML} with $\zeta(K,K_{\rm a},K_{\rm a}') = n \ln\big(\frac{K}{K_{\rm a}'}\big) K_{\rm a}^{-1}\big(\frac{1}{K_{\rm a}'} - \frac{1}{K}\big)^{-1}$ for ML estimation of $\rv{K}_{\rm a}$ and $\zeta(K,K_{\rm a},K_{\rm a}') = n\frac{K+K_{\rm a}'}{2 K_{\rm a}}$ for energy-based estimation of $\rv{K}_{\rm a}$. Furthermore, the last term in $\tilde{p}$ given by~\eqref{eq:p0} vanishes and thus $\tilde{p} \to \bar{p}$. Finally, the lower bounds $\bar{\epsilon}_{\rm MD}$ and $\bar{\epsilon}_{\rm FA}$ follows by substituting the asymptotic values of $p_0$, $q_0$, $p_{0,0}$, $q_{0,0}$, $\xi(K_{\rm a},K_{\rm a}')$, and $\tilde{p}$ computed above into $\epsilon_{\rm MD}$ and $\epsilon_{\rm FA}$, and by setting $\min\{p_t,q_t\}$ to zero for $t \ne 0$, and setting $\min\{p_{t,t'},q_{t,t'}\}$ to zero for $(t,t')\ne (0,0)$. 
	
	\section{Proof of Theorem~\ref{th:converse}} \label{proof:converse}
	The \gls{MD} and \gls{FA} probabilities are computed as $P_{\rm MD}=\E[\frac{|\Wc_{\rm MD}|}{|\widetilde{\Wc}|}]$ and $P_{\rm FA} = \E[\fa]$, respectively.  As derived in \eqref{eq:P_msg_collision} in Appendix~\ref{sec:change_measure}, the probability that at least two active users choose the same message to transmit is given by $\P[|\widetilde{\Wc}| = \rv{K}_{\rm a}] = \E_{\rv{K}_{\rm a}}\left[\frac{M!}{M^{\rv{K}_{\rm a}}(M-\rv{K}_{\rm a})!} \right]$. We have that
	\begin{align}
		P_{\rm MD} &= \P[|\widetilde{\Wc}| = \rv{K}_{\rm a}] \E[\frac{|\Wc_{\rm MD}|}{|\widetilde{\Wc}|} \; \bigg\vert\; |\widetilde{\Wc}| = \rv{K}_{\rm a}] \notag \\
        &\quad +  \P[|\widetilde{\Wc}| < \rv{K}_{\rm a}] \E[\frac{|\Wc_{\rm MD}|}{|\widetilde{\Wc}|} \; \bigg\vert\;  |\widetilde{\Wc}| < \rv{K}_{\rm a}] \\
		&\ge \E_{\rv{K}_{\rm a}}\left[\frac{M!}{M^{\rv{K}_{\rm a}}(M-\rv{K}_{\rm a})!} \right] \E[\frac{|\Wc_{\rm MD}|}{|\widetilde{\Wc}|} \; \bigg\vert\; |\widetilde{\Wc}| = \rv{K}_{\rm a}]. \label{eq:tmp501}
	\end{align}
	Similarly, it follows that
	\begin{align}
		P_{\rm FA} \ge \E_{\rv{K}_{\rm a}}\left[\frac{M!}{M^{\rv{K}_{\rm a}}(M-\rv{K}_{\rm a})!} \right] \E[\fa \; \bigg\vert\; |\widetilde{\Wc}| = \rv{K}_{\rm a}]. \label{eq:tmp505}
	\end{align}
	Denote the event that the estimation step outputs $K'_{\rm a}$ when $K_{\rm a}$ users are active, which we denote by $K_{\rm a} \to K'_{\rm a}$. Under this event and the condition in the right-hand side of~\eqref{eq:tmp501} and~\eqref{eq:tmp505}, we have that $|\widetilde{\Wc}| = K_{\rm a}$ and $\underline{K_{\rm a}'} \le |\widehat{\Wc}|\le \overline{K_{\rm a}'}$. The second expectations in the right-hand side of~\eqref{eq:tmp501} and~\eqref{eq:tmp505} can be expanded as
	\begin{align}
		&\E[\md \; \bigg\vert\; |\widetilde{\Wc}| = \rv{K}_{\rm a}] \notag \\
        &=  \sum_{K_{\rm a} \in \Dc, \; K_{\rm a} > 0} \bigg(P_{\rv{K}_{\rm a}}(K_{\rm a}) \sum_{K_{\rm a}' = K_\ell}^{K_u} \frac{|\widetilde{\Wc} \setminus \widehat{\Wc}|}{K_{\rm a}} \P[K_{\rm a} \to K_{\rm a}] \bigg),  \label{eq:tmp508296}\\
		&\E[\fa \; \bigg\vert\; |\widetilde{\Wc}| = \rv{K}_{\rm a}] \notag \\
        &= \sum_{K_{\rm a} \in \Dc} \bigg(P_{\rv{K}_{\rm a}}(K_{\rm a})	\sum_{K_{\rm a}' =K_\ell}^{K_u} \frac{|\widehat{\Wc} \setminus \widetilde{\Wc}|}{|\widehat{\Wc}|} \P[K_{\rm a} \to K_{\rm a}'] \bigg) \notag \\
		&\ge \sum_{K_{\rm a} \in \Dc} \bigg(P_{\rv{K}_{\rm a}}(K_{\rm a})	\sum_{K_{\rm a}' =K_\ell}^{K_u} \frac{|\widehat{\Wc} \setminus \widetilde{\Wc}|}{\overline{K_{\rm a}'}} \P[K_{\rm a} \to K_{\rm a}'] \bigg), \label{eq:tmp509275}
	\end{align}
	where $\Dc$ is the domain of $\rv{K}_{\rm a}$. Furthermore, it is straightforward that 
	\begin{align}
		|\widetilde{\Wc} \setminus \widehat{\Wc}| \ge (|\widetilde{\Wc}| - |\widehat{\Wc}|)^+ \ge (K_{\rm a} - \overline{K_{\rm a}'})^+, \label{eq:tmp508} \\
		|\widehat{\Wc} \setminus \widetilde{\Wc}| \ge (|\widehat{\Wc}| - |\widetilde{\Wc}|)^+ \ge (\underline{K_{\rm a}'} - K_{\rm a})^+. \label{eq:tmp509}
	\end{align}
	The first inequality in~\eqref{eq:tmp508} become equality if $\widehat{\Wc} \subset \widetilde{\Wc}$, i.e., all decoded messages have been transmitted. The first inequality in~\eqref{eq:tmp509} becomes equality if $\widetilde{\Wc} \subset \widehat{\Wc}$, i.e., all transmitted messages are decoded. In short, the first inequalities  in~\eqref{eq:tmp508} and~\eqref{eq:tmp509} become equalities if the MDs and FAs are caused by the mismatch between $K_{\rm a}$ and $K_{\rm a}'$ only. The second inequalities  in~\eqref{eq:tmp508}  and~\eqref{eq:tmp509} are due to $\underline{K_{\rm a}'} \le |\widehat{\Wc}|\le \overline{K_{\rm a}'}$. Next, by substituting~\eqref{eq:tmp508} into~\eqref{eq:tmp508296} and~\eqref{eq:tmp509} into~\eqref{eq:tmp509275}, by ignoring the terms for $K_{\rm a} \notin [K_\ell:K_u]$, and by averaging over $P_{\rv{x}}$, we obtain the converse bounds~\eqref{eq:converse_MD} and~\eqref{eq:converse_FA}.

	\section{Proof of Theorem~\ref{thm:TIN}} \label{app:TIN}
We first apply the same change of measure as in Appendix~\ref{sec:change_measure} 
at a cost of adding the term $\tilde{p}$. Then $P_{\rm MD}$ and $P_{\rm FA}$ can be bounded as
\begin{align}
	&P_{\rm MD} \notag \\
	&\leq \sum_{K_{\rm a} = \max\{K_\ell,1\}}^{K_{u}} \sum_{K_{\rm a}^{\prime} = K_\ell}^{K_{u}} \frac{P_{\rv{K}_{\rm a}}(K_{\rm a})}{K_{\rm a}}  \sum\limits_{j=1}^{K_{\rm a}} \mathbb{P}\left[\rv{w}_j \notin \widehat{\Wc},K_{\rm a} \to K_{\rm a}^{\prime}\right] \notag \\
    &\quad + \tilde{p} \\
	&\leq \sum_{K_{\rm a} = \max\{K_\ell,1\}}^{K_{u}} \sum_{K_{\rm a}^{\prime} = K_\ell}^{K_{u}} \frac{P_{\rv{K}_{\rm a}}(K_{\rm a})}{K_{\rm a}}  \sum\limits_{j=1}^{K_{\rm a}} \notag \\
    &\qquad \quad \min\left\{\P[\rv{w}_j \notin \widehat{\Wc}],\P[K_{\rm a}\to K_{\rm a}^{\prime}]\right\} \notag \\
    &\quad + \tilde{p}, \label{eq:md_alt_0}\\ 
	&P_{\rm FA} \notag \\
	&\leq \sum_{K_{\rm a} =K_\ell}^{K_{u}} \sum_{K_{\rm a}^{\prime} =K_\ell}^{K_{u}} \frac{P_{\rv{K}_{\rm a}}(K_{\rm a})}{K_{\rm a}^{\prime}} \sum\limits_{j=1}^{K_{\rm a}^{\prime}} \P[\widehat{\rv{w}}_j \notin \mathcal{W},K_{\rm a}\to K_{\rm a}^{\prime} ] \notag \\
    &\quad + \tilde{p}, \\
	&\leq \sum_{K_{\rm a} =K_\ell}^{K_{u}} \sum_{K_{\rm a}^{\prime} =K_\ell}^{K_{u}} \frac{P_{\rv{K}_{\rm a}}(K_{\rm a})}{K_{\rm a}^{\prime}} \sum\limits_{j=1}^{K_{\rm a}^{\prime}} \notag \\
    &\qquad \quad \min\left\{\P[\widehat{\rv{w}}_j \notin \mathcal{W}],\P[K_{\rm a}\to K_{\rm a}^{\prime} ]\right\} \notag \\
    &\quad + \tilde{p}. \label{eq:fa_alt_0}
\end{align}

Given that the decoded list size is $K_{\rm a}'$, if ${K_{\rm a}'} < K_{\rm a}$, the decoder commits at least $K_{\rm a} - {K_{\rm a}'}$ \glspl{MD}. On the contrary, if $K_{\rm a}' > K_{\rm a}$, the decoder commits at least ${K_{\rm a}'} - K_{\rm a}$ \glspl{FA}. We assume \gls{wlog} that the last $({K}_{\rm a} - {K_{\rm a}'})^+$ transmitted codewords are initially misdetected due to insufficient decoded list size, and that the last $({K_{\rm a}'}-{K}_{\rm a})^+$ decoded codewords are initially false-alarmed  due to excessive decoded list size. Then $\P\big[\rv{w}_j \notin \widehat{\Wc}\big] = 1$, $\forall j \in [\min\{K_{\rm a},K_{\rm a}'\} + 1: K_{\rm a}]$, and $\P[\widehat{\rv{w}}_j \notin \mathcal{W}] = 1$, $\forall j\in [\min\{K_{\rm a},K_{\rm a}'\} + 1: K_{\rm a}']$. As a consequence, we can write \eqref{eq:md_alt_0} and \eqref{eq:fa_alt_0} as
\begin{align}
	P_{\rm MD} &\leq \sum_{K_{\rm a} =\max\{K_\ell,1\}}^{K_{u}} \sum_{K_{\rm a}^{\prime} = K_\ell}^{K_{u}} \frac{P_{\rv{K}_{\rm a}}(K_{\rm a})}{K_{\rm a}} \notag \\
    &\quad \cdot \Biggl( (K_{\rm a}-K_{\rm a}^{\prime})^+ \P[K_{\rm a}\to K_{\rm a}^{\prime}]  \notag \\
	&\qquad + \sum\limits_{j=1}^{\min\{K_{\rm a},K_{\rm a}^{\prime}\}} \min\left\{\P[\rv{w}_j \notin \widehat{\Wc}],\P[K_{\rm a}\to K_{\rm a}^{\prime}]\right\}\Biggr) \notag \\
    &\quad + \tilde{p}, \label{eq:md_alt2_0}\\
	P_{\rm FA} &\leq \sum_{K_{\rm a} =K_\ell}^{K_{u}} \sum_{K_{\rm a}^{\prime} = K_\ell}^{K_{u}} \frac{P_{\rv{K}_{\rm a}}(K_{\rm a})}{K_{\rm a}^{\prime}} \Biggl((K_{\rm a}^{\prime}-K_{\rm a})^+\P[K_{\rm a}\to K_{\rm a}^{\prime}]  \notag \\
	&\qquad + \sum\limits_{j=1}^{\min\{K_{\rm a},K_{\rm a}^{\prime}\}} \min\left\{\P[\widehat{\rv{w}}_j \notin {\Wc}],\P[K_{\rm a}\to K_{\rm a}^{\prime}]\right\}\Biggr) \notag \\
    &\quad + \tilde{p}.\label{eq:fa_alt2_0}
\end{align}
We proceed to bound the probabilities $\P\big[\rv{w}_j \notin \widehat{\Wc}\big]$ and $\P[\widehat{\rv{w}}_j \notin \widetilde{\mathcal{W}}]$ for $j\in [\min\{K_{\rm a},K_{\rm a}'\}]$.

\subsubsection{Bound on $\P\big[\rv{w}_j \notin \widehat{\Wc}\big]$ for $j\in [\min\{K_{\rm a},K_{\rm a}'\}]$}
The event $\rv{w}_j \notin \widehat{\Wc}$ 
implies that at least one codeword that was not transmitted is closer to $\rvVec{y}$ than $\cv_{\rv{w}_j}$. 
%
Therefore, for a given codebook $\{\cv_1,\dots,\cv_M\}$, $\P\big[\rv{w}_j \notin \widehat{\Wc}\big]$ can be upper-bounded as
\begin{align}
	&\P\big[\rv{w}_j \notin \widehat{\Wc}\big] \notag \\
	&\le \P[{\exists i \in [M] \setminus \Wc \colon \lVert\rvVec{y}-\cv_{i}\rVert^2\leq \lVert\rvVec{y}-\cv_{\rv{w}_j}\rVert^2 }] \\
	&= \P[\bigcup\limits_{i \in [M] \setminus \Wc}\left\{\lVert\rvVec{y}-\cv_{i}\rVert^2\leq \lVert\rvVec{y}-\cv_{\rv{w}_j}\rVert^2\right\}]. \label{eq:tmp2070_0}
\end{align}
By applying the union bound on~\eqref{eq:tmp2070_0}, using the fact that the codewords are i.i.d., we conclude that 
\begin{align}
	\P&\big[\rv{w}_j \notin \widehat{\Wc}\big] \le \notag \\
    &\E_{\rvVec{y}, \rvVec{c}}\left[ \min\left\{1, (M-K_{\rm a}) \P[\lVert\rvVec{y}-\tilde{\rvVec{c}}\rVert^2\leq \lVert\rvVec{y}-\rvVec{c}\rVert^2 \cond \rvVec{y}, \rvVec{c}]  \right\}\right],
\end{align}
where $\{\rvVec{y}, \rvVec{c}, \tilde{\rvVec{c}}\}$ has the same joint distribution as $\{\rvVec{y}, \cv_{\rv{w}_j}, \cv_i\}$, $i \in [M] \setminus \Wc$. 
Next, by applying the Chernoff bound and proceeding as in~\cite[App.~A]{Ostman20}, we obtain the following RCUs bound~\cite[Th. 16]{Polyanskiy2010} for every $s>0$:
\begin{align} 
	&\P[\rv{w}_j \notin \widehat{\Wc}] \le \notag \\
    &~\E[ \min\left\{1, \exp\bigg(\ln(M-K_{\rm a}) - \sum\limits_{i=1}^n  \imath_s(\rv{x}_{i};\rv{y}_i)\bigg)\right\}]\label{eq:tmp2154_0}
\end{align}
with $\imath_s(x;y)$ given by~\eqref{eq:gen_infor_den} and $\{\rv{x}_i,\rv{y}_i\}$ defined in Theorem~\ref{thm:TIN}. Finally, by observing that, for every
positive random variable $\rv{v}$, it holds that $\E[\min\{1, \rv{v}\}] =
\P[\rv{v} \ge \rv{u}]$ where $\rv{u}$ is uniformly distributed on $[0,1]$, we obtain that the right-hand side of~\eqref{eq:tmp2154_0} is given by $\eta_s$ defined in~\eqref{eq:eta_s}.

\subsubsection{Bound on $\P[\widehat{\rv{w}}_j \notin \widetilde{\Wc}]$ for $j\in [\min\{K_{\rm a},K_{\rm a}'\}]$}
The event $\widehat{\rv{w}}_j \notin \widetilde{\Wc} = \{\rv{w}_1,\dots,\rv{w}_{K_{\rm a}}\}$ implies that $\widehat{\rv{w}}_j$ is closer to $\rvVec{y}$ than at least one transmitted codeword. 
In this case, we assume \gls{wlog} that $\|\rvVec{y} - \cv_{\widehat{\rv{w}}_j}\|^2 \le \|\rvVec{y} - \cv_{\rv{w}_1}\|^2$. It follows that
\begin{align}
	\P[\widehat{\rv{w}}_j \notin \widetilde{\Wc}] 
	&\le \P[\bigcup_{\widehat{\rv{w}}_j \in [M] \setminus \Wc}  \{\|\rvVec{y} - \cv_{\widehat{\rv{w}}_j}\|^2 \le \|\rvVec{y} - \cv_{\rv{w}_1}\|^2\} ]. \label{eq:tmp2098_0}
\end{align}
Next, by applying the union bound, 
then applying the Chernoff bound and proceeding as in \cite[App.~A]{Ostman20}, we deduce that $\P[\widehat{\rv{w}}_j \notin \widetilde{\Wc}]$ is upper-bounded by the right-hand side of \eqref{eq:tmp2154_0}, which can be expressed as $\eta_s$ defined in~\eqref{eq:eta_s}.

\vspace{.2cm}
Furthermore, the probability $\P[K_{\rm a}\to K_{\rm a}^{\prime}]$ can be upper-bounded by $\xi(K_{\rm a},K_{\rm a}')$ as in Appendix~\ref{sec:bound_xi}. By substituting the bounds on $\P\big[\rv{w}_j \notin \widehat{\Wc}\big]$, $\P[\widehat{\rv{w}}_j \notin \widetilde{\Wc}]$, and $\P[K_{\rm a}\to K_{\rm a}^{\prime}]$ into~\eqref{eq:md_alt2_0} and~\eqref{eq:fa_alt2_0}, we deduce that $P_{\rm MD}$ and $P_{\rm FA}$, averaged over the Gaussian codebook ensemble, are upper-bounded by $\epsilon_{\rm MD}$ and $\epsilon_{\rm FA}$ given in~\eqref{eq:md_TIN} and \eqref{eq:fa_TIN}, respectively. Finally, by proceeding as in \cite[Th.~19]{Polyanskiy2011feedback}, one can show that there exists a randomized coding strategy that achieves \eqref{eq:eps_MD} and \eqref{eq:eps_FA} simultaneously, and involves time-sharing among at most three deterministic codes.

\bibliographystyle{IEEEtran}
\bibliography{IEEEabrv,./biblio}

\begin{thebibliography}{10}
\providecommand{\url}[1]{#1}
\csname url@samestyle\endcsname
\providecommand{\newblock}{\relax}
\providecommand{\bibinfo}[2]{#2}
\providecommand{\BIBentrySTDinterwordspacing}{\spaceskip=0pt\relax}
\providecommand{\BIBentryALTinterwordstretchfactor}{4}
\providecommand{\BIBentryALTinterwordspacing}{\spaceskip=\fontdimen2\font plus
\BIBentryALTinterwordstretchfactor\fontdimen3\font minus
  \fontdimen4\font\relax}
\providecommand{\BIBforeignlanguage}[2]{{%
\expandafter\ifx\csname l@#1\endcsname\relax
\typeout{** WARNING: IEEEtran.bst: No hyphenation pattern has been}%
\typeout{** loaded for the language `#1'. Using the pattern for}%
\typeout{** the default language instead.}%
\else
\language=\csname l@#1\endcsname
\fi
#2}}
\providecommand{\BIBdecl}{\relax}
\BIBdecl

\bibitem{Ngo2021ISITmassive}
K.-H. Ngo, A.~Lancho, G.~Durisi, and A.~{Graell i Amat}, ``Massive
  uncoordinated access with random user activity,'' in \emph{Proc. {IEEE} Int.
  Symp. Inf. Theory {(ISIT)}}, Melbourne, Victoria, Australia, Jul. 2021, pp.
  3014--3019.

\bibitem{Chen2020_massiveAccess}
X.~{Chen}, D.~W.~K. {Ng}, W.~{Yu}, E.~G. {Larsson}, N.~{Al-Dhahir}, and
  R.~{Schober}, ``Massive access for {5G} and beyond,'' \emph{IEEE J. Sel.
  Areas Commun.}, vol.~39, no.~3, pp. 615--637, Mar. 2021.

\bibitem{Wu2020_massiveAccess}
Y.~{Wu}, X.~{Gao}, S.~{Zhou}, W.~{Yang}, Y.~{Polyanskiy}, and G.~{Caire},
  ``Massive access for future wireless communication systems,'' \emph{{IEEE}
  Wireless Commun. Mag.}, vol.~27, no.~4, pp. 148--156, Oct. 2020.

\bibitem{PolyanskiyISIT2017massive_random_access}
Y.~{Polyanskiy}, ``A perspective on massive random-access,'' in \emph{Proc.
  {IEEE} Int. Symp. Inf. Theory {(ISIT)}}, Aachen, Germany, Jun. 2017, pp.
  2523--2527.

\bibitem{Berioli2016NOW}
M.~Berioli, G.~Cocco, G.~Liva, and A.~Munari, ``Modern random access
  protocols,'' \emph{Foundations and Trends in Networking}, vol.~10, no.~4, pp.
  317--446, Nov. 2016.

\bibitem{Ordentlich2017low_complexity_random_access}
O.~{Ordentlich} and Y.~{Polyanskiy}, ``Low complexity schemes for the random
  access {G}aussian channel,'' in \emph{Proc. {IEEE} Int. Symp. Inf. Theory
  {(ISIT)}}, Aachen, Germany, Jun. 2017, pp. 2528--2532.

\bibitem{Vem2019}
A.~{Vem}, K.~R. {Narayanan}, J.~{Chamberland}, and J.~{Cheng}, ``A
  user-independent successive interference cancellation based coding scheme for
  the unsourced random access {G}aussian channel,'' \emph{IEEE Trans. Commun.},
  vol.~67, no.~12, pp. 8258--8272, Sep. 2019.

\bibitem{Fengler2019sparcs}
A.~Fengler, P.~Jung, and G.~Caire, ``{SPARCs} for unsourced random access,''
  \emph{IEEE Trans. Inf. Theory}, vol.~67, no.~10, pp. 6894--6915, Oct. 2021.

\bibitem{Amalladinne2020}
V.~K. {Amalladinne}, J.~F. {Chamberland}, and K.~R. {Narayanan}, ``A coded
  compressed sensing scheme for unsourced multiple access,'' \emph{IEEE Trans.
  Inf. Theory}, vol.~66, no.~10, pp. 6509--6533, Jul. 2020.

\bibitem{Amalladinne2021}
V.~K. Amalladinne, A.~K. Pradhan, C.~Rush, J.-F. Chamberland, and K.~R.
  Narayanan, ``Unsourced random access with coded compressed sensing:
  {I}ntegrating {AMP} and belief propagation,'' \emph{IEEE Trans. Inf. Theory},
  vol.~68, no.~4, pp. 2384--2409, Dec. 2021.

\bibitem{Pradhan2020}
A.~K. Pradhan, V.~K. Amalladinne, K.~R. Narayanan, and J.-F. Chamberland,
  ``Polar coding and random spreading for unsourced multiple access,'' in
  \emph{Proc. {IEEE} Int. Conf. Communications {(ICC)}}, Dublin, Ireland, Jun.
  2020, pp. 1--6.

\bibitem{Han2021}
Z.~Han, X.~Yuan, C.~Xu, S.~Jiang, and X.~Wang, ``Sparse {K}ronecker-product
  coding for unsourced multiple access,'' \emph{IEEE Wireless Commun. Lett.},
  vol.~10, no.~10, pp. 2274--2278, Oct. 2021.

\bibitem{Kowshik2020}
S.~S. {Kowshik}, K.~{Andreev}, A.~{Frolov}, and Y.~{Polyanskiy}, ``Energy
  efficient coded random access for the wireless uplink,'' \emph{IEEE Trans.
  Commun.}, vol.~68, no.~8, pp. 4694--4708, Jun. 2020.

\bibitem{Fengler2019nonBayesian}
A.~Fengler, S.~Haghighatshoar, P.~Jung, and G.~Caire, ``Non-{B}ayesian activity
  detection, large-scale fading coefficient estimation, and unsourced random
  access with a massive {MIMO} receiver,'' \emph{IEEE Trans. Inf. Theory},
  vol.~67, no.~5, pp. 2925--2951, May 2021.

\bibitem{Shyianov2021}
V.~Shyianov, F.~Bellili, A.~Mezghani, and E.~Hossain, ``Massive unsourced
  random access based on uncoupled compressive sensing: {A}nother blessing of
  massive {MIMO},'' \emph{{IEEE} J. Select. Areas Commun.}, vol.~39, no.~3, pp.
  820--834, Mar. 2021.

\bibitem{Stern2019}
K.~Stern, A.~E. Kal{\o}r, B.~Soret, and P.~Popovski, ``Massive random access
  with common alarm messages,'' in \emph{Proc. {IEEE} Int. Symp. Inf. Theory
  {(ISIT)}}, Paris, France, Jul. 2019, pp. 1--5.

\bibitem{sendelbach2013alarm}
S.~Sendelbach and M.~Funk, ``Alarm fatigue: a patient safety concern,''
  \emph{AACN advanced critical care}, vol.~24, no.~4, pp. 378--386, Oct.-Dec.
  2013.

\bibitem{ruskin2015alarm}
K.~J. Ruskin and D.~Hueske-Kraus, ``Alarm fatigue: impacts on patient safety,''
  \emph{Current Opinion in Anesthesiology}, vol.~28, no.~6, pp. 685--690, Dec.
  2015.

\bibitem{Effros2018ISIT}
M.~{Effros}, V.~{Kostina}, and R.~C. {Yavas}, ``Random access channel coding in
  the finite blocklength regime,'' in \emph{Proc. {IEEE} Int. Symp. Inf. Theory
  {(ISIT)}}, Vail, Colorado, USA, Jun. 2018, pp. 1261--1265.

\bibitem{Yavas2021Gaussian}
R.~C. Yavas, V.~Kostina, and M.~Effros, ``Gaussian multiple and random access
  channels: {F}inite-blocklength analysis,'' \emph{IEEE Trans. Inf. Theory},
  vol.~67, no.~11, pp. 6983--7009, Nov. 2021.

\bibitem{Yavas2021}
------, ``Random access channel coding in the finite blocklength regime,''
  \emph{IEEE Trans. Inf. Theory}, vol.~67, no.~4, pp. 2115--2140, Apr. 2021.

\bibitem{Liva2011IRSA}
G.~{Liva}, ``Graph-based analysis and optimization of contention resolution
  diversity slotted {ALOHA},'' \emph{IEEE Trans. Commun.}, vol.~59, no.~2, pp.
  477--487, Dec. 2011.

\bibitem{Decurninge2020}
A.~Decurninge, I.~Land, and M.~Guillaud, ``Tensor-based modulation for
  unsourced massive random access,'' \emph{IEEE Wireless Commun. Lett.},
  vol.~10, no.~3, pp. 552--556, Mar. 2021.

\bibitem{fengler2020pilot}
A.~Fengler, P.~Jung, and G.~Caire, ``Pilot-based unsourced random access with a
  massive {MIMO} receiver, {MRC} and polar codes,'' \emph{arXiv preprint
  arXiv:2012.03277}, 2020.

\bibitem{Polyanskiy2011feedback}
Y.~{Polyanskiy}, H.~V. {Poor}, and S.~{Verdu}, ``Feedback in the non-asymptotic
  regime,'' \emph{IEEE Trans. Inf. Theory}, vol.~57, no.~8, pp. 4903--4925,
  Jul. 2011.

\bibitem{Gallager1968information}
R.~G. Gallager, \emph{Information theory and reliable communication}.\hskip 1em
  plus 0.5em minus 0.4em\relax New York, USA: Wiley, 1968.

\bibitem{Polyanskiy2019lecture}
\BIBentryALTinterwordspacing
Y.~Polyanskiy and Y.~Wu, ``Lecture notes on information theory,''
  \emph{Massachusetts Institute of Technology (6.441), University of Illinois
  Urbana-Champaign (ECE563), and Yale University (STAT 664)}, 2012-2017.
  [Online]. Available:
  \url{http://people.lids.mit.edu/yp/homepage/data/itlectures\_v5.pdf}
\BIBentrySTDinterwordspacing

\bibitem{Zadik2019}
I.~Zadik, Y.~Polyanskiy, and C.~Thrampoulidis, ``Improved bounds on {G}aussian
  {MAC} and sparse regression via {G}aussian inequalities,'' in \emph{Proc.
  {IEEE} Int. Symp. Inf. Theory {(ISIT)}}, Paris, France, Jul. 2019, pp.
  430--434.

\bibitem{Ostman20}
J.~{\"O}stman, A.~Lancho, G.~Durisi, and L.~Sanguinetti, ``{URLLC} with massive
  {MIMO: A}nalysis and design at finite blocklength,'' \emph{{IEEE} Trans.
  Wireless Commun.}, vol.~20, no.~10, pp. 6387--6401, Oct. 2021.

\bibitem{Scarlett2017}
J.~Scarlett, V.~Y.~F. Tan, and G.~Durisi, ``The dispersion of nearest-neighbor
  decoding for additive non-{G}aussian channels,'' \emph{IEEE Trans. Inf.
  Theory}, vol.~63, no.~1, pp. 81--92, Jan. 2017.

\bibitem{Metzger2019}
F.~Metzger, T.~Ho{\ss}feld, A.~Bauer, S.~Kounev, and P.~E. Heegaard, ``Modeling
  of aggregated {IoT} traffic and its application to an {IoT} cloud,''
  \emph{Proc. {IEEE}}, vol. 107, no.~4, pp. 679--694, Apr. 2019.

\bibitem{Shalev_ML}
S.~Shalev-Shwartz and S.~Ben-David, \emph{Understanding Machine Learning:
  {F}rom Theory to Algorithms}.\hskip 1em plus 0.5em minus 0.4em\relax
  Cambridge, USA: Cambridge University Press, 2014.

\bibitem{Polyanskiy_ISIT2021_tutorial}
\BIBentryALTinterwordspacing
J.-F. Chamberland, K.~Narayanan, and Y.~Polyanskiy, ``Unsourced multiple access
  {(UMAC)}: Information theory and coding (tutorial),'' in \emph{{IEEE} Int.
  Symp. Inf. Theory {(ISIT)}}, Melbourne, Australia, Jul. 2021. [Online].
  Available:
  \url{https://people.lids.mit.edu/yp/homepage/data/isit2021_umac_tutorial.pdf}
\BIBentrySTDinterwordspacing

\bibitem{Polyanskiy2011minimum}
Y.~Polyanskiy, H.~V. Poor, and S.~Verdu, ``Minimum energy to send $k$ bits
  through the {G}aussian channel with and without feedback,'' \emph{IEEE Trans.
  Inf. Theory}, vol.~57, no.~8, pp. 4880--4902, Aug. 2011.

\bibitem{Ohnishi2021}
Y.~Ohnishi and J.~Honorio, ``Novel change of measure inequalities with
  applications to {PAC-Bayesian} bounds and {M}onte {C}arlo estimation,'' in
  \emph{Proc. Int. Conf. Artificial Intelligence and Statistics
  (AISTATS)}.\hskip 1em plus 0.5em minus 0.4em\relax Virtual Conference: PMLR,
  Apr. 2021, pp. 1711--1719.

\bibitem{DeGroot2012ProbStats}
M.~H. DeGroot and M.~J. Schervish, \emph{Probability and Statistics},
  4th~ed.\hskip 1em plus 0.5em minus 0.4em\relax Boston, MA, USA: Pearson,
  2012.

\bibitem{Kowshik2020fundamental}
S.~S. Kowshik and Y.~Polyanskiy, ``Fundamental limits of many-user {MAC} with
  finite payloads and fading,'' \emph{IEEE Trans. Inf. Theory}, vol.~67, no.~9,
  pp. 5853--5884, Sep. 2021.

\bibitem{Polyanskiy2010}
Y.~{Polyanskiy}, H.~V. {Poor}, and S.~{Verdu}, ``Channel coding rate in the
  finite blocklength regime,'' \emph{IEEE Trans. Inf. Theory}, vol.~56, no.~5,
  pp. 2307--2359, Apr. 2010.

\end{thebibliography}

\begin{IEEEbiographynophoto}{Khac-Hoang Ngo}
	(Member, IEEE) received the B.E. degree (Hons.) in electronics and telecommunications from University of Engineering and Technology, Vietnam National University, Hanoi, Vietnam, in 2014; and the M.Sc. degree (Hons.) and Ph.D. degree in wireless communications from CentraleSupélec, Paris-Saclay University, France, in 2016 and 2020, respectively. His Ph.D. thesis was also realized at Paris Research Center, Huawei Technologies France. Since September 2020, he has been a postdoctoral researcher at Chalmers University of Technology, Sweden. He is also an adjunct lecturer at University of Engineering and Technology, Vietnam National University Hanoi, Vietnam. His research interests include wireless communications and information theory, with an emphasis on massive random access, distributed learning, MIMO, noncoherent communications, coded caching, and network coding. He received the ``Signal, Image \& Vision Ph.D. Thesis Prize'' by Club EEA, GRETSI and GdR-ISIS, France, and the Marie Skłodowska-Curie Actions (MSCA) Individual Fellowship in 2021.
\end{IEEEbiographynophoto}

\begin{IEEEbiographynophoto}{Alejandro Lancho}
    (Member, IEEE) received the B.E., M.Sc., and Ph.D. degrees in electrical engineering from the Universidad Carlos III de Madrid, Spain, in 2013, 2014, and 2019, respectively.
    From 2019-2021, he was a Post-Doctoral Researcher with Chalmers University of Technology, Sweden. Since 2021, he is a Marie Curie Postdoctoral Global Fellow at the Massachusetts Institute of Technology, USA. His research interests include information theory and deep learning for wireless communications. He was among the six finalists for the IEEE Jack Keil Wolf ISIT Student Paper Award at the 2017 IEEE International Symposium on Information Theory.
\end{IEEEbiographynophoto}

\begin{IEEEbiographynophoto}{Giuseppe Durisi}
    (Senior Member, IEEE) received the Laurea (summa cum laude) and Ph.D. degrees from the Politecnico di Torino, Italy, in 2001 and 2006, respectively.
    From 2002 to 2006, he was with the Istituto Superiore Mario Boella, Turin, Italy. From 2006 to 2010, he was a Post-Doctoral Researcher at ETH Zurich, Zürich, Switzerland. In 2010, he joined the Chalmers University of Technology, Gothenburg, Sweden, where he is currently a Professor with
    the Communication Systems Group. His research interests are in the areas of communication and information theory and machine learning. He is the recipient of the 2013 IEEE ComSoc Best Young Researcher Award for the Europe, Middle East, and Africa region, and is coauthor of a paper that won the Student Paper Award at the 2012 International Symposium on Information Theory, and of a paper that won the 2013 IEEE Sweden VT-COM-IT Joint Chapter Best Student Conference Paper Award. From 2011 to 2014, he served as a publications editor for the IEEE TRANSACTIONS ON INFORMATION THEORY. From 2015 to 2021, he served as associate editor for the  IEEE TRANSACTIONS ON COMMUNICATIONS.
\end{IEEEbiographynophoto}

\begin{IEEEbiographynophoto}{Alexandre Graell i Amat}
    (Senior Member, IEEE) received the M.Sc. and Ph.D. degrees in electrical engineering from the Politecnico di Torino, Turin,
    Italy, in 2000 and 2004, respectively, and the M.Sc. degree in telecommunications engineering from the Universitat Politècnica de Catalunya, Barcelona, Catalonia, Spain, in 2001. From 2001 to 2002, he was a Visiting Scholar with the University of California at San Diego, La Jolla, CA,
    USA. From 2002 to 2003, he held a visiting appointment at Universitat Pompeu Fabra, Barcelona, and the Telecommunications Technological Center of Catalonia, Barcelona. From 2001 to 2004, he held a part-time appointment
    at STMicroelectronics Data Storage Division, Milan, Italy, as a Consultant on coding for magnetic recording channels. From 2004 to 2005, he was a Visiting Professor with Universitat Pompeu Fabra. From 2006 to 2010, he was with the Department of Electronics, IMT Atlantique (formerly ENST Bretagne), Brest, France. Since 2019, he has also been an Adjunct Research Scientist with Simula UiB, Bergen, Norway. He is currently a Professor with the Department of Electrical Engineering, Chalmers University of Technology,     Gothenburg, Sweden. His research interests are in the field of coding theory with application to 
    distributed learning and computing, storage, privacy and security, and communications. He received the Marie Skłodowska-Curie Fellowship from the European Commission and the Juan de la Cierva Fellowship from the Spanish Ministry of Education and Science. He received the IEEE Communications Society 2010 Europe, Middle East, and Africa Region Outstanding Young Researcher Award. He was the General Co-Chair of the 7th International Symposium on Turbo Codes and Iterative Information Processing, Sweden, in 2012, and the TPC Co-Chair of the 11th International Symposium on Topics in Coding, Canada, in 2021. He was an Associate Editor of the IEEE COMMUNICATIONS LETTERS from 2011 to 2013. He was an Associate Editor and the Editor-at-Large of the IEEE TRANSACTIONS     ON COMMUNICATIONS from 2011 to 2016 and 2017 to 2020, respectively. He is currently an Area Editor of the IEEE TRANSACTIONS ON  COMMUNICATIONS.
\end{IEEEbiographynophoto}

\end{document}